\documentclass[a4paper,10pt]{article}
\pdfoutput=1

\usepackage{amssymb}
\usepackage{amsmath}
\usepackage{amsfonts}
\usepackage{amsthm}

\usepackage[pdftex,colorlinks=true]{hyperref}
\usepackage{geometry}
\usepackage{xspace}
\usepackage{enumerate}
\usepackage{todonotes}

\usepackage[font={small}]{caption}
\usepackage{graphicx} 	% needed for including graphics e.g. EPS, PS
\usepackage[tight]{subfigure} % For subfigures.
\usepackage{float}

\usepackage{fullpage}
\usepackage[all=normal,bibliography=tight]{savetrees}

\graphicspath{{./}{./Images/}}

\newcommand{\image}[1]{\includegraphics[scale=0.55]{#1}}
\newcommand{\bigimage}[1]{\includegraphics[scale=0.85]{#1}}

\newcommand{\widefigure}[2]{
\begin{figure}[t]
\begin{center}
\makebox[\linewidth][c]{
#1
}
#2
\end{center}
\end{figure}
}

\usepackage{xcolor}
\definecolor{dark-red}{rgb}{0.4,0.15,0.15}
\definecolor{dark-blue}{rgb}{0.15,0.15,0.4}
\definecolor{medium-blue}{rgb}{0,0,0.5}
\definecolor{gray}{rgb}{0.5,0.5,0.5}

\hypersetup{
    colorlinks, linkcolor={dark-red},
    citecolor={dark-blue}, urlcolor={medium-blue}
}

\theoremstyle{plain}
\newtheorem{proposition}{Proposition}
\newtheorem{claim}{Claim}
\newtheorem{lemma}{Lemma}
\newtheorem{corollary}{Corollary}
\newtheorem{theorem}{Theorem}
\newtheorem*{theorem*}{Theorem}
\newtheorem*{lemma*}{Lemma}
\theoremstyle{definition}
\newtheorem{observation}{Observation}
\newtheorem{definition}{Definition}
\newtheorem{reduction}{Reduction}
\newtheorem{template}{Reduction template}

\newenvironment{claimproof}{\begin{proof}\renewcommand{\qedsymbol}{\claimqed}}{\end{proof}\renewcommand{\qedsymbol}{\plainqed}}

\let\plainqed\qedsymbol

\newcommand{\C}{\mathcal{C}}
\newcommand{\Oh}{\mathcal{O}}
\newcommand{\Ohtilde}{\widetilde{\mathcal{O}}}
\newcommand{\Gc}{G^\circ}

\newcommand{\Ptd}{\ensuremath{\Pi}}

\newcommand{\MPtd}{\ensuremath{\Ptd_T}}

\newcommand{\binv}{\ensuremath{\beta^{-1}}\xspace}
\renewcommand{\root}{\ensuremath{\mathrm{\textsc{root}}}\xspace}
\renewcommand{\succ}{\ensuremath{\mathrm{\textsc{succ}}}\xspace}
\newcommand{\pred}{\ensuremath{\mathrm{\textsc{pred}}}\xspace}

\newcommand{\adh}{\ensuremath{\mathop{\mathrm{\textsc{adh}}}\xspace}}

\renewcommand{\top}{\ensuremath{\mathop{\mathrm{\textsc{top}}}\xspace}}
\newcommand{\lca}{\ensuremath{\mathop{\mathrm{\textsc{lca}}}\xspace}}
\newcommand{\cliquebound}{\ensuremath{\omega_{k, |M|}}\xspace}
\newcommand{\containment}{\ensuremath{\mathsf{NP \subseteq coNP/poly}}\xspace}

\newcommand{\opt}{\textsc{opt}\xspace}

\newcommand{\ChVDlong}{\textsc{Chordal Vertex Deletion}\xspace}
\newcommand{\ChVD}{\textsc{ChVD}\xspace}
\newcommand{\AChVDlong}{\textsc{Annotated} \ChVDlong{}\xspace}
\newcommand{\AChVD}{\textsc{A}-\ChVD{}\xspace}

\newcommand{\poly}{\mathbf{\mathrm{poly}}}
\newcommand{\LPsol}{\mathbf{x}^\ast}
\newcommand{\multicut}{\textsc{Multicut}\xspace}

\newcommand{\terms}{\mathcal{T}}
\newcommand{\Gdown}{G^{\downarrow}}
\newcommand{\dist}{\mathrm{dist}}

\title{Approximation and Kernelization for \\ Chordal Vertex Deletion%
\footnote{The majority of this work was done while the authors were participating in the ``Fine-Grained Complexity and Algorithm Design'' program at the Simons Institute for the Theory of Computing in Berkeley. Bart Jansen is supported by the Netherlands Organisation for Scientific Research, Veni grant ``Frontiers in Parameterized Preprocessing'' and Gravitation grant ``Networks''.
Marcin Pilipczuk is supported by Polish National Science Centre grant DEC-2012/05/D/ST6/03214.
}%
}

\author{Bart M.\ P.\ Jansen\thanks{Eindhoven University of Technology, Eindhoven, The Netherlands, \texttt{b.m.p.jansen@tue.nl}.}
\and Marcin Pilipczuk\thanks{University of Warsaw, Warsaw, Poland, \texttt{marcin.pilipczuk@mimuw.edu.pl}.}}

\begin{document}

\maketitle

\begin{abstract}
The \ChVDlong (\ChVD) problem asks to delete a minimum number of vertices from an input graph to obtain a chordal graph.
In this paper we develop a polynomial kernel for \ChVD under the parameterization by the solution size, as well
as $\poly(\opt)$ approximation algorithm.
The first result answers an open problem of Marx from 2006 [WG 2006, LNCS 4271, 37--48].
\end{abstract}

\section{Introduction}
%\paragraph{Background}
\newcommand{\intrographclass}{\mathcal{G}}

Many important combinatorial problems can be phrased in terms of vertex deletions: given an input graph~$G$ and an integer~$k$, is it possible to delete~$k$ vertices from~$G$ to ensure the resulting graph belongs to a graph class~$\intrographclass$? For example, \textsc{Vertex Cover} asks to delete~$k$ vertices to obtain an edgeless graph, while in (undirected) \textsc{Feedback Vertex Set} one deletes vertices to obtain an acyclic graph. By a classic result of Lewis and Yannakakis~\cite{LewisY80}, the \textsc{$\intrographclass$ Vertex Deletion} problem is NP-hard for all nontrivial hereditary graph classes~$\intrographclass$. This motivated intensive research into alternative algorithmic approaches to deal with these problems, such as polynomial-time approximation algorithms and fixed-parameter tractable (exact) algorithms, which test for a solution of size~$k$ in time~$f(k) \cdot n^{\Oh(1)}$ for some arbitrary function~$f$. The notion of kernelization, which formalizes efficient and effective preprocessing, has also been used to cope with the NP-hardness of vertex deletion problems. A \emph{kernelization} (or \emph{kernel}) for a parameterized graph problem is a polynomial-time preprocessing algorithm that reduces any instance~$(G,k)$ to an equivalent instance~$(G',k')$ whose total size is bounded by~$g(k)$ for some function~$g$, which is the \emph{size} of the kernel. While all fixed-parameter tractable problems have a kernel, a central challenge is to identify problems for which the size bound~$g(k)$ can be made polynomial, a so-called \emph{polynomial kernel}; we refer to surveys~\cite{SurveyKratsch14,SurveyLokshtanov11,SurveyLokshtanovMS12} for recent examples.

The paradigm of parameterized complexity fits graph modification problems particularly well. Assume that an input graph $G$ represents results of an experiment or measurement. The experimental setup could result in the fact that~$G$ has certain properties, and therefore belongs to a specific graph class~$\intrographclass$, if the measurements were completely error-free. Due to measurement errors, graph $G$ might not be in $\intrographclass$, but close to a member of $\intrographclass$. Graph modification problems then capture the task of detecting and filtering out the errors. If the number of measurement errors is small, the solution size of the corresponding modification problem is small; an ideal target for parameterized algorithms.

The complexity of \textsc{Vertex Cover} and \textsc{Feedback Vertex Set} is well-understood with respect to all three aforementioned algorithmic paradigms: approximation algorithms, fixed-parameter algorithms, and polynomial kernelization. In this paper we analyze a generalization that has remained much more elusive: deleting vertices to make a graph \emph{chordal}, i.e., to ensure the resulting graph contains no induced cycles of length at least four (\emph{holes}). While \textsc{Vertex Cover} and \textsc{Feedback Vertex Set} have long been known to be fixed-parameter tractable~\cite{Bodlaender93,BussG93}, to have constant-factor approximation algorithms~\cite{BafnaBF99,BeckerG96,NemhauserT75}, and to have kernels of polynomial size~\cite{BodlaenderD10,BurrageEFLMR06,NemhauserT75,Thomasse10}, our understanding of \ChVDlong (\ChVD) is lacking behind greatly. The problem is known to be fixed-parameter tractable by a result of Marx~\cite{Marx10} and the current-best runtime is~$2^{\Oh(k \log k)} \cdot n^{\Oh(1)}$~\cite{CaoM15}. The kernelization and approximation complexity of \ChVD has remained unresolved, however. The question whether \ChVD has a polynomial kernel or not has been open for 10 years~\cite{Marx06d}, and has been posed by several groups of authors~\cite{CaoM15,FominSV13,HeggernesHJKV13} and during the 2013 Workshop on Kernelization~\cite{worker2013-opl}. The approximability of \ChVD is similarly open. To the best of our knowledge, nothing was known prior to this work.

% For which vertex deletion problems that have poly kernels do these properties fail? Vertex deletion to:
% Unit / proper interval graphs (the claw allows you to rule out a lot)
% OCT
% Deletion to Konig graphs? Is that covered by a recent rep-set paper?

\paragraph{Results}

We resolve the kernelization complexity of \ChVD and show that the problem has a kernel of polynomial size. Our main data reduction procedure is summarized by the following theorem.

\newcommand{\chvdkerneltheorem}{There is a polynomial-time algorithm that, given an instance~$(G,k)$ of \ChVDlong and a modulator~$M_0 \subseteq V(G)$ such that~$G - M_0$ is chordal, outputs an equivalent instance~$(G',k')$ with~$\Oh(k^{45} \cdot |M_0|^{29})$ vertices and~$k' \leq k$.}

\begin{theorem} \label{thm:chvd:kernel}
\chvdkerneltheorem
\end{theorem}

To be able to apply this reduction algorithm, a modulator~$M_0$ is needed to reveal the structure of the chordal graph~$G - M_0$. If \ChVD would have a polynomial-time constant-factor approximation algorithm, or even a $\poly(\opt)$-approximation algorithm, one could bootstrap Theorem~\ref{thm:chvd:kernel} with an approximate solution as modulator to obtain a polynomial kernel for \ChVD. Prior to this work, no such approximation algorithm was known. We developed a $\poly(\opt)$-approximation algorithm to complete the kernelization, which we believe to be of independent interest.

\begin{theorem}\label{thm:apx}
There is a polynomial-time algorithm that, given an undirected graph $G$ % with $n$ vertices 
and an integer $k$,
either correctly concludes that $(G,k)$ is a no-instance of \ChVD
or computes a set $X \subseteq V(G)$ of size $\Oh(k^4 \log^2 k)$
such that $G-X$ is chordal.
\end{theorem}

Combining the two theorems, we obtain a kernel with~$\Oh(k^{161} \log^{58} k)$ vertices for \ChVD. To appreciate the conceptual difficulties in obtaining a polynomial kernel for \ChVD, let us point out why some popular kernelization strategies fail. In general, the difficulty of the \textsc{$\intrographclass$ Vertex Deletion} problem is tightly linked to the amount of structure of graphs in~$\intrographclass$. If all graphs in~$\intrographclass$ have bounded treewidth, and~$\intrographclass$ is closed under taking minors, then one can obtain a kernel by repeatedly replacing large \emph{protrusions}~\cite{MetaKernelization09,FominLMS12} in the graph by smaller gadgets. Since edgeless graphs and acyclic graphs are minor-closed and have constant treewidth, this explains the existence of kernels for \textsc{Vertex Cover} and \textsc{Feedback Vertex Set}. This approach does not work for \ChVD, as chordal graphs can have arbitrarily large treewidth (every clique is chordal). Another kernelization strategy for vertex deletion problems into hereditary classes~$\intrographclass$ is to model the problem as an instance of \textsc{Hitting Set} and apply a kernelization strategy based on the Sunflower lemma (cf.~\cite[\S 9.1]{FlumG06}). This yields polynomial kernels when there is a characterization of~$\intrographclass$ in terms of a finite set of forbidden induced subgraphs, such as co-graphs. However, the number of forbidden induced subgraphs for chordal graphs is infinite as this set contains all induced cycles of length at least four, making this approach infeasible as well. 

\paragraph{Techniques} Several recent kernelization results rely on matroid theory~\cite{KratschW12,KratschW14} (resulting in randomized kernels), or on \emph{finite integer index}, or the Graph-Minor Theorem~\cite{MetaKernelization09,FominLMS12,GajarskyHOORRVS13,KimLPRRSS13} (resulting in proofs that kernelization algorithms exist without showing how they should be constructed). In contrast, our kernelization algorithm is deterministic and fully explicit; it is based on new graph-theoretical insights into the structure of chordal graphs. 

A key ingredient in our kernelization is an algorithmic version of Erd\H{o}s-P\'{o}sa type~\cite{ErdosP65} covering/packing duality for holes in chordal graphs, expressed in terms of \emph{flowers}. A $v$-flower of order~$k$ in graph~$G$ is a set~$\C$ of~$k$ holes~$\C = \{C_1, \ldots, C_k\}$ in~$G$, such that~$V(C_i) \cap V(C_j) = \{v\}$ for all~$i \neq j$. 

\newcommand{\erdosposa}{There is a polynomial-time algorithm that, given a graph~$G$ and a vertex~$v$ such that~$G - v$ is chordal, outputs a $v$-flower~$C_1, \ldots, C_k$ and a set~$S \subseteq V(G) \setminus \{v\}$ of size at most~$12k$ such that~$G-S$ is chordal.}
\begin{lemma} \label{lemma:algorithm:packing:vs:covering}
\erdosposa
\end{lemma}

So either there is a large packing of holes intersecting only in~$v$, or all the holes can be intersected by a small set of vertices. We prove Lemma~\ref{lemma:algorithm:packing:vs:covering} by presenting a local search algorithm that builds a flower, followed by a greedy algorithm to construct a hitting set when the local search cannot make further progress. This procedure forms the first step in a series of arguments that \emph{tidy the modulator} (in the sense of~\cite{BevernMN12}) and narrow down the structure of the graph, which enables us to identify irrelevant parts of the graph which can be reduced. Section~\ref{s:intuition} describes the intuition behind the kernelization in more detail.

When it comes to the approximation algorithm, we combine linear-programming rounding techniques with structural insights into the behavior of nearly-chordal graphs. Since chordal graphs have balanced clique separators, if~$(G,k)$ is a yes-instance of \ChVD then there is a balanced separator in~$G$ consisting of a clique plus~$k$ vertices. We use an approximation algorithm for minimum-weight vertex separators of Feige et al.~\cite{FeigeHL08} together with an enumeration of maximal cliques  to decompose the graph. This decomposition step reduces the problem of finding an approximate solution to the setting in which we have an input graph~$G$ and a clique~$X$ such that~$G - X$ is chordal, and the aim is to find a small chordal deletion set in~$G$ (smaller than~$|X|$). After a second decomposition step, an analysis of the structure of holes allows the problem to be phrased in terms of multicut on a directed graph, to which we apply a result of Gupta~\cite{Gupta03}.

\paragraph{Motivation}
%Finally, we would like to mention that
%in the recent years, 
The parameterized complexity of
graph modification problems to chordal graphs and related
graph classes gained significant interest in recent years.
Several factors explain the popularity of such problems. 
First, chordal graphs and related graph classes, such as interval graphs, appear naturally in applications in solving linear equations~\cite{Buneman74}, in scheduling~\cite{Bar-NoyBFNS01}, and in bioinformatics~\cite{Rose72,ZhangSFCWKB94}.
%\todo{This motivates why people look at interval graphs, but is not really compelling for chordal graphs. However, I see the importance of having a practical motivation.}
% Marcin: I found two references for chordal graphs. They mostly concern completion to chordal graphs, but I think they are good enough.
 Second, from a more theoretical perspective, these graph classes are closely related to widely used graph parameters, such as treewidth: a tree decomposition can be seen as a completion to a chordal graph, and similar relations hold between interval graphs and pathwidth, proper interval graphs and bandwidth, or trivially perfect graphs and treedepth (cf.~\cite{DrangeFPV15}).
Third, to this day, we do not know of a vertex-deletion problem that is FPT but has no polynomial kernel. (We do know of some ones that are not FPT, such as \textsc{Perfect Deletion}~\cite{HeggernesHJKV13}.) It is tempting to conjecture that all FPT vertex deletion problems have polynomial kernels. \ChVD is a candidate for refuting this conjecture, and resolving its complexity status gives us more insight into the landscape of kernelization complexity of vertex deletion problems.
Fourth, as witnessed by numerous works in the past years~\cite{BliznetsFPP15,BliznetsFPP16,Cao16,CaoM14b,CaoM15,DrangeFPV15,FominSV13}, most graph modification problems to chordal graphs and related graph classes turn out to be tractable but very challenging;
the corresponding algorithmic results build upon deep graph-theoretical insights into the studied graph class, making this study particularly interesting from the point of view of graph theory.

One motivation for studying the particular case of vertex deletion to a chordal graph comes from the following fact: given a graph~$G$ and a modulator~$M$ such that~$G-M$ is chordal, one can find a maximum independent set, maximum clique, and minimum vertex cover of~$G$ in time~$2^{|M|} \cdot n^{\Oh(1)}$ (cf.~\cite[\S 6]{Cai03a}). Hence an efficient algorithm for chordal vertex deletion allows these fundamental optimization problems to be solved efficiently on almost-chordal graphs.

\paragraph{Related work} Kernels are known for a wide variety of vertex-deletion problems~\cite{FominLMS12,FominSV13,GiannopoulouJLS15,KratschW14,NemhauserT75,Thomasse10}. The recent kernels for \textsc{Odd Cycle Transversal}~\cite{KratschW14} and \textsc{Proper Interval Vertex Deletion}~\cite{FominSV13} are noteworthy since, like chordal graphs, bipartite graphs and proper interval graphs do not have bounded treewidth and do not have a characterization by a finite set of forbidden induced subgraphs. There are linear kernels for \ChVD on planar graphs~\cite{MetaKernelization09} and $H$-topological-minor-free graphs~\cite{KimLPRRSS13}. Concerning \emph{structural} parameterizations, the \ChVD problem is known to have polynomial kernels when parameterized by the size of a vertex cover~\cite{FominJP12}, when parameterized by the size of a set of candidate vertices from which the solution is required to be selected~\cite{HeggernesHJKV13}, and on bounded-expansion graphs when parameterized by a constant-treedepth modulator~\cite{GajarskyHOORRVS13}. Approximation algorithms have been developed for (weighted) \textsc{Feedback Vertex Set}~\cite{BafnaBF99,BeckerG96}, \textsc{Odd Cycle Transversal}~\cite{GargVY94}, and its edge-deletion variant~\cite{AgarwalCMM05}.

\paragraph{Organization}
We present preliminaries on graphs in Section~\ref{s:preliminaries}. Section~\ref{s:intuition} contains an informal overview of the data reduction procedure. Section~\ref{s:erdosposa} is devoted to the proof of Lemma~\ref{lemma:algorithm:packing:vs:covering}. We introduce an annotated version of \ChVD in Section~\ref{s:annotated}, which captures some of the structure we can obtain in the instance by appropriate use of Lemma~\ref{lemma:algorithm:packing:vs:covering}. Sections~\ref{s:toughness}--\ref{s:single-component} present reduction rules for instances of the annotated problem. The approximation algorithm is given in Section~\ref{s:approximation}, which allows us to give the final kernel in Section~\ref{s:kernelization}.

\section{Preliminaries} \label{s:preliminaries}
The set~$\{1,\ldots,n\}$ is abbreviated as~$[n]$. A set~$S$ \emph{avoids} a set~$U$ if~$S \cap U = \emptyset$.
For a set $U$, by $\binom{U}{2}$ we denote the family of all unordered pairs of distinct vertices of $U$. 
We often use $xy$ instead of $\{x,y\}$ for an element of $\binom{U}{2}$.
A graph~$G$ consists of a vertex set~$V(G)$ and edge set~$E(G) \subseteq \binom{V(G)}{2}$. For a vertex~$v \in V(G)$ its open neighborhood is~$N_G(v) := \{u \in V(G) \mid uv \in E(G)\}$ and its closed neighborhood is~$N_G[v] := N_G(v) \cup \{v\}$. The open neighborhood of a set~$S \subseteq V(G)$ is~$N_G(S) := \bigcup _{v \in S} N_G(v) \setminus S$, and the closed neighborhood is~$N_G[S] := \bigcup _{v \in S} N_G[v]$. The graph obtained from~$G$ by deleting all vertices in the set~$S$ and their incident edges is denoted~$G - S$; for a single vertex~$v \in V(G)$ we use~$G - v$ as a shorthand for~$G - \{v\}$.

An $st$-\emph{walk} in a graph~$G$ is a sequence of vertices~$s=v_1, v_2, \ldots, v_\ell=t$ such that~$v_i = v_{i+1}$ or~$v_i v_{i+1} \in E(G)$ for~$i \in [\ell - 1]$. An $st$-\emph{path} is an $st$-walk on which all vertices are distinct. For a walk~$W$ or path~$P$ we use~$V(W)$ and~$V(P)$ to denote the vertices of~$W$ and~$P$, respectively. For an $st$-path~$P$, the vertices~$s$ and~$t$ are its \emph{endpoints} and~$V(P) \setminus \{s,t\}$ are its \emph{interior vertices}. 

\begin{observation} \label{observation:walk:implies:inducedpath}
For any~$st$-walk~$W$ in a graph~$G$, any shortest $st$-path in~$G[V(W)]$ is an induced~$st$-path~$P$ in~$G$ with~$V(P) \subseteq V(W)$.
\end{observation}

A \emph{hole} is an induced (i.e., chordless) cycle of length at least four. A graph is \emph{chordal} if it does not contain any holes; we say that such graphs are hole-free. The \emph{independence number} of a graph~$G$, denoted~$\alpha(G)$, is the maximum cardinality of a set of vertices in~$G$ in which no pair is connected by an edge. Similarly, the clique number~$\omega(G)$ is the maximum cardinality of a set of vertices that are pairwise all adjacent.

\begin{proposition} \label{proposition:independencenr:vs:cliquenr}
Any chordal graph~$G$ satisfies~$\alpha(G) \geq |V(G)| / \omega(G)$.
\end{proposition}
\begin{proof}
As any chordal graph is perfect~\cite{BrandstadtLS99}, the chromatic number of~$G$ equals~$\omega(G)$. A proper coloring with~$\omega(G)$ colors is a partition into~$\omega(G)$ independent sets. At least one color class has size~$|V(G)| / \omega(G)$, guaranteeing an independent set of that size.
\end{proof}

\paragraph{Tree decompositions}

A \emph{tree decomposition} of a graph~$G$ is a pair~$(T,\beta)$, where~$T$ is a rooted tree and~$\beta \colon V(T) \to 2^{V(G)}$ assigns to every node~$p \in V(T)$ a subset of the vertices of~$G$, called a \emph{bag}, such that the following holds:
\begin{enumerate}
	\item $\bigcup _{p \in V(T)} \beta(p) = V(G)$, and
	\item for each edge~$uv \in E(G)$ there is a node~$p \in V(T)$ such that~$uv \subseteq \beta(p)$, and
	\item for each~$v \in V(G)$, the nodes of~$T$ whose bag contains~$v$ induce a connected subtree of~$T$.
\end{enumerate}
For a vertex~$v \in V(G)$ we use~$\binv(v)$ to denote the set of nodes of~$T$ whose bag contains~$v$. The third property ensures that~$\binv(v)$ is a subtree of~$T$ for all~$v \in V(G)$. We use~$T_p$ to denote the subtree of~$T$ rooted at~$p \in V(T)$. The \emph{adhesion} of an edge~$e = pq$ of the decomposition tree is~$\adh(e) := \beta(p) \cap \beta(q)$. For a rooted tree decomposition~$(T,\beta)$ of~$G$ and a vertex~$v \in V(G)$, define~$\top(v)$ to be the node~$p \in V(T)$ whose bag~$\beta(p)$ contains~$v$ and whose distance in~$T$ to the root node is minimum. The third property of of tree decompositions ensures that this is well defined. For two nodes~$p,q \in V(T)$, we use~$\lca(p,q)$ to denote the lowest common ancestor of~$p$ and~$q$ in~$T$.

For a graph~$G$ and two vertex sets~$X$ and~$Y$, we use~$d_G(X,Y)$ to denote the number of edges on a shortest path from a vertex in~$X$ to a vertex in~$Y$. When~$(T,\beta)$ is a tree decomposition of~$G$ and~$s,t \in V(G)$, this means that~$d_T(\binv(s), \binv(t))$ is the distance between the subtree where~$s$ is represented and the subtree where~$t$ is represented. A path~$P$ in~$G$ \emph{connects} vertex sets~$X$ and~$Y$ if one endpoint of~$P$ belongs to~$X$ and the other endpoint belongs to~$Y$. It is a \emph{minimal} path connecting~$X$ and~$Y$ when no interior vertex belongs to~$X$ or~$Y$. Observe that for disjoint connected vertex subsets~$X$ and~$Y$ of a tree~$T$, the minimal path connecting~$X$ and~$Y$ is unique. This implies that the minimal path connecting~$\binv(s)$ and~$\binv(t)$ is unique when~$s$ and $t$ do not appear in the same bag of a tree decomposition~$(T,\beta)$ of~$G$. We use~$\MPtd(s,t)$ to refer to a minimal path connecting~$\binv(s)$ and~$\binv(t)$ in~$T$. We also use~$\MPtd(s, \root)$ to refer to the path from~$\top(s)$ to the root of~$T$, which is the minimal path connecting~$\binv(s)$ to the root node.

\begin{observation} \label{observation:stpaths:use:adhesion}
Let~$(T,\beta)$ be a tree decomposition of graph~$G$. If~$s, t$ are distinct vertices of~$G$ and~$\Ptd$ is a minimal path in~$T$ connecting~$\binv(s)$ to~$\binv(t)$, then for every edge~$e$ on~$\Ptd$ all $st$-paths in~$G$ contain a vertex of~$\adh(e)$.
\end{observation}

\begin{observation} \label{observation:connectedsets:onesubtree}
Let~$(T,\beta)$ be a tree decomposition of graph~$G$ and let~$p \in V(T)$. If~$S \subseteq V(G)$ such that~$G[S \setminus \beta(p)]$ is connected, then there is a single tree~$T'$ in the forest~$T - p$ such that all bags containing a vertex of~$S \setminus \beta(p)$, are contained in~$T'$.
\end{observation}

\paragraph{Clique trees}

A \emph{clique tree} is a tree decomposition~$(T,\beta)$ such that each bag~$\beta(p)$ is a maximal clique in~$G$. It is known that a graph is chordal if and only if it has a clique tree~\cite[Theorem 1.2.3]{BrandstadtLS99}. A clique tree of a chordal graph can be found in linear time~\cite[Theorem 1.2.4]{BrandstadtLS99}. The definition implies that if~$(T,\beta)$ is a clique tree of of~$G$ and~$S \subseteq V(G)$ is a clique in~$G$, then~$S \subseteq \beta(p)$ for some node~$p \in V(T)$.

Note that in a chordal graph $G$ with a clique tree $(T,\beta)$, two vertices $s,t$ appear in the same bag if and only if they are adjacent. Thus, in a chordal graph $G$ with a clique tree $(T,\beta)$, there is a unique minimal path connecting $\binv(s)$ and $\binv(t)$ for every nonadjacent pair of vertices $s,t$.

In our proofs we often work with a clique tree~$(T,\beta)$ of a chordal graph~$G$, and we will need to refer to both paths in~$G$ and paths in~$T$. To maximize readability, we use the convention that identifiers such as~$P$ and~$P'$ are used for paths in~$G$, while identifiers such as~$\Ptd$ and~$\Ptd'$ are used for paths in the decomposition tree~$T$.

\begin{proposition} \label{proposition:path:from:adhesions}
Let~$(T,\beta)$ be a clique tree of a graph~$G$, let~$U \subseteq V(G)$, and let~$s$ and~$t$ be distinct vertices in~$V(G) \setminus U$. If there is a path~$\Ptd$ in~$T$ connecting~$\binv(s)$ to~$\binv(t)$ such that~$\adh(e) \setminus U \neq \emptyset$ for each edge~$e$ on~$\Ptd$, then there is an induced $st$-path in~$G - U$.
\end{proposition}
\begin{proof}
If~$st \in E(G)$ the claim is trivial. Otherwise, the trees~$\binv(s)$ and~$\binv(t)$ are vertex-disjoint since each bag of the clique tree forms a clique in~$G$. Consider a path~$\Ptd$ as described, consisting of edges~$e_1, \ldots, e_\ell$ with~$\ell \geq 1$. For each edge~$e_i$ with~$i \in [\ell]$ pick a vertex~$v_i \in \adh(e_i) \setminus U$ and note that, by definition of adhesion,~$v_i$ occurs in the bags of both endpoints of~$e_i$. As~$v_1$ is in a common bag with~$s$ and each bag is a clique in~$G$, we know~$sv_1 \in E(G)$. Similarly we have~$v_\ell t \in E(G)$. Finally, edges~$e_i$ and~$e_{i+1}$ share an endpoint whose bag contains both~$v_i$ and~$v_{i+1}$ for~$i < \ell$, implying that~$s, v_1, \ldots, v_\ell, t$ is an $st$-walk in~$G$ that avoids~$U$. By Observation~\ref{observation:walk:implies:inducedpath}, this yields an induced~$st$-path in~$G$ that avoids~$U$.
\end{proof}

\begin{proposition} \label{proposition:induced:paths:in:clique:tree}
Let~$(T,\beta)$ be a clique tree of a graph~$G$, let~$P$ be an induced $st$-path in~$G$ with~$|V(P)| \geq 5$ for~$s,t \in V(G)$, and let~$T^*$ be the tree in the forest~$T - (\binv(s) \cup \binv(t))$ that contains the internal nodes of~$\MPtd(s,t)$. For any vertex~$u$ of~$P$ that is not among the first two or last two vertices on~$P$, the subtree~$\binv(u)$ is contained entirely in~$T^*$.
\end{proposition}
\begin{proof}
Assume the preconditions hold. Since~$u$ is not among the first two or last two vertices on the induced $st$-path~$P$, vertex~$u$ is not adjacent to~$s$ or~$t$ and therefore~$\binv(u)$ is vertex-disjoint with~$\binv(s)$ and~$\binv(t)$. It follows that~$\binv(u)$ is contained entirely within one tree~$T'$ of the forest~$T - (\binv(s) \cup \binv(t))$. Assume for a contradiction that~$T' \neq T^*$. Then either~$\MPtd(u,s)$ goes through~$\binv(t)$, or~$\MPtd(u,t)$ goes through~$\binv(s)$. Assume without loss of generality (by symmetry) that~$\MPtd(u,s)$ goes through~$\binv(t)$.

The subpath~$P'$ of~$P$ from~$s$ to~$u$ is induced. None of the vertices on~$P'$ are adjacent to~$t$, since~$P$ is induced and~$u$ itself is not adjacent to~$s$ or~$t$. Let~$e$ be an edge on~$\MPtd(u,s)$ that has an endpoint~$p$ in~$\binv(t)$. By Observation~\ref{observation:stpaths:use:adhesion}, the $su$-path~$P'$ contains a vertex~$w$ of~$\adh(e) \subseteq \beta(p)$, with~$p \in \binv(t)$. This shows that~$w$ is in a common bag with~$t$. Since each bag of~$T$ forms a clique, vertex~$w$ on~$P'$ is adjacent to~$t$; a contradiction.
\end{proof}

\begin{proposition} \label{proposition:inducedpath:on:minimalpath}
Let~$(T,\beta)$ be a clique tree of graph~$G$ and let~$P$ be an induced $st$-path in~$G$. For every internal vertex~$u$ of~$P$, the tree~$\binv(u)$ contains a node of~$\MPtd(s,t)$.
\end{proposition}
\begin{proof}
Assume for a contradiction that there is an internal vertex~$u$ of~$P$ for which~$\binv(u)$ contains no node of~$\MPtd(s,t)$. Orient~$P$ from~$s$ to~$t$ and consider the first vertex~$u_1$ for which this holds. Let~$u_0$ be the predecessor of~$u_1$ on~$P$. Let~$u_2$ be the first vertex after~$u_1$ on~$P$ for which~$\binv(u_2)$ contains a node of~$\MPtd(s,t)$, which is well-defined since~$t$ is such a vertex.

Let~$T'$ be the tree in~$T - V(\MPtd(s,t))$ containing~$\binv(u_1)$, which is well-defined since~$u_1$ does not occur in a bag of~$\MPtd(s,t)$ and the occurrences of~$u_1$ form a connected subtree. Let~$q$ be the unique node of~$\MPtd(s,t)$ that has a neighbor in~$T'$. Since~$u_0$ occurs in a common bag with~$u_1$ and also occurs in a bag on~$\MPtd(s,t)$, by the connectivity property of tree decompositions we have~$u_0 \in \beta(q)$. Since~$u_2$ is the first vertex after~$u_1$ that occurs in a bag of~$\MPtd(s,t)$, all interior vertices~$u$ of the subpath of~$P$ from~$u_0$ to~$u_2$ occur only in bags of~$T'$ by Observation~\ref{observation:connectedsets:onesubtree}. Since~$u_2$ occurs in a common bag with the predecessor of~$u_2$ on~$P$, it follows that~$u_2$ occurs both in a bag of~$T'$ and in a bag of~$\MPtd(s,t)$, and consequently~$u_2 \in \beta(q)$. Hence~$u_0$ and~$u_2$ occur in a common clique, implying that~$u_0$ and~$u_2$ are connected in~$G$. But this edge is a chord on~$P$ since~$u_0$ and~$u_2$ are not consecutive; a contradiction to the assumption that~$P$ is induced.
\end{proof}

\section{Informal description of the kernelization} \label{s:intuition}
% The obstructions to chordal graphs are all holes, which can be arbitrarily long. Standard sunflower-type arguments therefore do not work.  
In this section we give the high-level ideas behind the kernelization algorithm. Given an instance~$(G,k)$ of \ChVD, we run the approximation algorithm of Theorem~\ref{thm:apx}. If it concludes that~$(G,k)$ is a no-instance, then we output a constant-size no-instance as the result of the kernelization. Otherwise, we obtain an approximate solution~$M_0 \subseteq V(G)$ of size~$\poly(k)$ such that~$G - M_0$ is chordal. The reduction process proceeds in three phases.

\begin{enumerate}
	\item We find a vertex set~$Q \supseteq M_0$ of size~$\poly(k, |M_0|)$ such that all the connected components~$A$ of~$G - Q$ (which are chordal graphs since~$Q$ contains~$M_0$) are simple in the following way: all vertices in~$A$ have the same neighborhood in~$M_0$, and the neighborhood of~$A$ in the rest of the graph (i.e., $Q \setminus M_0$) consists of two cliques of size~$\poly(k, |M_0|)$. \label{reduction:findseparator}
	\item We introduce a reduction rule that bounds the number of connected components of~$G - Q$ in terms of~$|Q|$.\label{reduction:reducenumbercomp}
	\item We introduce a reduction rule that shrinks the size of each component~$A$ of~$G - Q$ to polynomial in the size of the two boundary cliques.\label{reduction:shrinkcomponents}
\end{enumerate}

This process reduces the size of the graph to polynomial in~$k$, when starting with a modulator~$M_0$ of size~$\poly(k)$. 

\paragraph{Reduction phase \ref{reduction:findseparator}.} To be able to find the desired separator~$Q$, several intermediate steps are needed. The flow of ideas is as follows. We prove Lemma~\ref{lemma:algorithm:packing:vs:covering}: if~$G$ is a graph containing a distinguished vertex~$v$ such that~$G - v$ is chordal, then either there are~$k+1$ holes in~$G$ that are vertex-disjoint except for intersecting at~$v$, or there is a set~$S_v \subseteq G$ of~$12k$ vertices such that~$G - S$ is chordal. The proof is algorithmic and yields a polynomial-time algorithm that computes the set of~$k+1$ holes or the set~$S$, depending on the outcome. Given the input graph~$G$ and the modulator~$M_0$, we apply this algorithm to the graph~$G_v := G - (M_0 \setminus \{v\})$ for each~$v \in M$; note that~$G_v - v$ is chordal. When~$G_v$ contains~$k+1$ holes that intersect only in~$v$, we know that any size-$k$ solution to the original instance must contain~$v$. Hence we can delete such vertices and decrease the budget~$k$ by one. After exhaustively applying this rule, we can collect the hitting sets~$S_v$ for the remaining vertices and add them to~$M_0$, obtaining a larger modulator~$M$. It yields the following useful structure: for every~$v \in M$, the graph~$G - (M \setminus \{v\})$ is chordal.
This property is exactly the \emph{tidyness of the modulator}
as defined by Bevern et al~\cite{BevernMN12}.

For a vertex $v \in M$, we define $G(\neg v) = G-(M \cup N(v))$, that is, the subgraph of $G-M$ induced by the nonneighbors of $v$.
By the previous Erd\H{o}s-P\'{o}sa type step, a boundary of a component of $G(\neg v)$, i.e., the set $N_{G-M}(C)$ for a connected component $C$ of $G(\neg v)$,
is a clique. To achieve the desired set $Q$, we first reduce the graph in two ways.
\begin{itemize}
\item We shrink large cliques of $G-M$, bounding the size of a clique in $G-M$ polynomially in $k$ and $|M|$ (Section~\ref{s:cliques}).
\item We reduce the number of connected components of $G(\neg v)$ for every $v \in M$ (Section~\ref{s:components}).
\end{itemize}
We start with $Q$ being all boundaries of connected components of $G(\neg v)$ for all $v \in M$.
To ensure that the boundary of~$A$ consists of two cliques, we use lowest-common-ancestor (LCA) closure of marked bags in a clique tree, similarly as in the case of the protrusion technique~\cite{MetaKernelization09}. This yields a separator~$Q$ with the desired properties.

\begin{figure}[t]
\begin{center}
\subfigure[Graph~$G$, modulator~$M$.]{\label{fig:modulator1}
\image{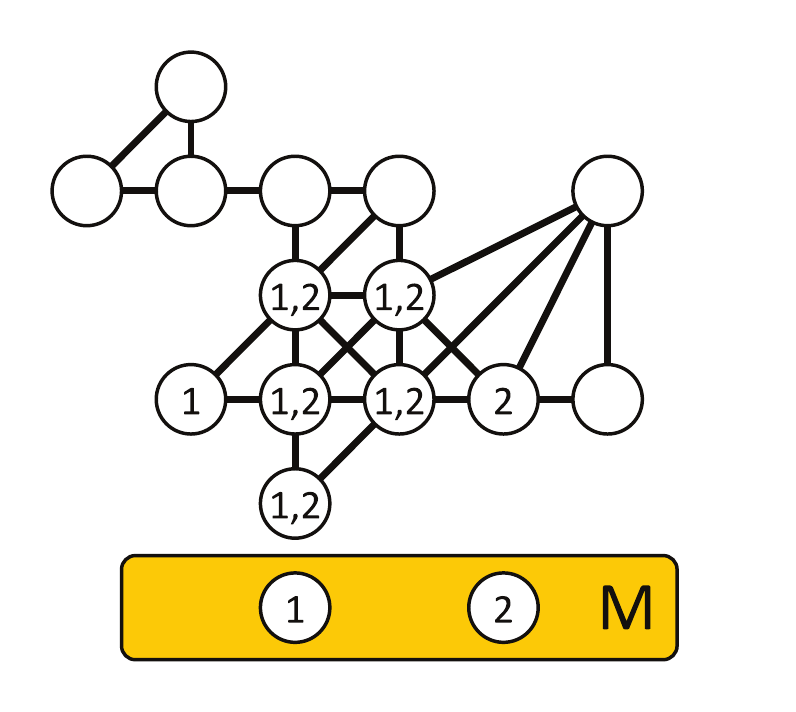}
}
\subfigure[Graph~$G(\neg 2)$.]{\label{fig:modulator2}
\image{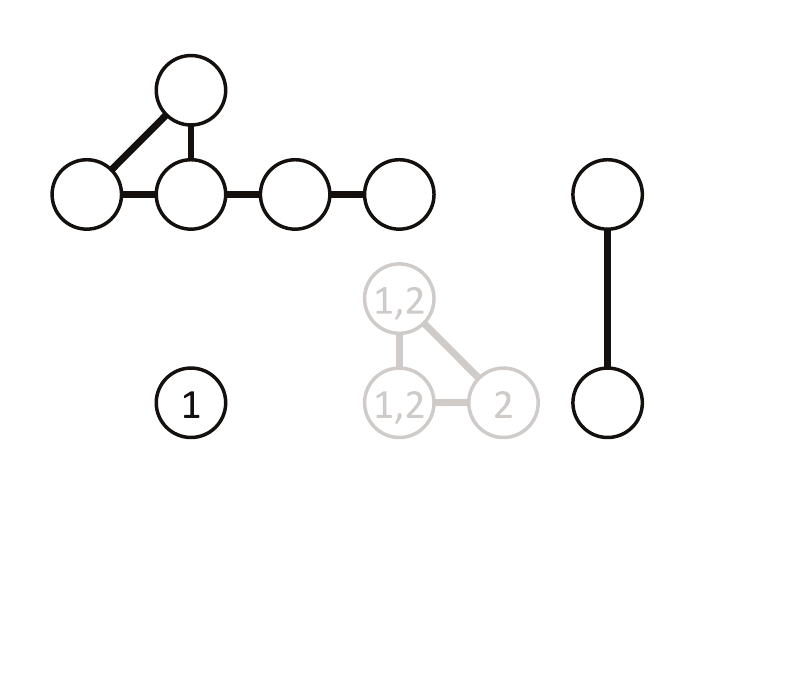}
}
\caption{\ref{fig:modulator1} A graph~$G$ with modulator~$M$ such that~$G - (M \setminus \{v\})$ is chordal for each~$v \in M$. To prevent clutter, the neighborhoods of the vertices~$\{1,2\} \in M$ are represented by drawing a~$1$ (resp.~$2$) in each circle representing a neighbor of~$1$ (resp.~$2$). \ref{fig:modulator2} The graph~$G(\neg 2)$ is drawn in black. The three gray vertices form a clique, which is the neighborhood of the rightmost connected component of~$G(\neg 2)$.}
\end{center}
\end{figure}

\paragraph{Reduction phase \ref{reduction:reducenumbercomp}.} The second phase is easy and presented in Section~\ref{s:toughness}. Since all holes pass through~$M \subseteq Q$, a component~$A$ of~$G - Q$ can only be used on a hole to connect two vertices in~$Q$ by an induced path; it turns out that two different types of connections are relevant. For each of the~$|Q|^2$ pairs of vertices~$x,y$ of~$Q$, we mark~$k+2$ components of each type that provide a connection between~$x$ and~$y$ (or fewer, if the connection type is realized by fewer than~$k+2$ components). Afterwards, any component that is not marked can be safely discarded from the problem.

\paragraph{Reduction phase \ref{reduction:shrinkcomponents}.} The third reduction phase is technical, and described in Section~\ref{s:single-component}. Since a component~$A$ of~$G - Q$ is chordal, has two cliques as its boundary, and all its vertices have the same neighborhood in~$M$, the only way~$A$ can be used to make a hole in the graph is to provide an induced path between its two boundary cliques. To break the holes in the graph, it might therefore be beneficial to have a separator between the two boundary cliques in a solution. By analyzing the local structure, and exploiting the fact that an earlier reduction rule has reduced the sizes of all cliques in~$G - M$ to~$\poly(k, |M|)$, we can pinpoint a small number of relevant cliques in~$A$ such that there is always an optimal solution that does not delete vertices of~$A$ that fall outside the relevant cliques. Afterwards we can shrink~$A$ by a ``bypassing'' operation on all the vertices that are not needed in an optimal solution.

% Q: Relate the components of G - Z to protrusions? They are obtained in a similar way (LCA closure on marked bags of a tree decomposition). The step of augmenting a modulator to get more structure is also used in~\cite{FominLMS12} and~\cite{GiannopoulouJLS15}. 

\section{\texorpdfstring{Erd\H{o}s-P\'{o}sa}{Erdos-Posa} for almost chordal graphs}\label{s:erdosposa}
The main goal of this section is to prove Lemma~\ref{lemma:algorithm:packing:vs:covering}. However, we first present some additional terminology and a subroutine that will be useful during the proof.

We use~$V(\C)$ to denote the union of the vertices used on the holes in a flower~$\C$. Observe that for any $v$-flower~$\C = \{C_1, \ldots, C_k\}$ of order~$k$, the structure~$\C-v$ consists of~$k$ pairwise vertex-disjoint chordless paths~$P_1, \ldots, P_k$ that connect two nonadjacent neighbors of~$v$, whose internal vertices belong to~$V(G) \setminus N_G[v]$. Throughout our proofs we will switch between the interpretation of a $v$-flower as a set of intersecting holes in~$G$, and as a collection of induced paths in~$G - v$ between nonadjacent neighbors of~$v$.

\begin{proposition} \label{proposition:twoflower}
There is a polynomial-time algorithm that, given a graph~$G$ and a vertex~$v$ such that~$G - v$ is chordal, computes a $v$-flower of order two or correctly determines that no such flower exists.
\end{proposition}
\begin{proof}
Observe that~$\C = \{C_1, C_2\}$ is a $v$-flower of order two if and only if~$P_1 := C_1 - v$ and~$P_2 := C_2 - v$ are pairwise vertex-disjoint induced paths whose endpoints are nonadjacent neighbors of~$v$, and whose internal vertices avoid~$N_G[v]$. To search for a flower of order two, we try all~$\Oh(n^4)$ tuples~$(s_1,t_1,s_2,t_2) \in (N_G(v))^4$ for the endpoints of the paths~$P_1$ and~$P_2$ in the set~$N_G(v)$. For all tuples consisting of distinct vertices for which~$s_1 t_1 \not \in E(G)$ and~$s_2 t_2 \not \in E(G)$, we formulate an instance of the \textsc{Disjoint Paths} problem in~$G-(N_G[v] \setminus \{s_1, t_1, s_2, t_2\})$ in which we want to connect~$s_i$ to~$t_i$ for~$i \in \{1,2\}$ by pairwise vertex-disjoint paths. Any FPT-algorithm for the general \textsc{$k$-Disjoint Paths} problem suffices for this purpose as the parameter is constant. We may also use the linear-time algorithm for \textsc{$k$-Disjoint Paths in Chordal Graphs} due to Kammer and Tholey~\cite{KammerT09} since~$G - v$ is chordal. If there are vertex-disjoint~$s_i t_i$-paths~$P_1$ and~$P_2$, then we can shortcut them to \emph{induced} paths, which gives a flower of order two. Conversely, if there is a flower of order two, then we will find such paths in the iteration where we guess the tuple containing their endpoints. Hence if the algorithm loops through all tuples without finding a solution to the \textsc{Disjoint Paths} problem, we may conclude that there is no flower of order two and terminate.
\end{proof}

We are now ready to prove Lemma~\ref{lemma:algorithm:packing:vs:covering}, which we re-state for completeness.

\begin{lemma*}
\erdosposa
\end{lemma*}
\begin{proof}
We first present an algorithm to compute a $v$-flower~$\C$. Then we show how to compute a hitting set based on the structure of the flower, and argue that the constructed set is not too large. To initialize the procedure, we compute a clique tree~$(T,\beta)$ of~$G - v$ and root it at an arbitrary node.

\paragraph{Computing a flower.} The algorithm to compute the flower is a local search procedure that maintains a $v$-flower~$\C$, initially empty, and iteratively improves this structure according to three rules that are described below.

\begin{enumerate}[(I)]
	\item If~$G - (V(\C) \setminus \{v\})$ contains a hole~$C$, then add~$C$ to~$\C$. Equivalently, if there are two distinct nonadjacent vertices~$x,y \in N_G(v) \setminus V(\C)$ such that graph~$G - (V(\C) \cup N_G[v] \setminus \{x,y\})$ contains an induced~$xy$-path~$P$, then add the hole~$\{v\} \cup V(P)$ to~$\C$.\label{improve:addpath}
	\item If the current flower~$\C$ contains a hole~$C_i$ such that the graph~$G - (V(\C) \setminus V(C_i))$ contains a $v$-flower~$\C' = \{C'_i, C''_i\}$ of order two, then remove the hole~$C_i$ from flower~$\C$ and add the holes~$C'_i$ and~$C''_i$ instead.\label{improve:addtwoflower}
	\item If the current flower~$\C$ contains a hole~$C_i$ such that the path~$P_i := C_i - v$ has endpoints~$s$ and~$t$, and there exists a~$t' \in N_G(v) \setminus (V(\C) \cup N_G(s))$ such that:
	\begin{itemize}
		\item $d_T(\binv(s), \binv(t)) > d_T(\binv(s), \binv(t'))$, and
		\item the graph~$(G - (N_G[v] \setminus \{s,t'\})) - (V(\C) \setminus V(P_i))$ contains an induced~$st'$-path~$P'$, i.e., there is an induced~$st'$ path~$P'$ in~$G$ whose interior vertices avoid the closed neighborhood of~$v$ and avoid all vertices used on other holes of the flower,
	\end{itemize}
	then replace~$C_i$ in the flower by the hole consisting of~$\{v\} \cup V(P')$.\label{improve:shorten}
\end{enumerate}

\begin{figure}[t]
\begin{center}
\subfigure[Graph~$G$.]{\label{fig:localsearch1}
\image{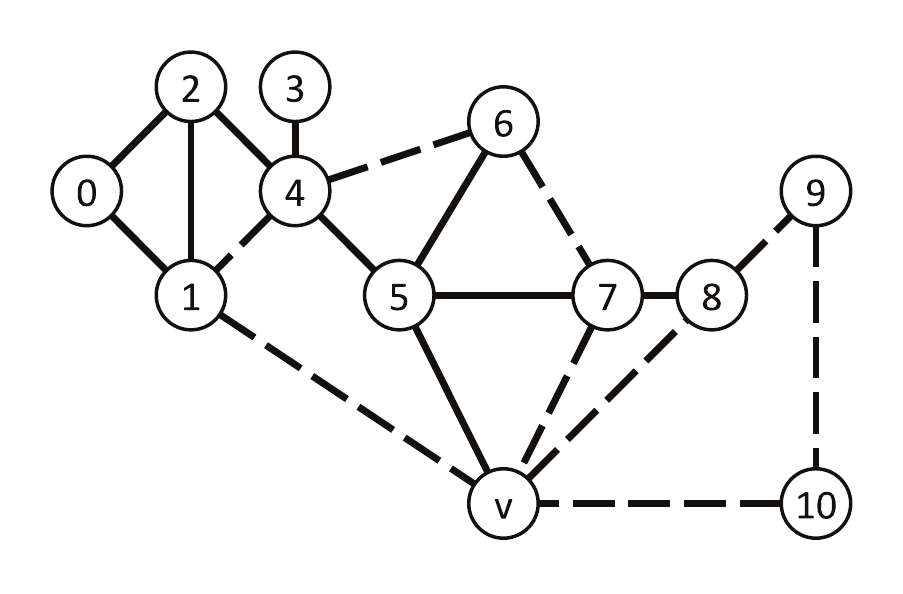}
}
\subfigure[Clique tree~$(T,\beta)$.]{\label{fig:localsearch2}
\image{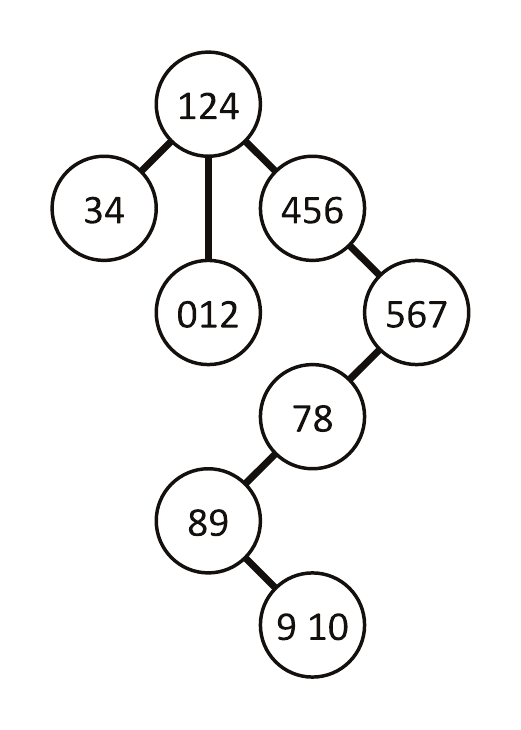}
}
\subfigure[Searching $1-5$ path.]{\label{fig:localsearch3}
\image{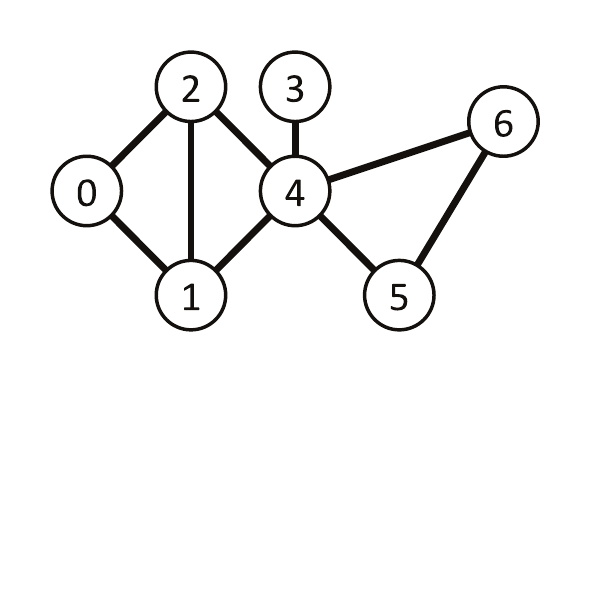}
}
\caption{\ref{fig:localsearch1} A graph~$G$ such that~$G-v$ is chordal. A $v$-flower of order two is highlighted by dotted edges. \ref{fig:localsearch2} A clique tree of~$G-v$, visualized by drawing the contents bags inside the nodes. With respect to this clique tree, the $v$-flower can be improved by Step~(\ref{improve:shorten}) by choosing~$P_i := (s=1,4,6,7=t)$ and~$t' := 5$. Note that~$d_T(\binv(s), \binv(t)) = 2$, while~$d_T(\binv(s),\binv(t')) = 1$. \ref{fig:localsearch3} The graph~$(G - (N_G[v] \setminus \{s,t'\})) - (V(\C) \setminus V(P_i))$ contains an induced $st'$-path~$P' := (1,4,5)$.} \label{fig:localsearch}
\end{center}
\end{figure}

Figure~\ref{fig:localsearch} illustrates these concepts. The applicability of each of the three improvement steps can be tested in polynomial time. For~(\ref{improve:addpath}) and~(\ref{improve:shorten}) this follows from the fact that induced paths can be found using breadth-first search; for~(\ref{improve:addtwoflower}) this follows from Proposition~\ref{proposition:twoflower}. The number of improvement steps of type~(\ref{improve:addpath}) and~(\ref{improve:addtwoflower}) is bounded by the maximum order of the flower, which is at most~$|V(G)|$. Step~(\ref{improve:shorten}) replaces a path in~$\C - v$ by one whose endpoints are closer together in~$T$. As there are only~$|V(G)|^2$ pairs of possible endpoints, a path cannot be replaced more than~$|V(G)|^2$ times by a better path. Since the flower contains at most~$|V(G)|$ paths at any time, there can be no more than~$|V(G)|^3$ successive applications of Step~(\ref{improve:shorten}) before the algorithm either terminates, or applies a step that increases the order of the flower. The algorithm therefore terminates after a polynomial number of steps. It is easy to verify that the stated conditions ensure that the structure~$\C$ is a $v$-flower at all times. Hence we can compute a $v$-flower~$\C$ that is maximal with respect to the three improvement rules in polynomial time.

\paragraph{Computing a hitting set.} Based on a maximal flower~$\C = \{C_1, \ldots, C_k\}$ and the associated set of paths~$P_1, \ldots, P_k$ in~$G - v$, we compute a set~$S \subseteq V(G) \setminus \{v\}$ that hits all holes. To every vertex of~$G - v$ we associate an edge of the decomposition tree~$T$. For~$u \in V(G) \setminus \{v\}$, define the \emph{cutpoint above~$u$}, denoted~$\pi(u)$, as follows. Let~$\pi(u)$ be the first edge~$e \in E(T)$ on the path from~$\top(u)$ to the root for which~$\adh(e) \subseteq N_G(v) \cup V(\C)$, or NIL if no such edge exists. The hitting set~$S$ is constructed as follows.

\begin{enumerate}[(a)]
	\item For each path~$P_i$ with~$i \in [k]$, add the endpoints of~$P_i$ to~$S$.\label{hittingset:endpoints}
	\item For each vertex~$u \in N_G(v) \setminus V(\C)$ with~$\pi(u) \neq$ NIL, add~$\adh(\pi(u)) \setminus N_G(v)$ to~$S$.\label{hittingset:cutpoint}
\end{enumerate}

It is easy to carry out this construction in polynomial time.

\begin{claim} \label{claim:s:subset:flower}
The set~$S$ is a subset of~$V(\C) \setminus \{v\}$.
\end{claim}
\begin{claimproof}
The endpoints of the paths~$P_i$ that are added to~$S$ in the first step clearly belong to the flower. The vertices added to~$S$ in the second step belong to~$\adh(\pi(u)) \setminus N_G(v)$. By definition of cutpoint we have~$\adh(\pi(u)) \subseteq N_G(v) \cup V(\C)$, hence~$\adh(\pi(u)) \setminus N_G(v) \subseteq V(\C)$. As~$(T,\beta)$ is a clique tree of~$G - v$, vertex~$v$ does not occur in any set~$\adh(\pi(u))$ and is not added to~$S$.
\end{claimproof}

\begin{claim} \label{claim:saturated:edge}
If~$\C$ is a maximal $v$-flower, then for any distinct~$s, t \in N_G(v) \setminus V(\C)$ with~$st \not \in E(G)$, the path~$\MPtd(s,t)$ contains~$\pi(s)$ or~$\pi(t)$.
\end{claim}
\begin{claimproof}
Let~$s,t \in N_G(v) \setminus V(\C)$ be nonadjacent in~$G$. Assume for a contradiction that~$\adh(e) \setminus (N_G(v) \cup V(\C)) \neq \emptyset$ for all~$e$ on~$\MPtd(s,t)$. Let~$U := (N_G(v) \cup V(\C)) \setminus \{s,t\}$, so~$\adh(e) \setminus U \neq \emptyset$ for all edges~$e$ on~$\MPtd(s,t)$. By applying Proposition~\ref{proposition:path:from:adhesions} to the chordal graph~$G - v$ we find an induced $st$-path in~$(G - v) - U$. Such a path forms a hole together with~$v$, since~$st \not \in E(G)$. So Step~(\ref{improve:addpath}) can be applied to increase the order of the $v$-flower,  contradicting the maximality of~$\C$.

Hence there is an edge~$e^*$ on~$\MPtd(s,t)$ for which~$\adh(e^*) \subseteq N_G(v) \cup V(\C)$. Observe that the trees~$\binv(s)$ and~$\binv(t)$ are vertex-disjoint as each bag forms a clique in~$G$ and~$st \not \in E(G)$. We conclude with a case distinction; see Figure~\ref{fig:pathcontainscutpoint} for an illustration.

\begin{figure}[t]
\begin{center}
\subfigure[Case~\ref{case:containcutpointsbelowt}.]{\label{fig:pathcontainscutpoint1}
\image{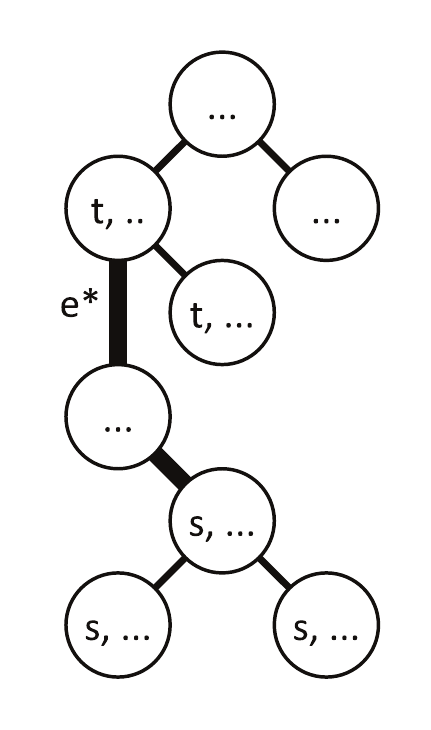}
}
\subfigure[Case~\ref{case:containcutpointlca}.]{\label{fig:pathcontainscutpoint2}
\image{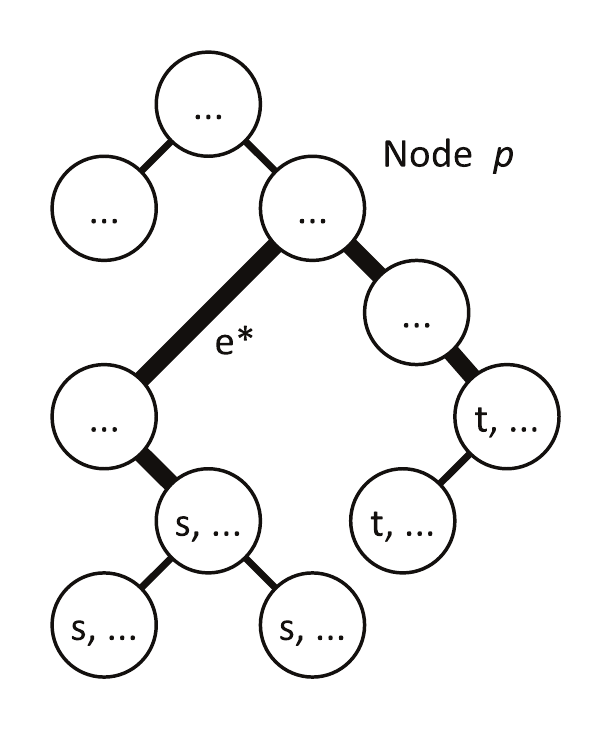}
}
\caption{Illustration of the proof of Claim~\ref{claim:saturated:edge}. The figure shows a tree decomposition~$(T,\beta)$ of~$G-v$. The minimal path~$\MPtd(s,t)$ in~$T$ is drawn with thick edges. \ref{fig:pathcontainscutpoint1} All bags containing~$s$ are contained in the subtree of~$T$ rooted at the highest bag containing~$t$. \ref{fig:pathcontainscutpoint2} If the first two cases do not hold, then the bags containing~$s$ are contained in a different subtree of the lowest common ancestor~$p$ than the bags containing~$t$.} \label{fig:pathcontainscutpoint}
\end{center}
\end{figure}

\begin{enumerate}
	\item \label{case:containcutpointsbelowt} If~$\binv(s)$ is contained in~$T_{\top(t)}$, then some vertex of~$\binv(t)$ is an ancestor to all vertices of~$\binv(s)$. Hence~$\MPtd(s,t)$ is the path from~$\top(s)$ to its first ancestor in~$\binv(t)$, so~$\MPtd(s,t) \subseteq \MPtd(s, \root)$. As~$e^*$ shows that at least one edge on~$\MPtd(s, \root)$ satisfies the defining condition for a cutpoint,~$\pi(s)$ is a well-defined. As~$\pi(s)$ is not higher than~$e^*$, which lies on~$\MPtd(s,t)$, the edge~$\pi(s)$ lies on~$\MPtd(s,t)$.
	\item If~$\binv(t)$ is contained in~$T_{\top(s)}$, then reversing the roles of~$s$ and~$t$ in the previous argument shows that~$\pi(t)$ lies on~$\MPtd(s,t)$.
	\item \label{case:containcutpointlca} If neither case holds, then~$\binv(s)$ and~$\binv(t)$ are contained in different child subtrees of~$p := \lca(\top(s), \top(t))$, implying that~$\MPtd(s,t)$ is the concatenation of a path~$\Ptd_s$ from~$\top(s)$ to its ancestor~$p$, and a path~$\Ptd_t$ from~$\top(t)$ to its ancestor~$p$. Since~$e^*$ lies on~$\MPtd(s,t)$, it lies on~$\Ptd_s$ or~$\Ptd_t$. Suppose that~$e^*$ lies on~$\Ptd_s$. Then~$\pi(s)$ is well-defined, lies above~$\top(s)$ and no higher than~$e^*$, and therefore lies on~$\MPtd(s,t)$. The case when~$e^*$ lies on~$\Ptd_t$ is symmetric.\qedhere
\end{enumerate}
\end{claimproof}

\begin{claim} \label{claim:s:hittingset}
The graph~$G - S$ is chordal.
\end{claim}
\begin{claimproof}
Assume for a contradiction that~$G - S$ contains a hole~$C^*$. Since~$G-v$ is chordal,~$C^*$ passes through~$v$ so that~$C^* - v$ is an induced $st$-path in~$G-S-v$ for nonadjacent~$s,t \in N_G(v)$. By Step~(\ref{hittingset:endpoints}) we know that~$s$ and~$t$ are not endpoints of paths in~$\C - v$, with respect to the computed maximal flower~$\C$. Since paths in~$\C - v$ cannot contain members of~$N_G(v)$ as interior vertices either, we find~$s,t \not \in V(\C)$.

Since~$s$ and~$t$ are nonadjacent vertices of~$N_G(v) \setminus V(\C)$, by Claim~\ref{claim:saturated:edge} we know~$\MPtd(s,t)$ contains~$\pi(s)$ or~$\pi(t)$. Let~$e^* \in \{\pi(s), \pi(t)\}$ be a well-defined edge on~$\MPtd(s,t)$. As~$s, t \in N_G(v) \setminus V(\C)$, by Step~(\ref{hittingset:cutpoint}) of the construction of~$S$ we know that~$\adh(e^*) \setminus S \subseteq N_G(v)$. As only one bag of~$\MPtd(s,t)$ contains~$s$ by minimality of~$\MPtd(s,t)$, and only one bag of~$\MPtd(s,t)$ contains~$t$, it follows that~$s,t \not \in \adh(e^*)$. On the other hand, since~$\MPtd(s,t)$ connects~$\binv(s)$ to~$\binv(t)$ in~$T$ and~$e^*$ lies on~$\MPtd(s,t)$, by Observation~\ref{observation:stpaths:use:adhesion} we know that the $st$-path~$P_v := C^* - v$ in~$G - S - v$ contains a vertex~$u \in \adh(e^*) \setminus S \subseteq N_G(v) \setminus \{s,t\}$. This shows that~$P_v$ contains a vertex from~$N_G(v)$ in its interior, and hence has a chord; a contradiction to the assumption that~$C^*$ is a hole.
\end{claimproof}

The claims so far established that we can compute a maximal $v$-flower~$\C$ and a corresponding chordal deletion set~$S$ that avoids~$v$, in polynomial time. To complete the proof, it remains to show that the size of~$S$ is bounded by twelve times the order of the computed flower. This is the most technical part. It follows from the next claim.

\begin{claim} \label{claim:twelve:from:path}
If~$C_i$ is a hole of the maximal flower~$\C$ and~$P_i := C_i - v$, then~$S$ contains at most~$12$ vertices from~$P_i$.
\end{claim}
\begin{claimproof}
Assume for a contradiction that~$C_i$ is a hole in~$\C$ such that~$S$ contains at least~$13$ vertices from the path~$P_i := C_i - v$. Let~$s,t \in N_G(v)$ be the endpoints of~$P_i$. We identify two special vertices in~$S \cap V(P_i)$ that allow us to derive a contradiction. Define~$p := \lca(\top(s), \top(t))$ and note that~$p$ may coincide with~$\top(s)$ or~$\top(t)$. Since~$\beta(p)$ is a clique in~$G-v$ and~$P_i$ is an induced path, we know~$|\beta(p) \cap V(P_i)| \leq 2$. Similarly,~$P_i$ contains at most two vertices of each of the sets~$\adh(\pi(s))$ and~$\adh(\pi(t))$.

Let~$X$ be the set containing the first two and last two vertices on~$P_i$, together with~$\beta(p) \cap V(P_i)$, $\adh(\pi(s)) \cap V(P_i)$, and~$\adh(\pi(t)) \cap V(P_i)$. Since~$|X| \leq 10$ and~$|S \cap V(P_i)| \geq 13$, there are at least three vertices in~$(S \cap V(P_i)) \setminus X$. Since they lie on the induced path~$P_i$, these three vertices do not form a clique implying that there are two distinct vertices in~$(S \cap V(P_i)) \setminus X$ that are not adjacent in~$G$. Call these~$x$ and~$y$ and choose their labels such that when traversing the path~$P_i$ from~$s$ to~$t$, we visit~$x$ before visiting~$y$. We will use the pair~$\{x,y\}$ to derive a contradiction to the maximality of~$\C$ with respect to the three presented improvement steps.

Since~$x$ and~$y$ are internal vertices of~$P_i$ and all paths in~$\C - v$ are vertex-disjoint,~$x$ and~$y$ are not endpoints of any path in~$\C - v$ and therefore Step~(\ref{hittingset:endpoints}) of the hitting set construction is not responsible for adding them to~$S$. Consequently,~$x$ and~$y$ are in~$S$ because there are vertices~$x',y' \in N_G(v) \setminus V(\C)$ with well-defined cutpoints~$\pi(x'), \pi(y')$ such that~$x \in \adh(\pi(x')) \setminus N_G(v)$ and~$y \in \adh(\pi(y')) \setminus N_G(v)$ that caused~$x$ and~$y$ to be added to~$S$ in Step~(\ref{hittingset:cutpoint}); we will argue later that~$x'$ and~$y'$ are distinct.

Since no internal vertex of~$\MPtd(s,t)$ lies in~$\binv(s)$ or~$\binv(t)$, there is a unique tree~$T^*$ in the forest~$T - (\binv(s) \cup \binv(t))$ that contains the internal nodes on~$\MPtd(s,t)$. To derive a contradiction we will argue for the following properties; note that~(\ref{twelve:twoflower}) contradicts the maximality of~$\C$ with respect to Step~(\ref{improve:addtwoflower}).

\begin{enumerate}[(i)]
	\item $\pi(x')$ and~$\pi(y')$ are distinct edges, hence~$x' \neq y'$.\label{twelve:xy:distinct:cutpoint}
	\item $\pi(x')$ and~$\pi(y')$ both belong to~$T^*$.\label{twelve:xy:cutpoints:in:subtree}
	\item The edges~$\pi(x')$ and~$\pi(y')$ are distinct from~$\pi(s)$ and~$\pi(t)$.\label{twelve:primetost:distinct:cutpoint}
	\item Neither one of~$\{x',y'\}$ is simultaneously adjacent to both~$s$ and~$t$.\label{twelve:prime:not:both:st}
	\item Neither one of~$\{x',y'\}$ is adjacent to~$s$ or~$t$.\label{twelve:prime:not:one:st}
	\item There are paths from~$x'$ to~$x$ and from~$y'$ to~$y$ that are pairwise vertex-disjoint and whose interior avoids~$V(\C) \cup N_G(v)$.\label{twelve:twopaths}
	\item The graph~$G - (V(\C) \setminus V(C_i))$ contains a $v$-flower of order two.\label{twelve:twoflower}
\end{enumerate}

(\ref{twelve:xy:distinct:cutpoint}) If~$\pi(x') = \pi(y')$, then~$x,y \in \adh(\pi(x')) = \adh(\pi(y'))$ are adjacent in~$G-v$ because they occur in a common bag of a clique tree. This contradicts our choice of~$x$ and~$y$ as nonadjacent vertices. Since~$\pi(x') \neq \pi(y')$ we must have~$x' \neq y'$.

(\ref{twelve:xy:cutpoints:in:subtree}) Since~$\pi(x')$ is an edge of~$T$ whose adhesion contains~$x$, it follows that~$\pi(x')$ is an edge of the subtree~$\binv(x)$. Similarly,~$\pi(y')$ is an edge of~$\binv(y)$. It therefore suffices to prove that the subtrees~$\binv(x)$ and~$\binv(y)$ are entirely contained in~$T^*$. This follows from Proposition~\ref{proposition:induced:paths:in:clique:tree}, since both~$x$ and~$y$ are vertices on the induced~$st$-path~$P_i$ in the chordal graph~$G - v$ that are not among the first two or last two vertices.

(\ref{twelve:primetost:distinct:cutpoint}) Suppose that~$\pi(x')$ equals~$\pi(s)$ or~$\pi(t)$. Then~$x \in \adh(\pi(x'))$ is a vertex on~$P_i$ that is in~$\adh(\pi(s)) \cup \adh(\pi(t))$. But this contradicts our choice of~$x$ to be a vertex of~$(S \cap V(P_i)) \setminus X$ above, since~$X$ contains~$\adh(\pi(s)) \cap V(P_i)$ and~$\adh(\pi(t)) \cap V(P_i)$. The argument when~$\pi(y')$ equals~$\pi(s)$ or~$\pi(t)$ is symmetric.

(\ref{twelve:prime:not:both:st}) Assume for a contradiction that~$x'$ is adjacent to both~$s$ and~$t$. We make a case distinction on the structure of~$\binv(s)$ and~$\binv(t)$ in the clique tree~$T$.

\begin{enumerate}
	\item If~$\binv(s)$ is contained in~$T_{\top(t)}$, then let~$p$ be the lowest node of~$\binv(t)$ for which~$\binv(s) \subseteq T_p$. Then tree~$T^*$ as defined above is contained in~$T_p$. As we assumed~$x'$ is adjacent to both~$s$ and~$t$, vertices~$x'$ and~$t$ occur in a common bag, showing that~$\top(x')$ is a proper ancestor to all nodes of~$T^*$. Since~$\pi(x')$ is an edge on~$\MPtd(x', \root)$, the edge~$\pi(x')$ lies above~$\top(x')$ while all of~$T^*$ lies below~$\top(x')$. This implies that~$\pi(x')$ does not belong to~$T^*$, contradicting~(\ref{twelve:xy:cutpoints:in:subtree}).
	
	\item If~$\binv(t)$ is contained in~$T_{\top(s)}$, then reversing the roles of~$s$ and~$t$ in the previous argument yields a contradiction.
	
	\item If neither case holds, then~$\binv(s)$ and~$\binv(t)$ are contained in different child subtrees of~$p := \lca(\top(s), \top(t))$. It follows that~$\MPtd(s,t)$ is the concatenation of a path from~$\top(s)$ to~$p$ with a path from~$\top(t)$ to~$p$. Since~$x'$ is adjacent to both~$s$ and~$t$, we know that~$\binv(x')$ intersects both~$\binv(s)$ and~$\binv(t)$, and therefore contains~$p$. It follows that~$\top(x')$ is an ancestor of~$p$, and therefore that the edge~$\pi(x')$ connects two ancestors of~$p$. By Proposition~\ref{proposition:inducedpath:on:minimalpath} we know that~$x$, which is a vertex on the induced $st$-path~$P_i$, occurs in a bag on~$\MPtd(s,t)$. Since~$x \in \adh(\pi(x'))$ and~$\pi(x')$ is an edge connecting two ancestors of~$p$, while all paths from ancestors of~$p$ to nodes in~$\MPtd(s,t)$ go through~$p$, we have~$x \in V(P_i) \cap \beta(p)$. This contradicts the choice of~$x$, as we chose it from a set that excludes~$V(P_i) \cap \beta(p)$.
\end{enumerate}

The argument when~$y'$ is adjacent to both~$s$ and~$t$ is symmetric.

(\ref{twelve:prime:not:one:st}) Assume for a contradiction that~$z' \in \{x',y'\}$ is adjacent to~$r \in \{s,t\}$; by~(\ref{twelve:prime:not:both:st}) we then already know that~$z'$ is not adjacent to~$\hat{r} := \{s,t\} \setminus \{r\}$. We distinguish two cases, depending on the relative position of~$\binv(r)$ and~$\binv(\hat{r})$.

\begin{enumerate}
	\item If~$\binv(\hat{r}) \subseteq T_{\top(r)}$, then let~$p$ be the lowest node of~$\binv(r)$ for which~$\binv(\hat{r}) \subseteq T_p$. It follows that~$T^* \subseteq T_p$. Now observe that since~$z'$ is adjacent to~$r$ in~$G$, the trees~$\binv(z')$ and~$\binv(r)$ share at least one node, and hence~$\top(z')$ is not a proper descendant of~$p$. Since~$\pi(z')$ is an edge on~$\MPtd(z', \root)$, it follows that~$\pi(z')$ does not lie in~$T^*$. But this contradicts~(\ref{twelve:xy:cutpoints:in:subtree}).
	\item If the previous case does not hold, then~$\MPtd(r,\hat{r})$ does not enter~$\binv(r)$ ``from below'', which implies that~$\MPtd(r,\hat{r})$ contains the edge from~$\top(r)$ to its parent. We now distinguish two cases.
	\begin{enumerate}
		\item If~$\top(z')$ is a proper ancestor of~$\top(r)$, then since~$z'$ occurs in a common bag with~$r$ we must have that~$z'$ is contained in the bag of the parent of~$\top(r)$. We claim that, in this situation, Step~(\ref{improve:shorten}) is applicable to improve the structure of the flower, contradicting the maximality of~$\C$ with respect to the given rules. To see this, observe that the presence of~$z'$ in the parent bag of~$r$, together with the fact that~$\MPtd(r,\hat{r})$ uses the edge from~$\top(r)$ to its parent, implies that~$d_T(\binv(z'), \binv(\hat{r}))$ is strictly smaller than~$d_T(\binv(r), \binv(\hat{r}))$; the latter quantity is exactly the number of edges on~$\MPtd(r,\hat{r})$. By Observation~\ref{observation:stpaths:use:adhesion}, the $st$-path~$P_i$ contains a vertex~$u$ of~$\adh(e)$, where~$e$ is the edge from~$\top(r)$ to its parent in~$T$. Consequently,~$u$ is not~$r$ itself since~$r$ does not appear in the parent bag of~$\top(r)$. Since both~$u$ and~$z'$ appear in that parent bag, we know that~$z'$ is adjacent to~$u$ and therefore that~$z'$ together with the subpath of~$P_i$ from~$u$ to~$\hat{r}$ is a~$z'\hat{r}$-path in~$G$. Hence, there is an induced $z'\hat{r}$-path in~$G$ on a vertex subset of~$V(P_i) \cup \{z'\}$. As~$z' \in \{x',y'\} \subseteq N_G(v) \setminus V(\C)$ and~$z$ is not adjacent to~$\hat{r}$, which is one of the endpoints of~$P_i$, this shows that Step~(\ref{improve:shorten}) can be applied to improve the structure of the flower~$\C$; a contradiction.
		\item If~$\top(z')$ is not a proper ancestor of~$\top(r)$, then since~$z'$ and~$r$ occur in a common bag we know that~$\top(z') \in T_{\top(r)}$. Since the cutpoint~$\pi(z')$ above~$z'$ lies on~$\MPtd(z', \root)$ by definition, and is contained in~$T^*$ (by~(\ref{twelve:xy:cutpoints:in:subtree})) which lies ``above''~$\top(r)$ by the precondition to this case, it follows that~$\pi(z')$ is an edge between two ancestors of~$\top(r)$. As the adjacent vertices~$z'$ and~$r$ occur in a common bag and~$\top(z')$ is not higher than~$\top(r)$, path~$\MPtd(z', \root)$ goes through~$\top(r)$ and all bags on the path from~$\top(z')$ to~$\top(r)$ contain~$r$. Since no bag in~$T^*$ contains~$r$ and~$\pi(z') \in T^*$ it follows that~$\pi(z')$ is an edge connecting two ancestors of~$r$, and must therefore be the first edge~$e$ on the path from~$\top(r)$ to the root for which~$\adh(e) \subseteq N_G(v) \cup V(\C)$. But then~$\pi(z') = \pi(r)$, which contradicts~(\ref{twelve:primetost:distinct:cutpoint}).
	\end{enumerate}
\end{enumerate}

\widefigure{
\subfigure[Clique tree.]{\label{fig:twoflower1}
\bigimage{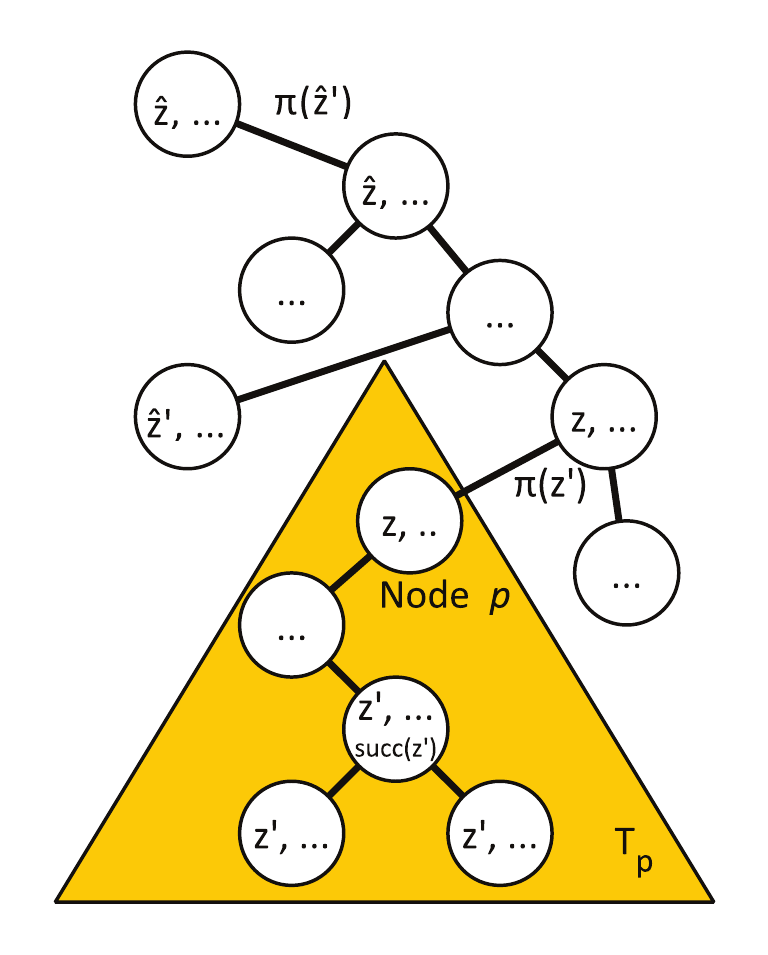}
}
\subfigure[Intersecting paths.]{\label{fig:twoflower2}
\image{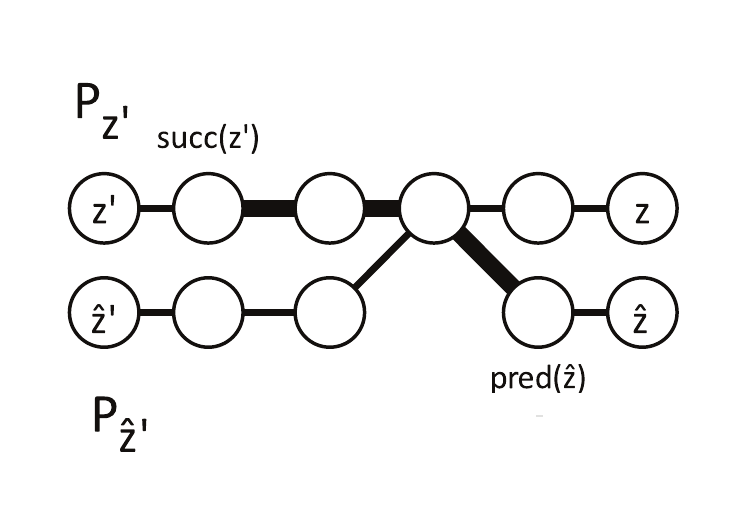}
}
\subfigure[$v$-flower.]{\label{fig:twoflower3}
\image{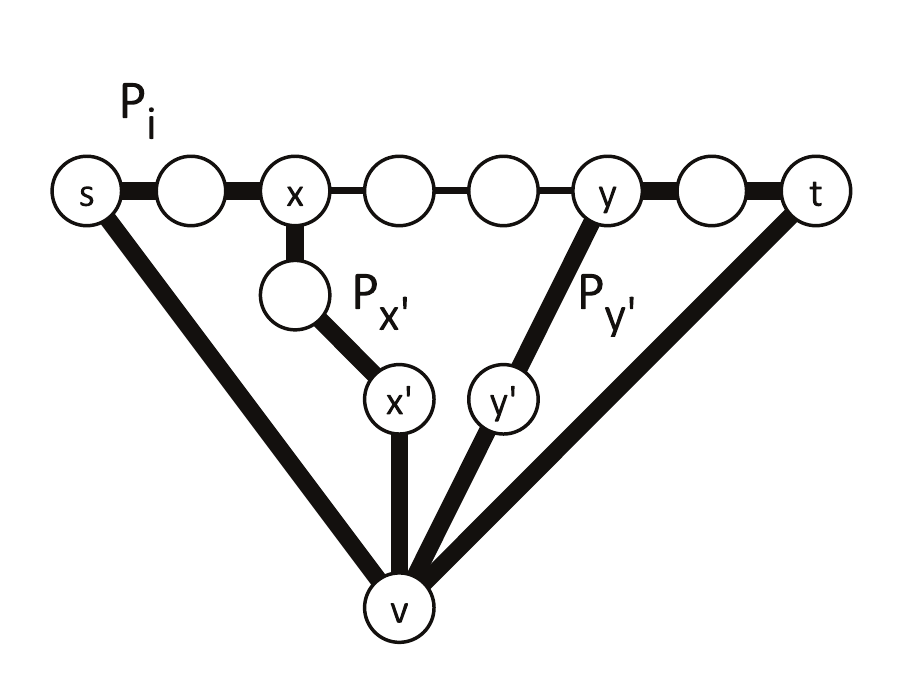}
}}{
\caption{\ref{fig:twoflower1} Illustration for part~(\ref{twelve:twopaths}) of the proof of Claim~\ref{claim:twelve:from:path}. Example of how the different entities can be arranged in the clique tree~$(T,\beta)$ of~$G-v$. \ref{fig:twoflower2} The structure of a hypothetical intersection in the interiors of the paths~$P_{z'}$ and~$P_{\hat{z}'}$. Path~$Q$ is shown by thick edges. \ref{fig:twoflower3} Illustration for part~(\ref{twelve:twoflower}) of Claim~\ref{claim:twelve:from:path}. Drawing of how the constructed paths~$P_{x'}$ and~$P_{y'}$ are combined with the original petal~$P_i := C_i - v$ of the flower, to obtain a flower of order two (visualized by thick edges).} \label{fig:twoflower}
}

(\ref{twelve:twopaths}) We first construct the paths and then prove that they have interiors that avoid~$V(\C) \cup N_G(v)$ and are pairwise vertex-disjoint. Let~$z' \in \{x',y'\}$ and let~$z$ be the corresponding vertex in~$\{x,y\}$. Let~$\Ptd'$ be the path in~$T$ from~$\top(z')$ to the lower endpoint of~$\pi(z')$. Since~$z \in \adh(\pi(z'))$ by definition of~$x,y,x',y'$, we know that~$z$ appears in the bag of the bottom endpoint of~$\pi(z')$ and so~$\Ptd'$ connects~$\binv(z)$ to~$\binv(z')$ in~$T$. Since none of the edges of~$\Ptd'$ satisfied the condition for being the cutpoint above~$z'$, for each edge~$e$ on~$\Ptd'$ the set~$\adh(e) \setminus (N_G(v) \cup V(\C))$ is not empty. Define~$U := (N_G(v) \cup V(\C)) \setminus \{z,z'\}$. Applying Proposition~\ref{proposition:path:from:adhesions} to the chordal graph~$G-v$ where~$z'$ plays the role of~$s$ and~$z$ plays the role of~$t$, we establish that there is an induced $zz'$-path in the graph~$(G-v)-U$. The interior vertices on this path avoid~$U \cup \{z,z'\}$ (since the latter are endpoints), and therefore the interior avoids~$N_G(v) \cup V(\C)$ as demanded by~(\ref{twelve:twopaths}). As this holds for both choices of~$z$, we obtain an induced~$xx'$-path~$P_{x'}$ and an induced~$yy'$-path~$P_{y'}$. It remains to prove that~$P_{x'}$ and~$P_{y'}$ are pairwise vertex-disjoint. By~(\ref{twelve:xy:distinct:cutpoint}) we know that~$x' \neq y'$. As we chose~$x$ and~$y$ to be distinct vertices, and~$x,y \in V(\C)$ while~$x',y' \in N(v) \setminus V(\C)$, it is clear that the endpoints of~$P_{x'}$ and~$P_{y'}$ are four distinct vertices. It remains to prove that the interior of~$P_{x'}$ is vertex-disjoint from the interior of~$P_{y'}$. If one of the two paths has no interior, then this is trivial. So assume for a contradiction that both paths have an interior, which includes a vertex common to both paths. Orient both paths from the prime to the non-prime endpoint.

Choose~$z' \in \{x',y'\}$ and~$\hat{z}' \in \{x',y'\} \setminus \{z'\}$ such that the bottom endpoint~$p$ of~$\pi(z')$ satisfies~$\pi(\hat{z}') \not \in T_p$. This is always possible; see Figure~\ref{fig:twoflower} for an illustration. Let~$z \in \{x,y\}$ depending on whether~$z' = x'$ or~$z' = y'$; let~$\hat{z} \in \{x,y\} \setminus \{z\}$. Since the two interiors of~$P_{z'}$ and~$P_{\hat{z}'}$ intersect, there is a path~$Q$ in~$G-v$ from the successor~$\succ(z')$ of~$z'$ on~$P_{z'}$, to the predecessor$~\pred(\hat{z})$ of~$\hat{z}$ on~$P_{\hat{z}'}$, such that all vertices on~$Q$ belong to the interior of~$P_{x'}$ or~$P_{y'}$. So~$V(Q)$ avoids~$N_G(v) \cup V(\C)$. Vertex~$\succ(z')$ occurs in a bag of~$T_p$, since it is adjacent to~$z'$ and~$\top(z')$ is a descendant of both endpoints of~$\pi(z')$. To reach a contradiction we show that~$\pred(\hat{z})$ does not occur in a bag in~$T_p$. To that end, we first show that~$\hat{z}$ does not occur in a bag of~$T_p$. This follows from the fact that~$\hat{z}$ occurs in the adhesion of~$\pi(\hat{z}')$, that~$\pi(\hat{z}')$ is not contained in~$T_p$, and that if~$\hat{z}$ also occurred in a bag of~$T_p$ then this would force~$\hat{z}$ to be in~$\beta(p)$ which contains~$z \in \adh(\pi(z'))$; but this would imply that~$z$ and~$\hat{z}$ are adjacent, contradicting our choice of~$x$ and~$y$ in the beginning of the proof. So~$\hat{z}$ does not occur in~$T_p$, implying that~$\pred(\hat{z})$ occurs in some bag outside~$T_p$. Consider a path~$\Ptd$ in~$T$ from a bag in~$T_p$ containing~$\succ(z')$, to a bag outside~$T_p$ containing~$\pred(\hat{z}')$; we established that~$\Ptd$ contains the edge from the root of~$T_p$ to its parent, which is~$\pi(z')$. By Observation~\ref{observation:stpaths:use:adhesion}, the path~$Q$ connecting~$\succ(z')$ to~$\pred(\hat{z}')$ contains a vertex of~$\adh(\pi(z'))$. By definition of cutpoint, we have that~$\adh(\pi(z')) \subseteq N_G(v) \cup V(\C)$. But~$Q$ avoids~$N_G(v) \cup V(\C)$; a contradiction. It follows that~$P_{x'}$ and~$P_{y'}$ are indeed vertex-disjoint, proving~(\ref{twelve:twopaths}).

(\ref{twelve:twoflower}) Consider the $xx'$-path~$P_{x'}$ and the $yy'$-path~$P_{y'}$, whose existence we proved in~(\ref{twelve:twopaths}). The interiors of these paths avoid~$V(\C) \cup N_G(v)$, the endpoints~$x'$ and~$y'$ belong to~$N_G(v)$, and the endpoints~$x$ and~$y$ belong to~$V(\C) \setminus N_G(v)$. Consider the walk in~$G-v$ that starts at~$s$, traverses the $st$-path~$P_i = C_i - v$ until~$x$, and then traverses~$P_{x'}$ to~$x'$. Since~$s$ is not adjacent to~$x'$ by~(\ref{twelve:prime:not:one:st}) and~$x' \in N_G(v)$, this walk contains an induced $sx'$-path between two neighbors of~$v$, whose internal nodes avoid~$N_G(v)$. Moreover, the only vertices of~$\C$ used on this path belong to~$P_i$. Similarly, there is an induced $ty'$-path within the walk starting at~$t$, traversing~$P_i$ from~$t$ to~$y$, and then traversing~$P_{y'}$ from~$y$ to~$y'$. Since we chose~$x$ and~$y$ such that~$x$ occurs before~$y$, the pieces of~$P_i$ used by these two induced paths are vertex-disjoint. By~(\ref{twelve:twopaths}), the pieces of~$P_{x'}$ and~$P_{y'}$ used are also vertex-disjoint. Hence these two induced paths are vertex-disjoint and each connect two nonadjacent neighbors of~$v$. Together with~$v$ these form a $v$-flower of order two in~$G$. Since the only vertices of~$\C$ used on this order-$2$ flower are those of~$P_i$ and~$v$, it follows that there is a $v$-flower of order two in~$G - (V(\C) \setminus V(C_i))$, establishing~(\ref{twelve:twoflower}). This contradicts the fact that~$\C$ is maximal with respect to Step~(\ref{improve:addtwoflower}) and completes the proof of Claim~\ref{claim:twelve:from:path}.
\end{claimproof}

Armed with these claims we finish the proof of Lemma~\ref{lemma:algorithm:packing:vs:covering}. Claim~\ref{claim:s:hittingset} shows that~$S$ is a hitting set avoiding~$v$ for the holes in~$G$. All vertices of~$S$ belong to~$V(\C) \setminus \{v\}$ by Claim~\ref{claim:s:subset:flower}, and therefore all lie on some path of the structure~$\C - v$. Since the number of paths equals the order of the flower, while~$S$ contains at most~$12$ vertices from each path by Claim~\ref{claim:twelve:from:path}, it follows that the size of~$S$ is at most~$12$ times the order of the flower. This concludes the proof of Lemma~\ref{lemma:algorithm:packing:vs:covering}.
\end{proof}

%For a graph~$G$, let~$\pack_v(G)$ be the minimum cardinality of a vertex set~$S \subseteq V(G)$ that does not include~$v$, such that~$G - S$ is chordal. Similarly, let~$\cover_v(G)$ be the maximum order of a $v$-flower in~$G$. The following is a direct consequence of Lemma~\ref{lemma:algorithm:packing:vs:covering}.
%
%\begin{corollary}
%Let~$G$ be a graph and let~$v \in V(G)$ such that~$G - v$ is chordal. Then~$\cover_v(G) \leq 12 \cdot \pack_v(G)$.
%\end{corollary}

\section{Annotated Chordal Vertex Deletion} \label{s:annotated}
In what follows it is more convenient to work with a variant of the \ChVD
problem with some annotations on the set $M$.
More formally, the input to the problem \AChVDlong (\AChVD for short)
is a tuple $(G,k,M,E^h)$ with the following properties:
\begin{enumerate}
\item $G$ is a graph, $k$ is an integer, $M \subseteq V(G)$, and $E^h \subseteq E(G[M])$;
\item $G-M$ is a chordal graph, and moreover, for every $v \in M$ the graph
$G-(M \setminus \{v\})$ is a chordal graph.
\end{enumerate}
The \AChVD problem asks for the existence of a set $X \subseteq V(G)$,
called henceforth a \emph{solution}, such that 
$|X| \leq k$, $G-X$ is chordal, and for every pair $\{x,y\} \in E^h$,
at least one of $x$ and $y$ belongs to $X$. 

Lemma~\ref{lemma:algorithm:packing:vs:covering} can be used to transform an instance of the general \ChVD problem into an equivalent instance satisfying the requirements of the annotated problem, if a modulator~$M$ is known. For the final kernelization algorithm, such a modulator will be obtained from the approximation algorithm of Theorem~\ref{thm:apx}.

\begin{lemma} \label{lem:annotate}
There is a polynomial-time algorithm that, given a graph~$G$, an integer~$k$, and a set~$M_0 \subseteq V(G)$ such that~$G-M_0$ is chordal, either correctly determines that~$(G,k)$ is a no-instance of \ChVD or computes an instance~$(G',k',M,E^h)$ of \AChVD with the same answer such that~$|M| \in \Oh(k \cdot |M_0|)$ and~$k' \leq k$.
\end{lemma}
\begin{proof}
For every $v \in M_0$, we apply the algorithm of Lemma~\ref{lemma:algorithm:packing:vs:covering} to the nearly-chordal graph~$G - (M_0 \setminus \{v\})$, obtaining a $v$-flower $\C_v$ and a hitting set $S_v$. If $|\C_v| > k$ for some $v$, we delete the vertex $v$ from $G$, decrease $k$ by one, and restart the algorithm. If the budget~$k$ drops below zero then we report that the original input was a no-instance of \ChVD. Clearly, if $|\C_v| > k$, then every size-$k$ solution to \ChVD on~$G$ contains the vertex $v$, so~$(G,k)$ is equivalent to~$(G-v, k-1)$. If $|\C_v| \leq k$ for every $v \in M_0$, we define $M := M_0 \cup \bigcup_{v \in M_0} S_v$; note that $|M| \leq |M_0|(12k+1)$. Letting~$G'$ and~$k'$ denote the reduced graph and budget, we output the \AChVD instance~$(G',k',M,E^h = \emptyset)$. It follows from Lemma~\ref{lemma:algorithm:packing:vs:covering}
that it is a valid \AChVD instance, and clearly it is a yes-instance if and only if the original \ChVD instance $(G,k)$ is a yes-instance.
\end{proof}

The next few sections are devoted to the proof of the following theorem.

\begin{theorem}\label{thm:AChVD-kernel}
\AChVD admits a polynomial kernel when parameterized by $|M|$.
More precisely, given an \AChVD instance $(G,k,M,E^h)$, one can in polynomial
time compute an equivalent instance $(G',k,M,F^h)$ with $E^h \subseteq F^h$
and~$|V(G')| \in \Oh(k^{16} |M|^{29})$.
\end{theorem}

As typical for kernelization algorithms, the algorithm of Theorem~\ref{thm:AChVD-kernel} consists of a number of reduction rules;
at every step, we apply the lowest-numbered reduction rule. Observe that any instance with~$k \geq |M|$ has a trivial YES-answer as the set~$M$ is a solution. For such instances we can output a constant-size YES-instance of \AChVD. In the remainder we assume~$|M| > k$.

Let $(G,k,M,E^h)$ be a \AChVD instance.
We say that a vertex set 
$S \subseteq V(G) \setminus M$ is \emph{irrelevant}
if the instances $(G,k,M,E^h)$ and $(G-S,k,M,E^h)$ are equivalent.
We say that a pair $xy \in \binom{M}{2}$ is a \emph{new forced pair}
if $xy \notin E^h$, but
the instances $(G,k,M,E^h)$ and $((V(G),E(G) \cup \{xy\}),k,M,E^h \cup \{xy\})$ are equivalent.
Note that adding an edge $xy$ to $G$ is meaningless, as a solution is required to delete $x$ or $y$ anyway,
but it makes notation later cleaner.

A hole $H$ in the graph $G$ is a \emph{permitted hole} if $V(H)$ does not contain
any pair of $E^h$. Note that a hole that is not permitted is hit by any
solution to $(G,k,M,E^h)$ due to the constraints imposed by $E^h$, and,
informally speaking, we can ignore such holes in the further argumentation.

In what follows, we will use the following shorthand notation.
Let $x_1,x_2,\ldots,x_a,\linebreak[0] y_1,y_2,\ldots,y_b \in M$ be (not necessarily
distinct) vertices for some $a,b \geq 0$.
Then we define:
\begin{align*}
V(x_1,x_2,\ldots,x_a,\neg y_1, \neg y_2, \ldots, \neg y_b) &= \left ( (V(G) \setminus M) \cap \bigcap_{i=1}^a N(x_i) \right ) \setminus \left( \bigcup_{j=1}^b N(y_j) \right) \\
G(x_1,x_2,\ldots,x_1,\neg y_1,\neg y_2, \ldots, \neg y_b) &= G[V(x_1,x_2,\ldots,x_a,\neg y_1, \neg y_2, \ldots, \neg y_b)].
\end{align*}
Consider the following reduction rule.
\begin{reduction}\label{red:common-neighbours}
For every $x,y \in M$ with $x \neq y$, $xy \notin E(G)$,
if the size of the largest independent set in $G(x,y)$ is
at least $k+2$, then add $xy$ to $E^h$ and to $E(G)$.
\end{reduction}
The correctness is straightforward: any solution to $(G,k,M,E^h)$ can delete
at most $k$ out of $k+2$ vertices of any maximum independent set in $G(x,y)$;
the remaining vertices form a $C_4$ with $x$ and $y$, forcing the solution
to delete either $x$ or $y$. We add the edge $xy$ to $E(G)$ to preserve the properties of an \AChVD instance.
Furthermore, note that as $G(x,y)$ is chordal, maximum independent sets
can be computed in polynomial time~\cite{Gavril72}.

In the remainder we will often use the following simple observation to find holes, in proofs by contraposition.
\begin{observation}[{\cite[Proposition 3]{Marx10}}] \label{obs:find-hole}
If a graph $G$ contains a vertex $v$ with two nonadjacent neighbors $u_1,u_2 \in N(v)$, and a walk from $u_1$ to $u_2$ with all internal vertices in $V(G) \setminus N[v]$, then $G$ contains a hole passing through $v$.
\end{observation}

We conclude this section with an important immediate observation about \AChVD.
\begin{proposition}\label{prop:NC-clique}
Let $(G,k,M,E^h)$ be an instance of \AChVD. Then, for every $v \in M$
and every connected component $A$ of $G(\neg v)$, 
the set $N_G(A) \setminus M$ is a clique.
\end{proposition}
\begin{proof}
Suppose there exist two nonadjacent vertices $u_1,u_2 \in N(A) \setminus M$. Apply Observation~\ref{obs:find-hole} to the graph~$G - (M \setminus \{v\})$ with the vertex $v$ and a path between $u_1$ and $u_2$ with internal vertices in $A$. This yields a hole in~$G - (M \setminus \{v\})$, which contradicts the requirement of \AChVD instances that~$G - (M \setminus \{v\})$ is chordal for all~$v \in M$.
\end{proof}

\section{Reducing the number of components \texorpdfstring{wrt.\,}{wrt. }a separator}\label{s:toughness}
The next procedure reduces an instance~$(G,k,M,E^h)$ of \AChVD if we have access to a small vertex set whose removal splits the graph into many components. In general, it is not easy to find such a vertex set since it is related to computing the toughness of a graph, which is NP-hard~\cite{BauerHS90}. The reduction can therefore only be applied for a concrete separator that is supplied to the procedure. For that reason, we call it a \emph{reduction template}, instead of a concrete reduction rule. We will apply the template in situations where such a vertex set can be identified efficiently. 

In the following, we will say ``mark up to $f(k)$ objects of type~$X$'' to mean the following. If there are more than~$f(k)$ objects of type~$X$, then choose~$f(k)$ of them arbitrarily and mark them. If there are fewer, then mark all of them.

\begin{template}\label{template:toughness}
Given a vertex subset~$S \supseteq M$ of~$G$, do the following.
\begin{enumerate}
	\item For every pair~$xy \in \binom{S}{2} \setminus E(G)$, mark up to~$k+2$ connected components~$C$ of~$G - S$ that contain both a neighbor of~$x$ and a neighbor of~$y$.
	\item For every pair~$xy \in \binom{S}{2} \cap E(G)$, mark up to~$k+1$ connected components~$C$ of~$G - S$ that contain a path~$P$ between a vertex in~$N_G(x) \cap V(C)$ and a vertex in~$N_G(y) \cap V(C)$ such that~$V(P)$ avoids~$N_G(x) \cap N_G(y)$.
\end{enumerate}
Delete the vertices in an unmarked component~$C$ from the graph.
\end{template}

Template~\ref{template:toughness} yields a stand-alone reduction that does not rely on any other reduction rules being exhaustively applied. Observe that if~$G - S$ has more than~$(k+2)\binom{|S|}{2}$ components, then the marking scheme will leave at least one component unmarked, and therefore the reduction will shrink the graph. Given~$S$, it is easy to apply the procedure in polynomial time. Observe that the existence of a path~$P$ through a component~$C$ satisfying the conditions of the second item can be tested by trying for each~$v \in N_G(x) \cap V(C)$ and~$u \in N_G(y) \cap V(C)$ whether~$u$ and~$v$ belong to the same connected component of~$G[V(C)] - (N_G(x) \cap N_G(y))$. Correctness follows from the next lemma. 

\begin{lemma}
If a component~$C$ of~$G - S$ is not marked by Reduction Template~\ref{template:toughness}, then~$V(C)$ is irrelevant.
\end{lemma}
\begin{proof}
Since~$S \supseteq M$, no forced pairs involve~$V(C)$. It therefore suffices to prove that if~$(G- V(C)) - X$ is chordal for some set~$X$ of size at most~$k$, then~$G-X$ is chordal as well. Assume for a contradiction that~$(G - V(C)) - X$ is chordal but~$G-X$ is not, and consider some hole~$H$ in~$G - X$. It uses at least one vertex~$v \in V(C)$. Let~$x$ and~$y$ be the first vertices of~$S$ that you encounter when traversing the hole~$H$ forwards and backwards, respectively, when starting in~$v$. By definition of \AChVD and the fact that~$M \subseteq S$ such vertices exist and~$x \neq y$.

If~$xy \not \in E(G)$, then~$C$ is unmarked and contains a path from a neighbor of~$x$ to a neighbor of~$y$. We marked at least~$k+2$ components that provide such a path and~$X$ intersects at most~$k$ of them, leaving two marked components~$C_1, C_2$ free. As each component~$C_i$ contains a path between a neighbor of~$x$ and a neighbor of~$y$, by shortcutting these paths we find induced~$xy$-paths whose interior vertices belong to~$C_i$, for each~$i \in \{1,2\}$. As~$C_1$ and~$C_2$ are distinct components of~$G-S$ and the only vertices from~$S$ used on these paths are~$x$ and~$y$, the interiors of the path are vertex-disjoint and the interior vertices of one path are not adjacent to interior of the other. Hence the concatenation of these paths is a hole in~$(G - V(C)) - X$, a contradiction.

Now consider the case that~$xy \in E(G)$. Since~$x$ and~$y$ are the first vertices of~$S$ encountered when traversing~$H$ starting at~$v$, it follows that~$H \cap V(C)$ is an induced path in~$C$ between a neighbor of~$x$ and a neighbor of~$y$ that avoids~$N_G(x) \cap N_G(y)$. The hole~$H$ consists of~$xy$ together with this single segment in~$H \cap V(C)$. Since~$C$ was not marked, there are~$k+1$ marked components containing a $N_G(x) \cap V(C)$ to~$N_G(y) \cap V(C)$ path avoiding~$N_G(x) \cap N_G(y)$. Consequently, one of these components~$C'$ is not intersected by~$X$. Let~$P$ be an induced path in~$C'$ from~$u_1 \in N_G(x) \cap V(C')$ to~$u_2 \in N_G(y) \cap V(C')$ that avoids~$N_G(x) \cap N_G(y)$. Let~$u'_1$ be the last vertex on~$P$ that is adjacent to~$x$, and let~$u'_2$ be the first vertex after~$u'_1$ that is adjacent to~$y$. The subpath of~$P$ from~$u'_1$ to~$u'_2$ is an induced path with at least two vertices from a neighbor of~$x$ to a neighbor of~$y$, and no interior vertex of this subpath is adjacent to~$x$ or~$y$. Moreover,~$u'_1 y, u'_2 x \not \in E(G)$ since~$P$ avoids~$N_G(x) \cap N_G(y)$. Hence~$P$ forms a hole together with~$x$ and~$y$, and this hole avoids~$X$ and~$V(C)$. So~$(G-V(C)) - X$ is not chordal; a contradiction.
\end{proof}

\section{Reducing cliques}\label{s:cliques}

Our goal for this section is to prove the following statement.

\begin{lemma}\label{lem:red-cliques}
Given an \AChVD instance $(G,k,M,E^h)$ 
on which Reduction~\ref{red:common-neighbours} is
not applicable, and a clique $K$
in $G-M$ of size $|K| > (k+1)(|M|^3 + (k+3)|M|^2)$,
one can in polynomial time find either an irrelevant vertex
or a new forced pair.
\end{lemma}

Let $\Gc := G-M$. Fix a clique tree of $\Gc$
and root it in a vertex $p$ such that $K \subseteq \beta(p)$. 

\subsection{Sets \texorpdfstring{$Q^{xy}$}{Q xy} and forced pairs}

We first look for a new forced pair in the following manner. 
Let $x$ and $y$ be two distinct vertices of $M$ with $xy \notin E(G)$.
We say that a node $q \in V(T)$ is \emph{$xy$-good}
if there exists an $xy$-path path $P_q$ in $G$ whose every internal
vertex $v$ belongs to $\Gc$ and satisfies $\binv(v) \subseteq T_q$.
Note that if $q$ is $xy$-good, then all its ancestors in the tree $T$
are also $xy$-good.
Let $Q^{xy}$ be the set of all maximally bottommost $xy$-good nodes in $T$,
that is, $q \in Q^{xy}$ if and only if 
$q$ is $xy$-good but none of the children of $q$ are $xy$-good. We claim the following.
\begin{claim}\label{cl:xy-good}
If $|Q^{xy}| \geq k+2$ for some pair $xy \in \binom{M}{2} \setminus E(G)$,
then every solution to the instance $(G,k,M,E^h)$ contains
at least one vertex from the pair $xy$.
\end{claim}
\begin{claimproof}
For every $q \in Q^{xy}$, let $P_q$ be a path witnessing that $q$ is $xy$-good;
without loss of generality, assume that $P_q$ is an induced path.
Consider two distinct $q,q' \in Q^{xy}$. 
Observe that, by the definition of $Q^{xy}$, the nodes $q$ and $q'$
are not in ancestor-descendant relationship in $T$, and thus the subtrees
$T_q$ and $T_{q'}$ are node-disjoint. Conseqently, no internal vertex 
of $P_q$ is equal or adjacent to any internal vertex of $P_{q'}$.
As $xy \notin E(G)$, the union of the paths $P_q$ and $P_{q'}$ is a hole in $G$.
As there are at least $k+2$ elements of $Q^{xy}$, any solution to $(G,k,M,E^h)$
avoids the set of internal vertices of at least two such paths, and therefore
contains $x$ or $y$.
\end{claimproof}
Claim~\ref{cl:xy-good} justifies the following reduction rule.
\begin{reduction}\label{red:xy-good}
If $|Q^{xy}| \geq k+2$ for some pair $xy \in \binom{M}{2} \setminus E(G)$,
then add $xy$ to $E^h$ and to $E(G)$.
\end{reduction}
Observe that it is straightforward to check directly from the definition
if a node $q \in V(T)$ is $xy$-good for a given pair $xy$.
Thus, Reduction~\ref{red:xy-good} can be applied in polynomial time,
and henceforth we assume $|Q^{xy}| \leq k+1$ for every pair
$xy \in \binom{M}{2} \setminus E(G)$.

\subsection{Marking relevant parts of a clique}

Having bounded the size of the sets $Q^{xy}$, we now mark a number
of vertices of the clique $K$, and then argue that any unmarked vertex
is irrelevant.
The argumentation here was loosely inspired by the analogous
part of the work of Marx~\cite[Section 5]{Marx10}.
%\todo{I recall you saying that you got some inspiration from reading Daniel's paper. Does it relate to this clique marking step, since he also has an argument to find irrelevant vertices in a clique? Should we say this procedure is ``inspired by Daniel's paper''?}
% Marcin: I admit I cannot now point exactly what was inspired, but I remember getting some idea from DM's paper and it was somewhere here. I think we can safely mention here inspiration.

Note that for every~$x \in M$, there is at most one connected component of~$G(\neg x)$ that contains a vertex of~$\beta(p)$, since~$\beta(p)$ is a clique in~$\Gc$. By~$A^x$ we denote the vertex set of this connected component of~$G(\neg x)$ if it exists. We define~$A^x := \emptyset$ otherwise, which occurs when~$\beta(p) \subseteq N(x)$. By Proposition~\ref{prop:NC-clique}, if $A^x \neq \emptyset$, then
$N(A^x) \setminus M$ is a clique in $\Gc$; let $p^x$ be a node of $T$
such that $N(A^x) \setminus M \subseteq \beta(p^x)$.
For a vertex $x \in M$ with $A^x = \emptyset$, we define $p^x = p$, which is the root of the clique tree.

\widefigure{
\subfigure[Graph~$G$, modulator~$M$.]{\label{fig:cliquemarking2}
\image{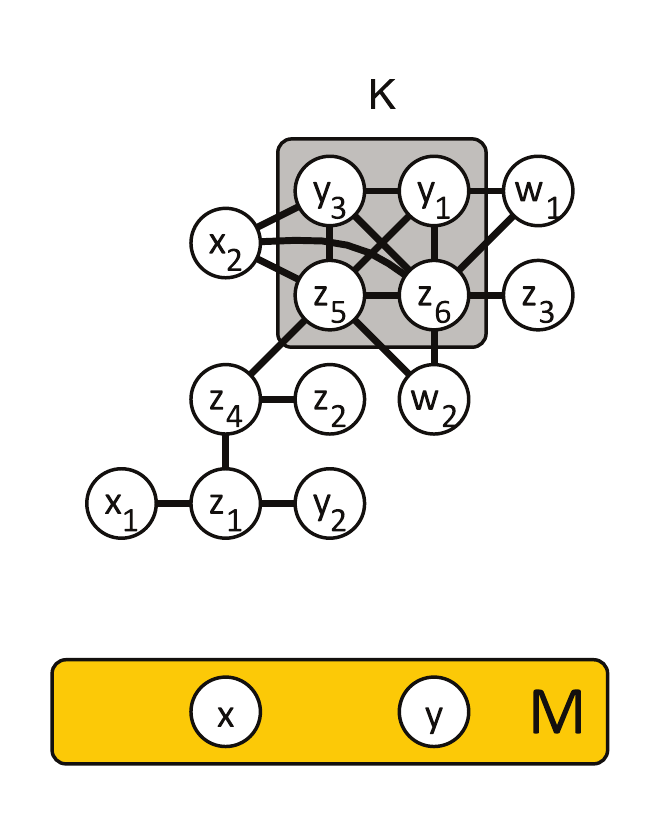}
}
\subfigure[Graph~$\Gc$, with~$G(\neg x)$ in black.]{\label{fig:cliquemarking3}
\image{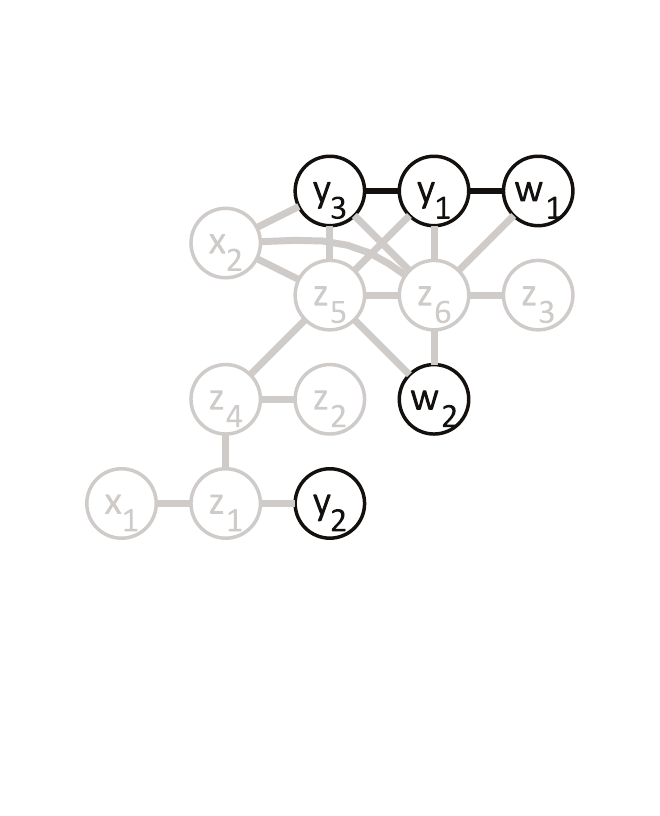}
}
\subfigure[Clique tree~$(T,\beta)$.]{\label{fig:cliquemarking1}
\bigimage{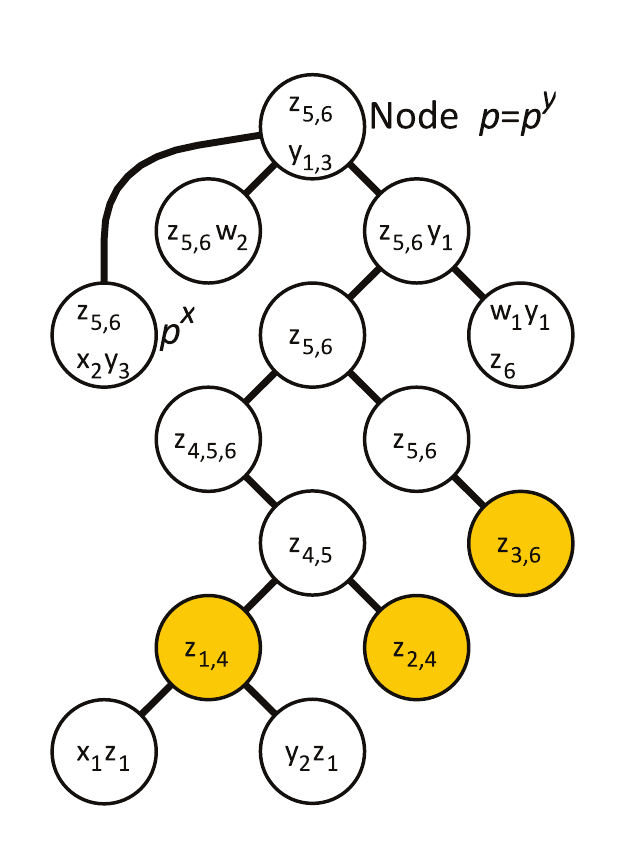}
}}{
\caption{Illustration for the clique marking procedure. \ref{fig:cliquemarking2} Instance of \AChVD with graph~$G$ and modulator~$M$, and a clique~$K$ in~$\Gc$. For readability, edges between~$\Gc$ and~$M$ are not drawn, but encoded through vertex names. Vertex~$x$ is adjacent to the~$x_i$'s and~$z_i$'s, vertex~$y$ is adjacent to the~$y_i$'s and~$z_i$'s. \ref{fig:cliquemarking3} The graph~$G(\neg x)$ contains the vertices in black and has a component with vertex set~$A^x = \{y_1, y_3, w_1\}$ that intersects~$K$. Since~$K \subseteq N_G(y)$ we have~$A^y = \emptyset$. \ref{fig:cliquemarking1} A clique tree of~$\Gc$ rooted at a node~$p$ whose bag contains the entire clique~$K$. The three nodes of~$Q^{xy}$ are yellow. A node~$p^x$ containing all vertices of~$A^x$ is chosen. When a letter such as~$z_{5,6}$ has multiple subscripts, this means the bag contains both~$z_5$ and~$z_6$.} \label{fig:cliquemarking}
}

% BMP: I commented this out because I moved this description to one section earlier, where we use it in the marking to reduce # components wrt separator.
%In what follows, by a phrase ``from a set $Z$ we mark up to $\ell$ vertices''
%we mean that we mark $\ell$ arbitrary vertices from $Z$ if $|Z| > \ell$,
%or all vertices of $Z$ if $|Z| \leq \ell$. 

The marking procedure for this rule works with additional constraints. By a statement of the form ``mark up to~$\ell$ vertices of~$Z$, preferring vertices maximizing/minimizing distance from $q$'' for some node $q \in V(T)$, we mean that we do not choose
the vertices of $Z$ to mark arbitrarily (in the case $|Z| > \ell$),
but we sort vertices $u \in Z$ according to the distance between
$\binv(u)$ and $q$ in~$T$, and mark $\ell$ vertices with the largest/smallest distance. We are now ready to present our marking algorithm.
\begin{enumerate}[(a)]
\item \label{mark:clique:triple} For every triple $(x_1,x_2,y) \in M\times M \times M$, mark up to $k+1$ vertices from $K \cap V(x_1,x_2,\neg y)$.
\item \label{mark:clique:Q} For every pair $(x_1,x_2) \in \binom{M}{2} \setminus E(G)$ and every node $q \in Q^{x_1x_2}$,
mark up to $k+1$ vertices of $K \cap V(x_1,x_2)$, preferring vertices
maximizing distance from $q$.
\item \label{mark:clique:px1} For every pair $(x,y) \in M \times M$,
mark up to $k+1$ vertices from $K \cap V(x,\neg y)$, preferring
vertices minimizing distance from $p^y$.
\item \label{mark:clique:px2} For every pair $(x,y) \in M \times M$,
mark up to $k+1$ vertices from $K \cap V(\neg x,\neg y)$,
preferring vertices minimizing distance from $p^{y}$.
\end{enumerate}
A direct calculation shows that we mark at most $(k+1)(|M|^3 + (k+3)|M|^2)$
vertices of $K$, and it is straightforward to implement the marking algorithm
to run in polynomial time. Figure~\ref{fig:CliqueMarkingProof} illustrates the different types of holes that are preserved by this marking procedure.

\subsection{Every unmarked vertex is irrelevant}

The remainder of this section is devoted to the proof that any unmarked vertex of $K$ is irrelevant.
For that we need the following simple technical step.
\begin{claim}\label{cl:cliques:towards}
Let $v \in K$ be an unmarked vertex, $y \in M$ be a vertex such that $vy \notin E(G)$, and let $P$ be a path between $v$ and $y$ with all internal vertices in $\Gc$.
Then, for every $x \in M$, and every vertex $v'$ marked for the pair
$(x,y)$ either in Point~\eqref{mark:clique:px1} if $vx \in E(G)$
or in Point~\eqref{mark:clique:px2} if $vx \notin E(G)$,
the vertex $v'$ is equal to 
or adjacent to at least one vertex of $V(P) \setminus \{v,y\}$.
\end{claim}
\begin{claimproof}
First, note that since in Points~\eqref{mark:clique:px1} and~\eqref{mark:clique:px2}
we prefer vertices minimizing distance to $p^y$, and the vertex $v$
is a candidate for marking when $v'$ has been marked, the distance
from $\binv(v')$ to $p^y$ is not larger than the distance from $\binv(v)$
to $p^y$.

Without loss of generality we may assume that $P$ is an induced path: shortcutting
the path $P$ to an induced one would only make our task harder.
Let $u$ be the neighbor of $y$ on the path $P$. 
We have $u \in V(\Gc)$; since $vy \notin E(G)$, we have also $u \neq v$.
Furthermore, as $P$ is an induced path, all vertices on $P$ except
for $y$ and $u$ lie outside $N[y]$, and thus they lie in $A^y$
as $v \in A^y$. Consequently, $u \in N(A^y)$, and thus $p^y \in \binv(u)$.

If $p^y \in \binv(v)$, then also $p^y \in \binv(v')$
and $v' = u$ or $v'u \in E(G)$.
Otherwise, let $e$ be the first edge of $T$ on the unique
path from $p$ to $p^y$ that is not contained in $T[\binv(v)]$,
and let $p'$ be the endpoint of $e$ that belongs to $\binv(v)$.
Since $P \setminus \{y\}$ induces a connected subgraph of $\Gc$
containing vertices $v \in \beta(p)$ and $u \in \beta(p^y)$, there exists
a vertex $w \in V(P) \setminus \{y\}$ with $w \in \adh(e)$. As~$v \not \in \adh(e)$ we know~$w \neq v$. 
Furthermore, we have $p' \in \binv(v')$, since $\binv(v')$ is at least as close to $p^y$ as $\binv(v)$,
and both $\binv(v)$ and $\binv(v')$ contain the root $p$. Hence, $w = v'$ or $wv' \in E(G)$.
This finishes the proof of the claim.
\end{claimproof}

We are now ready to prove the last claim,
concluding the proof of Lemma~\ref{lem:red-cliques}.
\begin{claim}\label{cl:cliques:irrelevant}
Every unmarked vertex of $K$ is irrelevant.
\end{claim}
\begin{claimproof}
Let $v \in K$ be an unmarked vertex.
It suffices to show that if $X$ is a solution to $(G-v,k,M,E^h)$,
then it is also a solution to $(G,k,M,E^h)$.
Assume the contrary; as $v \notin M$, it must hold that there is
a hole $H$ in $G-X$. Since $G-v-X$ is chordal, $v \in V(H)$.
By the properties of an \AChVD instance $(G,k,M,E^h)$, the hole $H$
contains at least two vertices of $M$.
We make a case distinction, depending on whether the neighbors of $v$
on the hole $H$ belong to $M$ or lie in $\Gc$ (see Fig.~\ref{fig:CliqueMarkingProof}).

\medskip 

\noindent\textbf{Case 1: both neighbors of $v$ on $H$ lie in $M$.}
Denote these two neighbors by $x_1,x_2$. As $H$ is a hole, $x_1x_2 \notin E(G)$; in particular, 
$x_1x_2 \notin E^h$.

Assume first that $x_1$ and $x_2$ are the only two vertices of $V(H) \cap M$.
Then $H$ consists of vertices $v$, $x_1$, $x_2$, and a path $P$
between $x_1$ and $x_2$ with all internal vertices in $\Gc$.
Consider the set $P_T$ of all nodes of $T$ whose bag contain at least
one internal vertex of $P$. Since $v \in K \subseteq \beta(p)$, 
and $H$ is a hole, we have $p \notin P_T$. Furthermore, by the connectivity
of $P$, $P_T$ induces a subtree of $T$. Let $q_0$ be the topmost (closest to $p$)
vertex of $P_T$; note that $q_0 \neq p$.

By the existence of the path $P$, the node $q_0$ is $x_1x_2$-good.
By the definition of $Q^{x_1x_2}$ there exists a node $q \in Q^{x_1x_2}$
that is a descendant of $q_0$.
Consider a vertex $v'$ marked at Point~\eqref{mark:clique:Q} for
the pair $x_1x_2$ and the node $q$ that does not belong to the solution $X$.
Since in Point~\eqref{mark:clique:Q} we mark vertices that maximize the distance
from $q$, the distance between $\binv(v')$ and $q$ is at least as large
as the distance between $\binv(v)$ and $q$.
Consequently, since the root $p$ belongs to $\binv(v')$ and to $\binv(v)$,
and $q_0$ is an ancestor of $q$, we have $q_0 \notin \binv(v')$.
We infer that $v'$ is not adjacent to any internal vertex of $P_T$,
and thus if we replace $v$ with $v'$ on the hole $H$, we obtain
a hole as well, contradicting the fact that $G-v-X$ is a chordal graph.

In the remaining case, there exists a third node $y \in V(H) \cap M$,
different from $x_1$ and $x_2$; in particular, $vy \notin E(G)$.
Consider a vertex $v'$ marked at Point~\eqref{mark:clique:triple}
for the triple $(x_1,x_2,y)$ that does not belong to $X$.
Consider the path $P := H-y$ between the neighbors of $y$ on the hole $H$,
and let $P'$ be the walk created from $P$ by replacing the vertex $v$
with $v'$. Since $v' \in N(x_1) \cap N(x_2)$, $P'$ is indeed a walk.
Since $v' \notin N(y)$, all internal vertices of $P'$ avoid $N[y]$.
Since no vertex of $P'$ belongs to $X$ nor equals $v$,
by Observation~\ref{obs:find-hole} the graph $G-v-X$ contains a hole,
a contradiction.

\widefigure{
\subfigure[$|N_H(v) \cap M| = 2$ and $|V(H) \cap M| = 2$.]{\label{fig:CliqueMarkingProofA}
\image{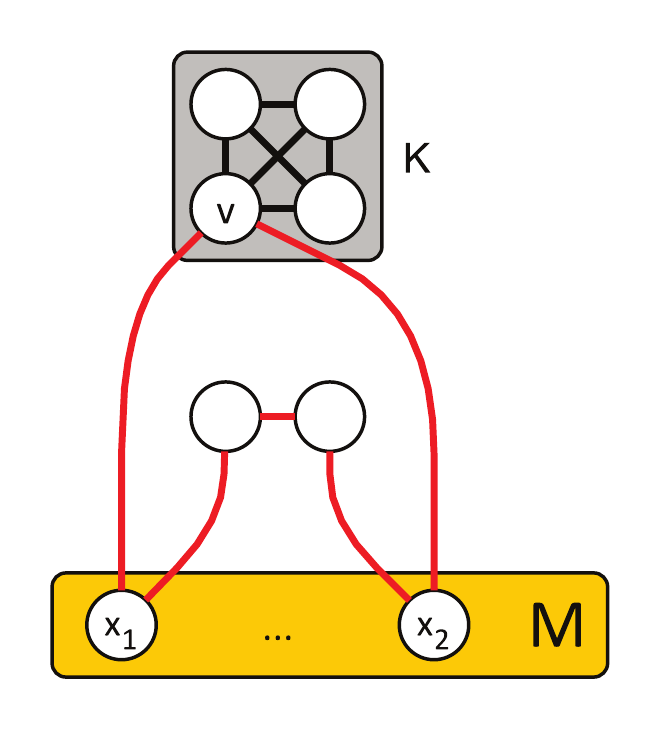}
}
\subfigure[$|N_H(v) \cap M| = 2$ and $|V(H) \cap M| > 2$.]{\label{fig:CliqueMarkingProofB}
\image{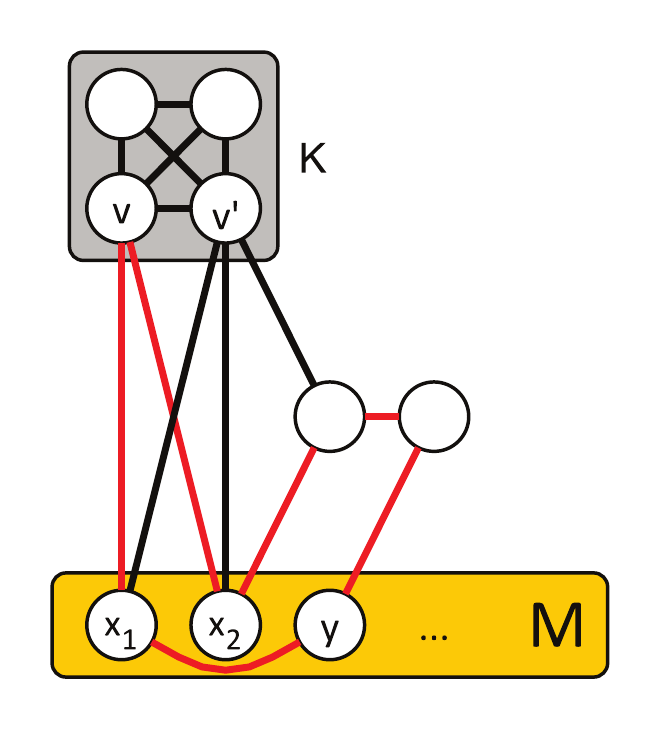}
}
\subfigure[$|N_H(v) \cap M| = 1$.]{\label{fig:CliqueMarkingProofC}
\image{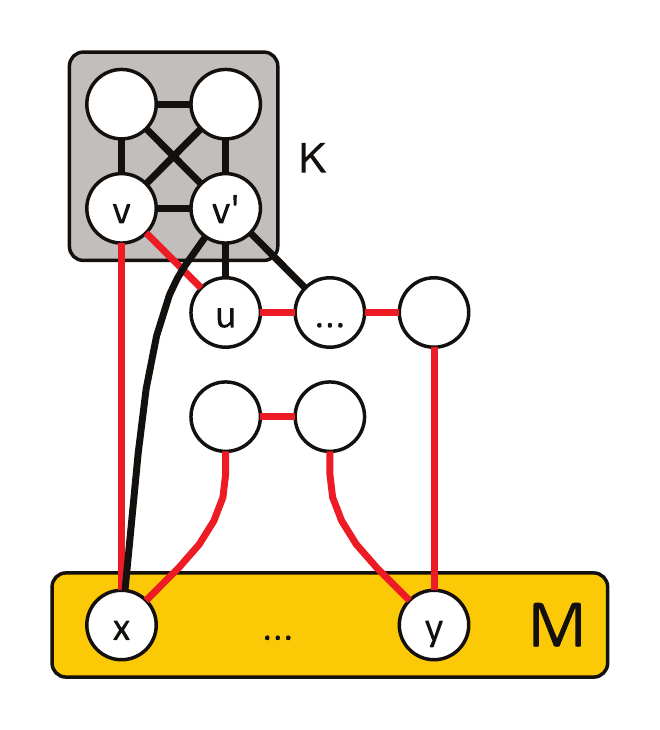}
}
\subfigure[$|N_H(v) \cap M| = 0$.]{\label{fig:CliqueMarkingProofD}
\image{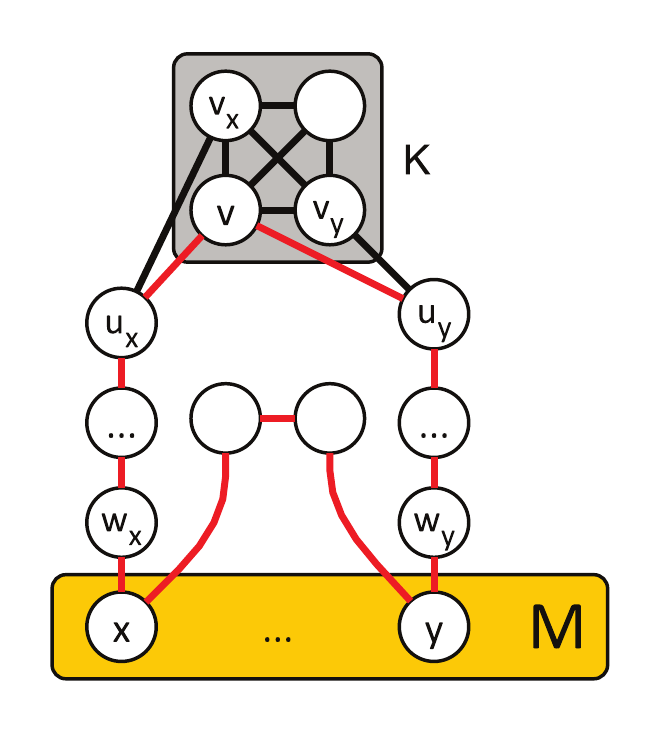}
}}{
\caption{Illustration for the correctness of the clique marking procedure (Claim~\ref{cl:cliques:irrelevant}). The four subfigures correspond to the four marking steps (\ref{mark:clique:triple})--(\ref{mark:clique:px2}). Each subfigure shows a schematic view of an instance with a modulator~$M$ and clique~$K$ in~$\Gc$ that is the target of the marking procedure. For each of the four types of holes highlighted in red, the corresponding marking step ensures that the constraints they impose on size-$k$ solutions are preserved after deleting the unmarked vertices. Vertex names correspond to those used in the proof of Claim~\ref{cl:cliques:irrelevant}.} \label{fig:CliqueMarkingProof}
}

\medskip

\noindent\textbf{Case 2: one neighbor of $v$ on $H$ lies in $M$.}
Let $u$ and $x$ be the two neighbors of $v$ on $H$, with $x \in M$
and $u \in V(\Gc)$. Let $y$ be the first vertex of $V(H) \cap M$
that we encounter if we traverse  $H$ starting from $v$ in the direction 
towards $u$. Since $V(H) \cap M$ contains at least two vertices,
$x \neq y$, in particular, $vy \notin E(G)$.
Let $P_u,P_x$ be the subpaths of $H$ between $v$ and $y$, with
$u \in V(P_u)$ and $x \in V(P_x)$.

Consider a vertex $v'$ marked at Point~\eqref{mark:clique:px1}
for the pair $(x,y)$ that does not belong to the solution $X$.
By Claim~\ref{cl:cliques:towards}, $(V(P_u) \setminus \{v\}) \cup \{v'\}$
induces a connected subgraph of $G$.
Since $v'y \notin E(G)$ and $v'x \in E(G)$,
$(V(P_x) \cup V(P_u) \setminus \{v\}) \cup \{v'\}$
contains a path from the two neighbors of $y$ on the hole $H$
whose internal vertices belong to $V(G) \setminus (N[y] \cup \{v\})$.
By Observation~\ref{obs:find-hole}, the graph $G-v-X$ contains a hole, a contradiction.

\medskip

\noindent\textbf{Case 3: none of the neighbors of $v$ on $H$ lies in $M$.}
Let $u_x$ and $u_y$ be the two neighbors of $v$ on $H$.
Let $x$ be the first vertex of $M$ that we encounter if we traverse
$H$ starting from $v$ in the direction towards $u_x$, and similarly
define $y$ for $u_y$. Since $V(H) \cap M$ contains at least two vertices,
$x \neq y$. Clearly, we have also $vx,vy \notin E(G)$.
Let $P_x$ be the subpath of $H$ between $v$ and $x$ containing $u_x$,
and let $w_x$ be the neighbor of $x$ on $P_x$.
Similarly define $P_y$ and $w_y$. Note that it is possible that $w_x=u_x$
or $w_y = u_y$.

Consider vertices $v_y$ and $v_x$ marked at Point~\eqref{mark:clique:px2}
for the pairs $(x,y)$ and $(y,x)$, respectively,
that do not belong to the solution $X$.
Note that it is possible that $v_y = v_x$.
By Claim~\ref{cl:cliques:towards}, both $(V(P_x) \setminus \{v\}) \cup \{v_x\}$
and $(V(P_y) \setminus \{v\}) \cup \{v_y\}$ induce connected graphs.
As both $v_x$ and $v_y$ belong to the clique $K$, it follows that
$(V(P_x) \cup V(P_y) \cup \{v_x,v_y\}) \setminus \{v\}$ induces a connected
subgraph of $G$; in particular, it contains a path from $w_x$ to $w_y$
whose internal vertices belong to $V(\neg x, \neg y)$.
This path, together with the subpath of $H$ between $x$ and $y$ not containing $v$,
 witnesses that there exists a walk in $G-v-X$ between two neighbors of 
$y$ on $H$ whose internal vertices do not belong to $N[y]$. 
By Observation~\ref{obs:find-hole}, $G-v-X$ contains a hole, a contradiction.
\end{claimproof}

Define~$\cliquebound := (k+1) (|M|^3 + (k+3)|M|^2) \leq \Oh(k|M|^3)$ as a shorthand for the bound of Lemma~\ref{lem:red-cliques}; note that instances with~$k \geq |M|$ trivially have a solution consisting of the entire set~$M$. Having proven Lemma~\ref{lem:red-cliques}, we may state the following reduction
rule.
\begin{reduction}\label{red:cliques}
Apply the algorithm of Lemma~\ref{lem:red-cliques} to any 
maximal clique of $\Gc$ whose size is larger than~$\cliquebound$. 
Depending on the outcome, either add the new forced pair to $E^h$ and to $E(G)$ or 
delete an irrelevant vertex from $G$.
\end{reduction}

Exhaustive application of Rule~\ref{red:cliques} ensures that the clique number~$\omega(\Gc)$ of the chordal graph~$\Gc$ is bounded by~$\cliquebound$. Since chordal graphs have perfect elimination orderings (cf.~\cite[\S 1.2]{BrandstadtLS99}), the number of maximal cliques is linear in the number of vertices. The maximal cliques can be identified efficiently from a clique tree, implying that Rule~\ref{red:cliques} can be applied exhaustively in polynomial time.

\section{Reducing the number of components of nonneighbors}\label{s:components}
The main goal in this section is to reduce the number of connected components of~$G(\neg x)$ for each~$x \in M$. Rather than reducing this number by deleting vertices, we will get rid of components of~$G(\neg x)$ by making them adjacent to~$x$. To ensure that this process does not change the answer of the instance, we have to understand how holes in~$G$ use such components. We therefore need to introduce some notation.

\begin{definition} \label{def:cnegxy}
Let~$x,y \in M$ with~$x \neq y$. Let~$\C_{\neg x}(y)$ be the connected components~$C$ of~$G(\neg x)$ such that~$V(C) \cap N_G(y) \neq \emptyset$ and~$N_{\Gc}(C) \setminus N_G(y) \neq \emptyset$.
\end{definition}

A component~$C$ of~$\C_{\neg x}(y)$ contains the interior vertices of an induced path~$P_{C}$ between~$y$ and a vertex~$v$ of~$N_{\Gc}(C) \setminus N_G(y)$, which has at least two edges. Observe that if we have two components~$C_1, C_2 \in \C_{\neg x}(y)$ for which~$v \in (N_{\Gc}(C_1) \cap N_{\Gc}(C_2)) \setminus N_G(y)$ is a common neighbor, then the concatenation of the two paths~$P_{C_1}$ and~$P_{C_2}$ is a hole through~$y$ and~$v$ (see Fig.~\ref{fig:ReducingNonneighborComponents0}). The key point here is that since~$C_1$ and~$C_2$ are different components of~$G(\neg x)$, no vertex of~$C_1$ is adjacent to a vertex of~$C_2$. The role of~$x$ in Definition~\ref{def:cnegxy} is to ensure this property.

Furthermore, we observe that such a hole, created from a concatenation of $P_{C_1}$ and $P_{C_2}$, contains only one vertex of $M$, namely $y$. 
This is a contradiction with the definition of an \AChVD instance.
Consequently, we have the following observation that is crucial to our goal of reducing the number of components of~$G(\neg x)$ for all~$x \in M$.
\begin{observation}\label{obs:forbidden}
For fixed vertices $x,y \in M$, the sets $\{N_{\Gc}(C) \setminus N_G(y)~|~C \in \C_{\neg x}(y)\}$ are pairwise disjoint.
\end{observation}

We are now ready to state our reduction rule.

\begin{reduction}\label{red:components}
Let~$(G,k,M,E^h)$ be exhaustively reduced under Rule~\ref{red:cliques}. First, exhaustively apply Reduction Template~\ref{template:toughness} to sets 
$$S_{\neg x,y} := M \cup \bigcup \{ N_G(C) \mid \mbox{$C$ is a component of~$G(\neg x)$ that contains a neighbor of~$y$}\},$$
for all~$x,y \in M$. Second, for every $x \in M$ do the following marking process.
\begin{enumerate}
	\item For all~$y \in M$ with~$xy \not \in E(G)$, mark up to~$(k+1) ((k+2) \cdot \cliquebound + |M|)^2 + 1$ components~$C$ of~$G(\neg x)$ that contain a neighbor of~$y$.\label{mark:redcomponents:nonadjacent}
	\item For each~$y \in M$ with~$xy \in E(G)$, mark up to~$k+1$ components~$C$ of~$\C_{\neg x}(y)$.\label{mark:redcomponents:relevantpaths}
	\item For each~$y_1, y_2 \in M$ with~$y_1 y_2 \not \in E(G)$, mark up to~$k+1$ components~$C$ of~$G(\neg x)$ for which~$V(C) \cap N_G(y_1) \neq \emptyset$ and~$V(C) \cap N_G(y_2) \neq \emptyset$.\label{mark:redcomponents:connectnonadj}
\end{enumerate}
If some component~$C$ of~$G(\neg x)$ is unmarked when the process finishes, then make~$x$ adjacent to all vertices of~$C$.
\end{reduction}

If~$G(\neg x)$ has at least~$(k+1)((k+2) \cdot \cliquebound + |M|)^2+(k+1)(|M| + |M|^2) + 2$ components, then Rule~\ref{red:components} will be triggered to make~$x$ adjacent to all vertices of one such component. Hence exhaustive application of the rule bounds the number of components of~$G(\neg x)$ by~$\Oh(k^5|M|^6)$.

\begin{lemma} \label{lem:nonneighbor:comps}
If Rule~\ref{red:components} reduces an instance~$(G,k,M,E^h)$ of \AChVD to $(G',k,M,E^h)$ by making~$x$ adjacent to all vertices of~$C$, then the instances are equivalent.
\end{lemma}
\begin{proof}
As adding edges to a nonchordal graph can make it chordal, and adding edges to a chordal graph can make it nonchordal, both directions of the equivalence are nontrivial.

($\Rightarrow$) Assume that~$(G,k,M,E^h)$ has a solution~$X$ of size at most~$k$, and assume for a contradiction that~$X$ is not a solution for the reduced instance. Since the instances have the same forced pairs, graph~$G - X$ is chordal but~$G' - X$ is not. Consider a hole~$H$ in~$G' - X$. Since~$H$ is not a hole in~$G - X$, while adding edges to a graph cannot make an existing cycle chordless, one of the edges that was added between~$x$ and~$C$ lies on~$H$. Let~$xv \in E(H)$ with~$v \in V(C)$ be such an edge. Orient~$H$ such that~$x$ is the predecessor of~$v$ on~$H$, and consider the successor~$y$ of~$v$. Since~$H$ is a hole,~$xy \not \in E(G')$ and therefore~$y \not \in V(C)$ since~$G'$ has all edges between~$x$ and~$V(C)$. Moreover,~$y \not \in N_G(C) \setminus M$ since~$N_G(C) \setminus M \subseteq N_G(x)$ as~$C$ is a component of~$G(\neg x)$. Hence~$y \in M$. Since~$C$ is an unmarked component of~$G(\neg x)$ that contains~$v \in N_G(y)$, in Point~(\ref{mark:redcomponents:nonadjacent}) we marked~$(k+1) ((k+2) \cdot \cliquebound + |M|)^2 + 1$ other components of~$G(\neg x)$ that contain a neighbor of~$y$. Consider the set~$\C'$ of components of~$G(\neg x)$ that contain a $G$-neighbor of~$y$ and define~$S := \bigcup _{C \in \C'} N_{\Gc}(C)$.

\begin{claim}\label{claim:addedges:independentset}
There is a subset~$S' \subseteq S$ of size~$k+2$ that is independent in~$G$.
\end{claim}
\begin{claimproof}
Proposition~\ref{proposition:independencenr:vs:cliquenr} implies~$\alpha(G[S]) \geq |S| / \omega(G[S])$. By the precondition that Rule~\ref{red:cliques} cannot be applied, we know~$\omega(G[S]) \leq \cliquebound$. It follows that if~$|S| \geq (k+2) \cdot \cliquebound$, an independent subset of size~$k+2$ exists.

Assume for a contradiction that~$|S| < (k+2) \cdot \cliquebound$. Let~$S^* := S \cup M$ and observe that each component of~$\C'$ occurs as a connected component of~$G - S^*$. Hence~$G - S^*$ has at least~$|\C'| \geq (k+1) ((k+2) \cdot \cliquebound + |M|)^2 + 1 > (k+1)|S^*|^2$ connected components. But this implies that applying Template~\ref{template:toughness} with the separator~$S^* = S_{\neg x, y}$ will remove an irrelevant component, contradicting the precondition that this template is exhaustively applied for~$S_{\neg x, y}$.
\end{claimproof}

\begin{claim}\label{claim:addedges:tworemain}
There are two distinct marked components~$C_1, C_2 \in \C'$ that avoid~$X$, along with distinct vertices~$u_1 \in N_G(C_1) \setminus (M \cup X)$ and~$u_2 \in N_G(C_2) \setminus (M \cup X)$ such that there are no edges in~$G$ between~$V(C_1) \cup \{u_1\}$ and~$V(C_2) \cup \{u_2\}$.
\end{claim}
\begin{claimproof}
Consider an independent subset~$S'$ of~$S$ of size~$k+2$, whose existence is guaranteed by Claim~\ref{claim:addedges:independentset}. By definition of~$S$, for each~$v \in S'$ there exists a component~$C_v \in \C'$ whose $G$-neighborhood contains~$v$. For~$v \neq v' \in S'$, the corresponding components~$C_v$ and~$C_{v'}$ are distinct. This follows from the fact that~$vv' \not \in E(G)$ as~$S'$ is independent in~$G$, together with the fact that for each component~$C$ of~$G(\neg x)$ the set~$N_{\Gc}(C)$ is a clique by Proposition~\ref{prop:NC-clique}. Moreover,~$v$ is not adjacent to any vertex in~$C_{v'}$, as otherwise~$v \in N_{\Gc}(C_{v'})$ would be adjacent to~$v' \in N_{\Gc}(C_{v'})$ as the neighborhood is a clique. Similarly,~$v'$ is not adjacent to any vertex in~$N_{\Gc}(C_v)$.

Consider the set~$S'$ of size~$k+2$. For each vertex~$u \in X$, discard~$v$ from~$S'$ if~$v=u$ or if~$u \in C_v \cap X$. As each vertex of~$X$ discards at most one vertex, two vertices~$u_1, u_2 \in S'$ exist that are not discarded. Let~$C_1$ and~$C_2$ in~$\C'$ be the corresponding components of~$G(\neg x)$ that contain a~$G$-neighbor of~$y$. By the discarding process we know that~$X$ contains no vertex of~$\{u_1, u_2\} \cup V(C_1) \cup V(C_2)$. The argumentation above shows that no vertex of~$\{u_1\} \cup V(C_1)$ is adjacent in~$G$ to a vertex of~$\{u_2\} \cup V(C_2)$. 
\end{claimproof}

\widefigure{
\subfigure[Observation~\ref{obs:forbidden}.]{\label{fig:ReducingNonneighborComponents0}
\image{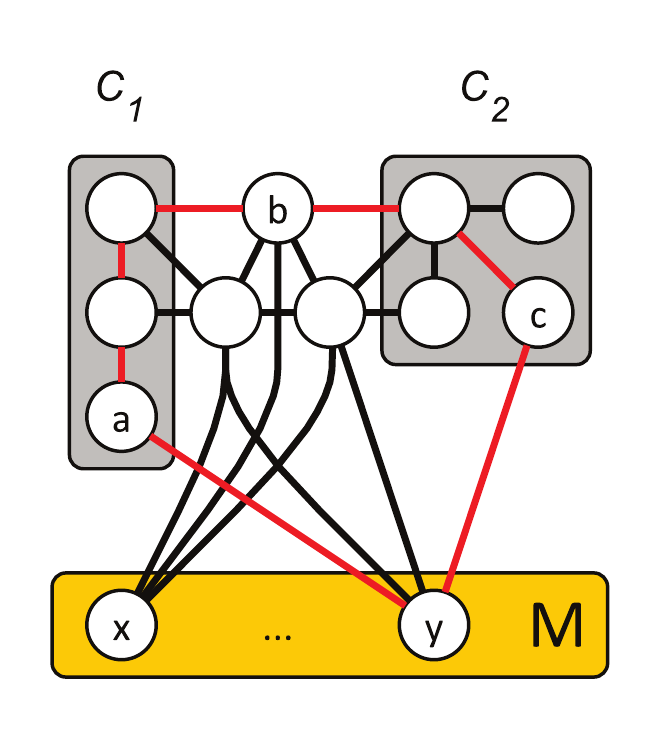}
}
\subfigure[$(\Rightarrow)$, Lemma~\ref{lem:nonneighbor:comps}.]{\label{fig:ReducingNonneighborComponents1}
\image{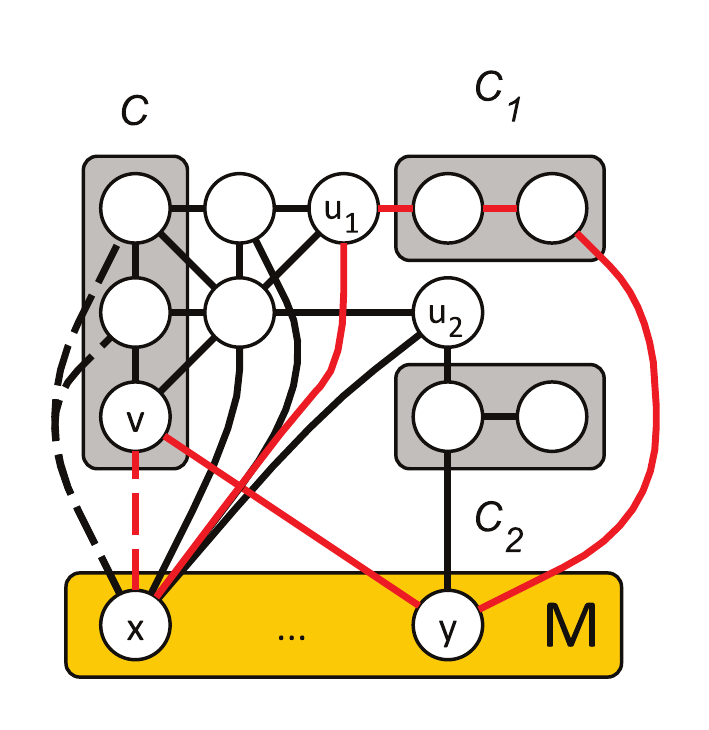}
}
\subfigure[$(\Leftarrow)$,~$|V(H) \cap M| = 2$.]{\label{fig:ReducingNonneighborComponents2}
\image{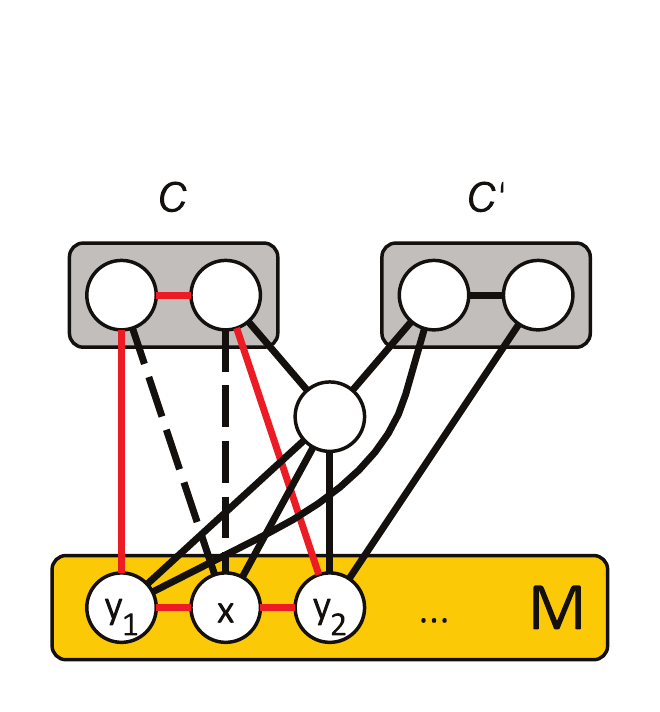}
}
\subfigure[$(\Leftarrow)$,~$|V(H)| \cap M| = 1$.]{\label{fig:ReducingNonneighborComponents3}
\image{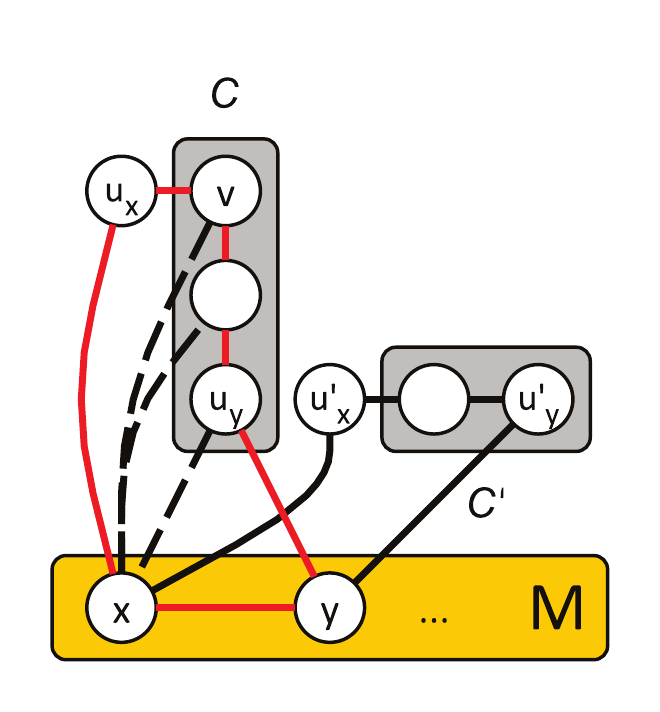}
}
}{
\caption{Illustrations for Section~\ref{s:components}. All figures show an instance of \AChVD with a graph~$G$ and modulator~$M$. Gray boxes highlight components of~$G(\neg x)$. Solid lines indicate edges in the graph. Red edges are used to highlight a hole. For the last three figures, dotted lines represent edges between~$v$ and~$C$ added by the reduction rule to obtain~$G'$. \ref{fig:ReducingNonneighborComponents0} Illustration of Observation~\ref{obs:forbidden}. Components~$C_1$ and~$C_2$ of~$G(\neg x)$ contain neighbors~$a$ and~$c$ of~$y$. The nonneighbor~$b$ of~$y$ that belongs to~$N_{\Gc}(C_1) \cap N_{\Gc}(C_2)$ implies the existence of the red hole using only one vertex from~$M$. \ref{fig:ReducingNonneighborComponents1} If adding edges between~$x$ and~$C$ creates a hole in~$G' - X$ (in red) that was not in~$G - X$, then a hole in~$G - X$ can be found through the marked components. \ref{fig:ReducingNonneighborComponents2} Case~1 of Lemma~\ref{lem:nonneighbor:comps}. If~$G' - X$ is chordal but~$G - X$ contains a hole through~$x$ (in red), then another hole can be found in~$G' - X$ through marked components. \ref{fig:ReducingNonneighborComponents3} Case~2 of Lemma~\ref{lem:nonneighbor:comps}.} \label{fig:ReducingNonneighborComponents}
}

Claim~\ref{claim:addedges:tworemain} allows us to derive a contradiction. Consider the marked components~$C_1, C_2$ with the corresponding vertices~$u_1$ and~$u_2$, whose existence is guaranteed by that claim. Since each~$C_i$ contains a neighbor of~$y$, while each~$u_i$ is in the $G$-neighborhood of a component of~$G(\neg x)$ and is therefore adjacent to~$x$, it follows that each~$C_i \cup \{u_i\}$ contains the interior vertices of an induced~$xy$-path in~$G$. As the claim guarantees that these paths are vertex-disjoint, not adjacent to each other, and avoid~$X$, their concatenation forms a hole in~$G-X$ (see Fig.~\ref{fig:ReducingNonneighborComponents1}). The hole avoids~$C$, since its interior uses marked components and~$C$ is unmarked. But then this is also a hole in~$G'-X$, a contradiction. This concludes the proof of the forward direction of the equivalence.

($\Leftarrow$) Assume that~$G'-X$ is chordal but~$G-X$ is not, and let~$H$ be a hole in~$G-X$. Since~$G-X$ is obtained from~$G'-X$ by removing all edges between~$v$ and~$V(C) \setminus X$, hole~$H$ goes through~$x$ and the edges from~$x$ to~$C$ form chords on~$H$ in~$G'$. Consider~$N_H[x]$ containing~$x$ and its two neighbors on~$H$. As~$H$ is a hole in~$G-X$, for all~$v \in V(H) \setminus N_H[x]$ we have~$xv \not \in E(G)$. 

\begin{claim} \label{claim:allbutneighbors:in:c}
All vertices of~$V(H) \setminus N_H[x]$ belong to~$C$.
\end{claim}
\begin{claimproof}
If~$v \in V(H) \setminus (N_H[x] \cup V(C))$, then~$xv \not \in E(G)$ implies~$xv \not \in E(G')$. The subpath of~$H$ connecting the neighbors of~$v$ on~$H$ forms a walk connecting~$N_{G'}(v)$ in~$G'-X$ whose interior avoids~$N_{G'}[v]$, implying by Observation~\ref{obs:find-hole} that~$G' - X$ has a hole. 
\end{claimproof}

We do a case distinction on where the neighbors of~$x$ on~$H$ appear in the graph.

\noindent\textbf{Case 1: both neighbors of $x$ on $H$ lie in $M$.} Denote these two neighbors by~$y_1, y_2$. As~$H$ is a hole in~$G-X$ we have~$y_1 y_2 \not \in E(G)$, and as~$G'$ and~$G$ differ only in edges between~$x$ and~$V(C)$ we have~$y_1 y_2 \not \in E(G')$. By Claim~\ref{claim:allbutneighbors:in:c}, all vertices of~$V(H) \setminus \{y_1, y_2, x\}$ belong to~$C$, showing that both~$y_1$ and~$y_2$ have a $G$-neighbor in~$C$. Hence~$C$ was a candidate for marking in Point~(\ref{mark:redcomponents:connectnonadj}), and as~$C$ was not marked there is a marked component~$C'$ of~$G(\neg x)$ containing a $G$-neighbor of~$y_1$ as well as a neighbor of~$y_2$, such that~$X$ avoids~$C'$. There is an induced path~$P$ connecting~$y_1$ to~$y_2$ such that the interior of~$P$ is contained in~$C'$. As the interior of~$P$ belongs to a component of~$G(\neg x)$, the interior is not adjacent to~$x$ in~$G$. The concatenation of~$P$ with~$(y_1, x, y_2)$ is therefore a hole in~$G - X$ (see Fig.~\ref{fig:ReducingNonneighborComponents2}). Since no vertex belongs to~$C$, this is also a hole in~$G' - X$; a contradiction.

\noindent\textbf{Case 2: one neighbor of $x$ on $H$ lies in $M$.} Let~$u_x$ and~$y$ be the two neighbors of~$x$ on~$H$, with~$y \in M$ and~$u_x \in V(\Gc)$. Let~$v$ be the neighbor of~$u_x$ on~$H$ that is not~$x$, and let~$u_y$ be the neighbor of~$y$ that is not~$x$. Since the hole contains~$(y,x,u_x)$ we know~$yu_x\not \in E(G)$. By Claim~\ref{claim:allbutneighbors:in:c} we have~$v \in V(C)$ and therefore~$u_x \in N_{\Gc}(C)$. Claim~\ref{claim:allbutneighbors:in:c} also ensures~$u_y \in V(C)$. 

So~$C$ is a component of~$G(\neg x)$ that contains a neighbor~$u_y$ of~$y$ and the $\Gc$-neighborhood of~$C$ contains a nonneighbor~$u_x$ of~$y$. It follows that~$C \in \C_{\neg x}(y)$ and hence $C$ was eligible for marking in Point~(\ref{mark:redcomponents:relevantpaths}), but was not marked because some other~$k + 1$ components of~$G(\neg x)$ were marked. Consider the set~$\C_M$ of marked components, and discard all components $C$ that contain a vertex of~$X$ either in $C$ or in $N_{\Gc}(C) \setminus N_G(y)$.
By Observation~\ref{obs:forbidden}, every vertex of~$X$ discards at most one
component of~$\C_M$, implying that some marked component~$C' \in \C_M$ is not discarded and~$C'$ contains a $G$-neighbor of~$y$.
Let~$u'_x$ be a vertex of~$N_{\Gc}(C') \setminus (N_G(y) \cup X)$, which exists by the discarding process, and note that this implies~$u'_x x \in E(G)$. There is an induced path from~$y$ to~$u'_x$ whose interior belongs to~$C'$, and therefore the path avoids~$X$. Since the interior belongs to~$C'$, which is a component of~$G(\neg x)$, the interior is not adjacent to~$x$. This path forms a hole in~$G - X$ together with~$x$ (see Fig.~\ref{fig:ReducingNonneighborComponents3}), and as the hole avoids~$C$ it is also a hole in~$G' - X$; a contradiction.

\noindent\textbf{Case 3: none of the neighbors of $x$ on $H$ lie in $M$.} This case is the easiest. Claim~\ref{claim:allbutneighbors:in:c} ensures that the vertices of~$H$ that are not~$x$ or its neighbors, belong to~$C \subseteq \Gc$. If the neighbors of~$x$ on~$H$ do not lie in~$M$, then they also lie in~$\Gc$ so that~$x$ is the only vertex from~$H$ that belongs to~$M$. But then~$G - (M \setminus \{x\})$ is not chordal. So~$(G,k,M,E^h)$ is not a valid instance of \AChVD; a contradiction.
\end{proof}

\section{Reducing a single component}\label{s:single-component}
\newcommand{\qtop}{q^\uparrow}
\newcommand{\qbot}{q^\downarrow}
\newcommand{\utop}{u^\uparrow}
\newcommand{\ubot}{u^\downarrow}
\newcommand{\rtop}{r^\uparrow}
\newcommand{\rbot}{r^\downarrow}
\newcommand{\wtop}{w^\uparrow}
\newcommand{\wbot}{w^\downarrow}
\newcommand{\Ptop}{P^\uparrow}
\newcommand{\Pbot}{P^\downarrow}
\newcommand{\sstop}{s^\uparrow}
\newcommand{\ssbot}{s^\downarrow}

In this section we assume that we are working
with an \AChVD instance $(G,k,M,E^h)$ for which none of the so-far defined reduction rules apply.
Recall that we denote $\Gc = G-M$.
We fix some clique tree $(T,\beta)$ of $\Gc$.
We mark a set $Q_0 \subseteq V(T)$ of important bags as follows.
\begin{enumerate}[1.]
\item For every pair $xy \in \binom{M}{2} \setminus E(G)$,
for every maximal clique of $G(x,y)$ we mark a node of $T$ that contains
this clique in its bag.
%Since Reduction~\ref{red:common-neighbours}
%is not applicable and $G(x,y)$ is a chordal graph, we can mark at most
%$k+1$ nodes of $T$ for every pair $xy$.
\item For every $x \in M$ and every connected component $A$ of
$G(\neg x)$, we mark a node of $T$ whose bag contains $N_{\Gc}(A)$;
recall that by Proposition~\ref{prop:NC-clique} the set
$N_{\Gc}(A)$ induces a clique in $\Gc$.
%Since Reduction~\ref{red:components} is inapplicable, 
%we mark at most TODO nodes of $T$ for ever $x \in M$.
\end{enumerate}
Let $Q$ be the set $Q_0$ together with the root of $T$
and all lowest common ancestors
in $T$ of every pair of nodes from $Q_0$. Standard arguments
show that $|Q| \leq 1+2|Q_0|$. We define $S_Q = \bigcup_{s \in Q} \beta(s)$.

\begin{reduction} \label{red:reduce:comp:wrt:q}
Apply Reduction Template~\ref{template:toughness} with the separator~$S_Q \cup M$.
\end{reduction}

\begin{proposition} \label{prop:count:components:after:rules}
If Rules~\ref{red:common-neighbours},~\ref{red:xy-good},~\ref{red:cliques},~\ref{red:components}, and~\ref{red:reduce:comp:wrt:q} have been applied exhaustively, then~$|S_Q \cup M| \in \Oh(k^6 |M|^{10})$ and~$G - (S_Q \cup M) = \Gc - S_Q$ has~$\Oh(k^{13} |M|^{20})$ connected components.%~$\Oh(k^{15} |M|^{18})$ connected components.
\end{proposition}
\begin{proof}
We first bound the size of~$Q_0$ and~$Q$. For every pair $xy \in \binom{M}{2} \setminus E(G)$, we add a bag to~$Q_0$ in the first step for each maximal clique of~$G(x,y)$. Since Rule~\ref{red:common-neighbours} is not applicable and~$xy \not \in E(G)$, we have~$\alpha(G(x,y)) < k + 2$. Proposition~\ref{proposition:independencenr:vs:cliquenr} shows that~$|V(x,y)| \leq \alpha(G(x,y)) \cdot \omega(G(x,y))$, and since Rule~\ref{red:cliques} is not applicable we know~$\omega(G(x,y)) \leq \omega(\Gc) \leq \cliquebound$. Hence~$|V(x,y)| \leq (k+2) \cdot \cliquebound$. Since the number of maximal cliques in the chordal graph~$G(x,y)$ does not exceed its order, for each nonadjacent pair~$x,y$ in~$M$ we add at most~$|V(x,y)|$ bags to~$Q_0$. Hence the first step contributes at most~$|M|^2(k+2) \cdot \cliquebound \in \Oh(k^2|M|^5)$ bags to~$|Q_0|$.

In the second step, we add a bag to~$Q_0$ for each connected component of~$G(\neg x)$ for each~$x \in M$. Since Rule~\ref{red:components} is not applicable, we know that each~$x \in M$ yields~$\Oh(k^5 |M|^6)$ connected components in~$G(\neg x)$. The contribution of the second step therefore dominates the size of~$Q_0$ and we have~$|Q_0| \in \Oh(k^5 |M|^7)$, which implies~$|Q| \in \Oh(k^5 |M|^7)$. % \todo{I see $k^5 |M|^7$ here. Am I right? Bart: Yes. (The bound in the previous paragraph also had $k$ and $M$ mixed up.)} 

Since each bag is a clique in~$\Gc$, it contains~$\cliquebound \in \Oh(k |M|^3)$ vertices since Rule~\ref{red:cliques} is exhaustively applied. It follows that~$|S_Q| \leq |Q| \cdot \cliquebound \in \Oh(k^6 |M|^{10})$ vertices, which also forms a size bound for the entire separator~$S_Q \cup M$. As Rule~\ref{red:reduce:comp:wrt:q} is exhaustively applied, it follows that the number of components of~$\Gc - S_Q$ is~$\Oh(k \cdot |S_Q \cup M|^2) \in \Oh(k^{13} |M|^{20})$. %~$\Oh(k |S_Q \cup M|^2) \in \Oh(k^{15} |M|^{18})$.
\end{proof}

As Proposition~\ref{prop:count:components:after:rules} shows that~$|S_Q \cup M|$ and the number of components of~$G - (S_Q \cup M)$ is bounded polynomially in the parameter, it only remains to bound the size of each component. Our main focus in this section is therefore to bound the size of a single connected component $A$ of $\Gc - S_Q$.
Let $Q^A = \bigcup_{v \in A} \binv(v)$; since $\Gc[A]$ is connected,
$T[Q^A]$ is connected as well.
By the construction of $Q$, there exists at most two nodes $q \in N_T(Q_A)$ such that if $p$ is the unique neighbor of $q$ in $Q_A$ then
$\beta(p) \cap \beta(q) \neq \emptyset$.
%\todo{I don't think this is necessarily true if~$\Gc$ is disconnected; see TeX comment.} 
% If \Gc is disconnected, then there is freedom in how to merge the cliquetrees for different components together into a cliquetree for the entire graph. If you take components of \Gc in which no bags are marked, then you can attach them to Q^A wherever you like without triggering new markings from LCA-closure. All these attachments give neighborhoods in the clique tree, causing the number of neighbors to become larger than 2. One way to fix this would be to say: ``if \Gc is disconnected, then take a clique tree of the following type: starting from cliquetrees of the individual components, pick a first one and root it somewhere, then attach the other trees to the root of the first one.'' However this is not very elegant. If all we need is eq:comp:NA, however, then we don't have to do this and can suffice by updating the claim that N_T(Q_A) has only 2 nodes of T to saying, for example, that there are only two edges sticking out of Q^A that have nonempty adhesions.
Furthermore, one of these nodes $q$ is the parent
of the topmost vertex of $Q^A$, and we denote it by $\qtop$.
If the second one exists, we denote it by $\qbot$; otherwise,
we add a dummy node $\qbot$ with an empty bag as a child of one of the
leaves of $T[Q^A]$.
Consequently, we have that
\begin{equation}\label{eq:comp:NA}
N_G(A) \subseteq M \cup \beta(\qtop) \cup \beta(\qbot).
\end{equation}
Let $\Ptd$ be the unique path from $\qtop$ to $\qbot$ in the tree $T$. Refer to Figure~\ref{fig:SingleComponentStructure} for an illustration of these concepts.
\begin{figure}[t]
\begin{center}
\subfigure[$(T,\beta)$.]{\label{fig:SingleComponentStructure1}
\bigimage{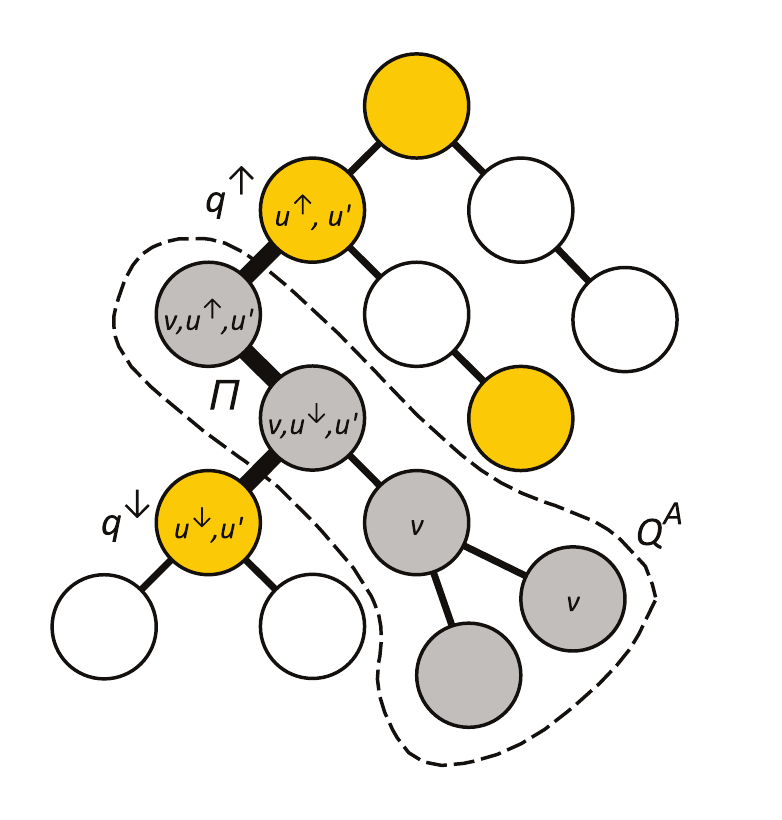}
}
\subfigure[$G$ with modulator~$M$.]{\label{fig:SingleComponentStructure2}
\image{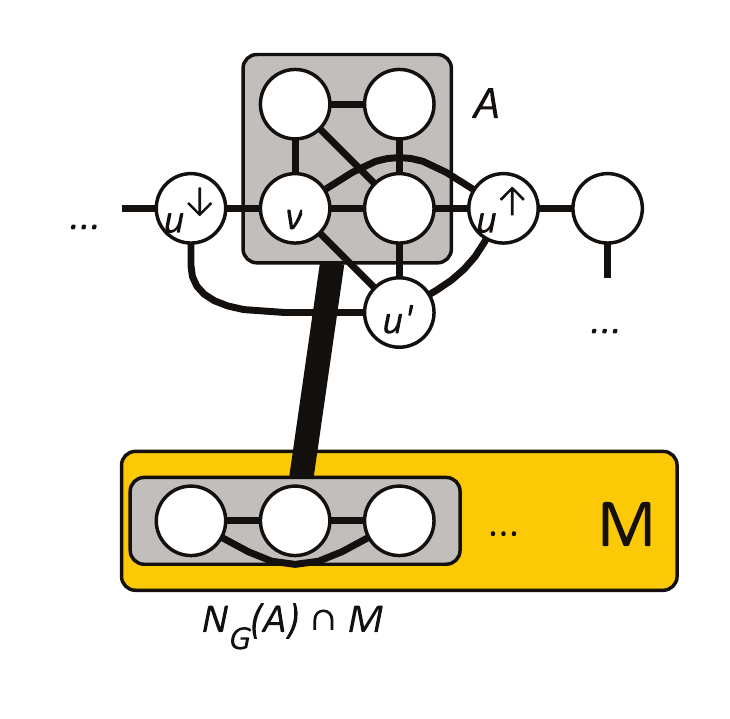}
}
\caption{Illustration of the structure of components~$A$ of~$\Gc - S_Q$ analyzed in Section~\ref{s:single-component}. \ref{fig:SingleComponentStructure1}~Clique tree~$(T,\beta)$ of the graph~$\Gc$. The marked bags~$Q$ are drawn in yellow; note their closure under taking least common ancestors. The bags~$Q^A$ that contain a vertex of~$A$ are drawn gray. The subtree induced by~$Q^A$ has two neighbors in~$T$,~$\qtop$ and~$\qbot$. The unique simple~$\qtop \qbot$-path~$\Ptd$ in~$T$ is highlighted with thick edges. \ref{fig:SingleComponentStructure2} Local structure in the corresponding graph~$G$ with modulator~$M$. The component~$A$ of~$\Gc-S_Q$ is visualized by a box. Its neighborhood in~$\Gc$ is a subset of the vertices in~$\beta(\qtop)$ and~$\beta(\qbot)$, both of which are cliques in~$\Gc$. These two cliques can overlap, as they do here in vertex~$u'$. By Observation~\ref{obs:comp:uniform}, all vertices in~$A$ have the same neighborhood in~$M$, which is visualized by a box. The thick edge between the two boxes represents all possible edges between the two sets. By Claim~\ref{cl:comp:permitted-holes}, any permitted hole~$H$ using a vertex of~$A$ uses~$A$ to form an induced path between some~$\utop \in \beta(\qtop)$ and~$\ubot \in \beta(\qbot)$. For any vertex~$v$ used on such a path, the subtree~$\binv(v)$ uses a vertex of~$\Pi$.} \label{fig:SingleComponentStructure}
\end{center}
\end{figure}
We start with the following simple observation.
\begin{observation}\label{obs:comp:uniform}
For every vertex $v \in A$ and two distinct neighbors $x \in N(v) \cap M$ and $u \in N(v)$ we have $xu \in E(G)$.
In particular, we have $N(v) \cap M = N(A) \cap M$ (i.e., every vertex in $A$ has exactly the same set of neighbors in $M$)
and $N(A) \cap M$ is a clique in $G$.
\end{observation}
\begin{proof}
First, assume the contrary of the first claim.
If $u \in M$, then $v \in V(x,u)$, but since $xu \notin E(G)$, $v$ belongs to $\bigcup_{q \in Q_0} \beta(q)$ by the first marking step. Hence~$v \in S_Q$ and cannot belong to a component~$A$ of~$\Gc - S_Q$; contradiction. 
If $ u \not \in M$, then $v \in N_{\Gc}(V(\neg x))$ since $u \in V(\neg x)$, and thus $v \in \bigcup_{q \in Q_0} \beta(q)$ by the second marking step.
The second claim follows from the connectivity of $A$.
\end{proof}

\subsection{Permitted holes through \texorpdfstring{$A$}{A}}

In this section we analyze permitted holes in $G$ that contain
some vertices from $A$.

\begin{claim}\label{cl:comp:permitted-holes}
Let $H$ be a permitted hole in $G$ that passes through $A$
and let $v \in V(H) \cap A$ be arbitrary.
Then:
\begin{enumerate}
\item both neighbors of $v$ on $H$ lie in $\Gc$;
\item $H$ contains a subpath between a vertex of $\beta(\qtop)$
and a vertex of $\beta(\qbot)$ that contains $v$ and whose internal
vertices all lie in $A$;
\item $\binv(v)$ contains at least one internal vertex of $\Ptd$.
\end{enumerate}
\end{claim}
\begin{claimproof}
The first claim follows from Observation~\ref{obs:comp:uniform}: if one of the neighbors of $v$ on $H$ lies in $M$, then
it is adjacent to the second neighbor of $v$ on $H$, a contradiction.

For the second claim, let $P$ be a maximal subpath of $H$ containing $v$
such that all internal vertices lie in $A$; recall that at least two vertices
of $H$ lie in $M$, so the endpoints of $P$ are distinct.
Furthermore, from the first claim, applied to all internal vertices of $P$,
we infer that the endpoints of $P$ lie in $\Gc$.
Since $N(A) \subseteq M \cup \beta(\qtop) \cup \beta(\qbot)$ by~\eqref{eq:comp:NA}, %\todo{Maybe $N(A) \subseteq M \cup \beta(\qtop) \cup \beta(\qbot)$ should be established explicitly above?}
they lie in $\beta(\qtop) \cup \beta(\qbot)$.
Since $H$ needs to contain at least two vertices of $M$, and $P$ does not
contain any such vertices but contains $v \in A$, the two endpoints of $P$
are nonadjacent, and thus one endpoint of $P$ lies in $\beta(\qtop)$
and one lies in $\beta(\qbot)$ as both sets are cliques.

For the last claim, let~$P'$ be a subpath of~$H$ between a vertex~$\utop \in \beta(\qtop)$ and~$\ubot \in \beta(\qbot)$ that contains~$v$, whose existence we just established. The minimal path~$\MPtd(\utop,\ubot)$ connecting~$\binv(\utop)$ and~$\binv(\ubot)$ in~$T$ is a subpath of~$\Pi$. Since~$P'$ is an induced path in~$\Gc$, Proposition~\ref{proposition:inducedpath:on:minimalpath} shows that~$\binv(v)$ contains a node of~$\MPtd(\utop,\ubot) \subseteq \Pi$. This must be an internal node of~$\Pi$, since~$v \in A$ does not occur in~$\beta(\qtop)$ or~$\beta(\qbot)$.
\end{claimproof}

\subsection{Important vertices in \texorpdfstring{$A$}{A}} \label{subsec:importantvertices}

We now define a set $Z$ of important vertices in $A$.
Intuitively, these are the only vertices that a ``reasonable'' solution
may delete in $A$.

We define by induction sets $Q_i \subseteq V(\Ptd)$ for $i \in \{0,1,2,\ldots\}$.
Let $Q_0 = \{\qtop, \qbot\}$. 
Given $Q_i$, we first define $Z_i = \bigcup_{q \in Q_i} \beta(q)$ and
then $Q_{i+1}$ as follows:
we start with $Q_{i+1} = Q_i$ and then, for every $u \in Z_i$
we add to $Q_{i+1}$ the node on $V(\Ptd) \cap \binv(u)$
that is closest to $\qtop$ and the node that is closest to $\qbot$.
In what follows we will work with sets $Q_0$, $Q_1$, and $Q_2$ (see Figure~\ref{fig:IrrelevantVertexMarking1}).

Furthermore, we construct a set of edges $R \subseteq E(\Ptd)$ as follows:
for every two nodes $q_1,q_2 \in Q_2$ that are consecutive
on $\Ptd$ (i.e., no node of $\Ptd$ between them belongs to $Q_2$),
we add to $R$ an edge $e$ of the subpath of $\Ptd$ between $q_1$ and $q_2$
that minimizes $|\adh(e) \cap A|$ (see Figure~\ref{fig:IrrelevantVertexMarking2}).

We define the set
$$Z := A \cap \left(Z_2 \cup \bigcup_{e \in R} \adh(e)\right).$$ %\todo{Don't we also want the adhesions of~$e \in R$ to be in~$Z$?}
Since every clique in $\Gc$ has size $\Oh(k|M|^3)$
(by Reduction~\ref{red:cliques}), we have
$|Q_i| = \Oh(k^i |M|^{3i})$ and $|Z_i| = \Oh(k^{i+1} |M|^{3i+3})$ for every fixed $i$,
and, consequently, $|Z| = \Oh(k^3|M|^9)$.

\widefigure{
\bigimage{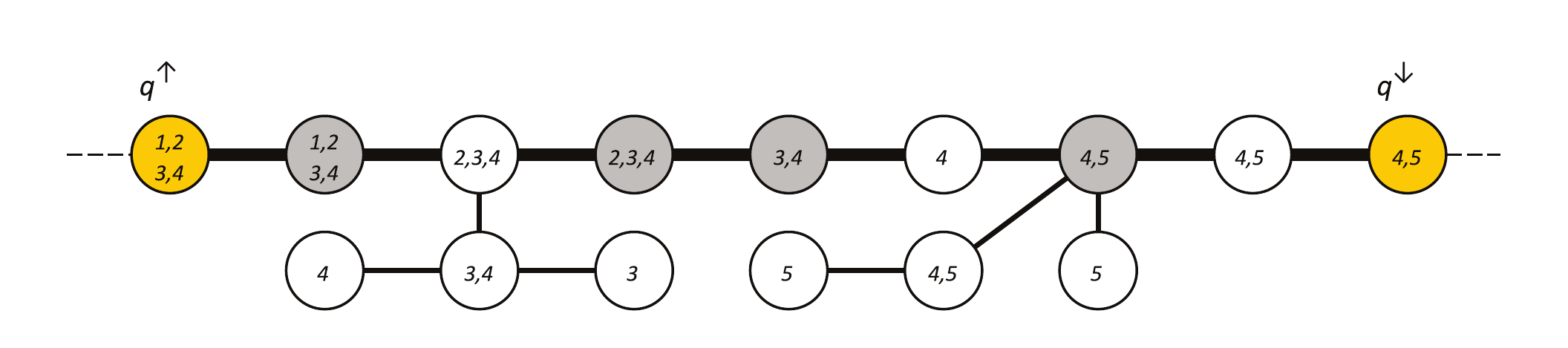}
}{
\caption{Illustration of the process of marking nodes on~$\Pi$ to determine sets~$Q_i \subseteq V(\Pi)$ and~$Z_i \subseteq V(A) \subseteq V(\Gc)$. Path~$\Pi$ in~$(T,\beta)$ laid out horizontally. Only the occurrence of the vertices in~$\beta(\qtop) \cup \beta(\qbot)$ in bags of~$T$ is shown. The bags also contain other vertices, which are hidden. Nodes~$Q_0 = \{\qtop,\qbot\}$ are drawn in yellow. The set~$Q_{i+1}$ contains~$Q_i$, and for each vertex~$v$ in a bag~$q \in Q_i$ it contains the leftmost and rightmost occurrence of~$v$ on~$\Pi$. The gray bags are added to~$Q_1$ by this process. For example, the first gray node from the left is added to~$Q_1$ because it has the rightmost occurrence of~$1 \in \beta(\qtop \in Q_0)$, i.e., the occurrence of~$1$ on~$\Pi$ that is closest to~$\qbot$ in~$T$.
} \label{fig:IrrelevantVertexMarking1}
}
\widefigure{
\bigimage{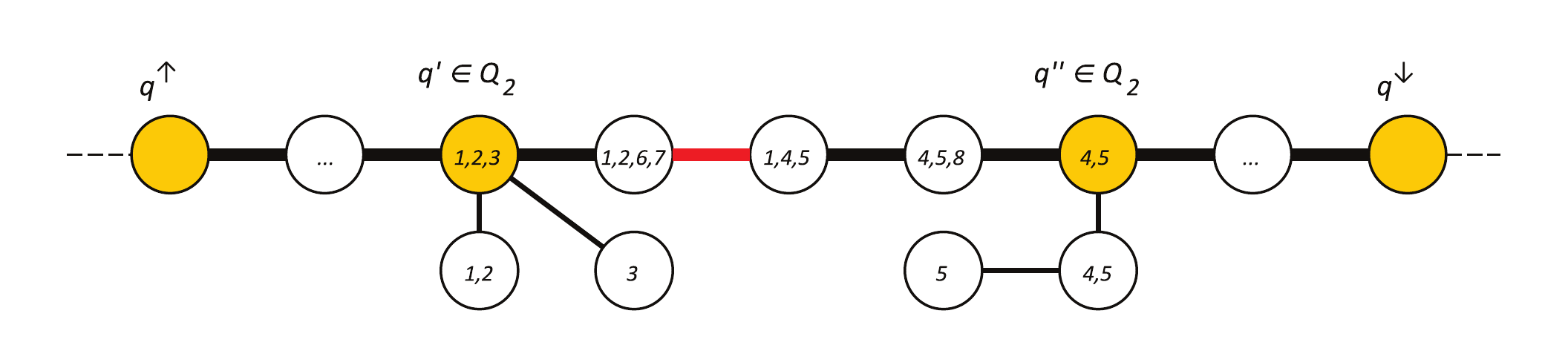}
}{
\caption{The set~$Q_2$ on the $\qtop \qbot$-path~$\Pi$ is drawn in yellow. For some bags on the path, the vertices from~$A$ they contain are drawn within their circles. Between consecutive nodes from~$Q_2$ we find an edge~$e$ minimizing~$|\adh(e) \cap A|$ and add it to~$R$. The red edge is added to~$R$, since its adhesion is~$\{1,2,6,7\} \cap \{1,4,5\} = \{1\} \in A$, while the other edges between~$q'$ and~$q''$ all have at least two vertices from~$A$ in their adhesion.} \label{fig:IrrelevantVertexMarking2}
}

An important property of the set $Q_1$ is the following.
\begin{claim}\label{cl:comp:v-not-in-Q1}
Let $H$ be a permitted hole containing a vertex $v \in A \setminus Z_1$.
Then both neighbors of the vertex $v$ on the hole $H$ lie in $A$ as well.
In particular, if $P$ is the path obtained from Claim~\ref{cl:comp:permitted-holes} for $H$ and $v$,
then $v$ is not adjacent to either of the endpoints of the path~$P$.
\end{claim}
\begin{claimproof}
We apply Claim~\ref{cl:comp:permitted-holes} to the hole $H$ with vertex $v$,
obtaining a path $P$ with endpoints $\utop \in \beta(\qtop)$ and $\ubot \in \beta(\qbot)$, and a node $s \in V(\Ptd) \cap \binv(v)$.
Note that $s \notin Q_1$.
Let $\sstop$ be the first node
of $Q_1$ encountered if we traverse $\Ptd$ from $s$ in the direction of $\qtop$,
and similarly define $\ssbot$ (recall that $\qtop,\qbot \in Q_1$, so these
nodes are well-defined).

The crucial observation, stemming from the definition of $Q_1$, is that the bag of every node of $V(\Ptd) \cap \binv(v)$
contains exactly the same subset of vertices of $Z_0$.
Furthermore, this subset of vertices forms a clique in $\Gc$, as they appear together in $\beta(s)$.
By Claim~\ref{cl:comp:permitted-holes}, each neighbor $w$ of $v$ on $P$ lies in $A \cup Z_0$.
We have that $\binv(w)$ contains a node of $V(\Ptd)$: this follows from Claim~\ref{cl:comp:permitted-holes} if $w \in A$
and for $w \in Z_0$ it follows from the facts that $w \in \beta(\qtop) \cup \beta(\qbot)$ and $w$ is adjacent to $v \in A$.
Furthermore, since $w$ is adjacent to $v$, we have that $\binv(w)$ actually needs to contain a vertex of $V(\Ptd) \cap \binv(v)$. % The claim that \binv(w) contains a nodle of \Ptd does not follow from Claim cl:comp:permitted-holes when w \in Z_0; but then it follows from the fact that w \in \qtop or \qbot and to make an induced path to the ``other side'' it must also occur somewhere in the middle, which gives a vertex on \Ptd. Should we clarify this?
Consequently, every neighbor $w \in N(v) \cap Z_0$ is equal or adjacent to both neighbors of $v$ on $P$,
which implies that $w$ cannot lie on $P$.
\end{claimproof}

\subsection{Vertices outside \texorpdfstring{$Z$}{Z} can be made undeletable}

Our goal now is to show the following lemma that says that
the vertices of $A \setminus Z$ can be made undeletable. In the following,~$\triangle$ denotes symmetric difference.
\begin{lemma}\label{lem:comp:Z}
For every solution $X$ to $(G,k,M,E^h)$ there exists a solution
$X'$ with $|X'| \leq |X|$, $X \triangle X' \subseteq A$,
and $X' \cap A \subseteq Z$.
\end{lemma}
\begin{proof}
Assume the contrary, and let $X$ be a counterexample
with minimum number of vertices from $A \setminus Z$.
Let $v \in X \cap (A \setminus Z)$. 
By the minimality of $X$, we have that $G-(X \setminus \{v\})$
contains a hole $H$ passing through $v$.
Since $X$ is a solution to $(G,k,M,E^h)$ and $v \notin M$, this hole is permitted.
We apply Claim~\ref{cl:comp:permitted-holes}: let $P$ be a subpath of
$H$ whose existence is asserted in the second claim, and let
$r \in \binv(v) \cap V(\Ptd)$.

Let $x$ be an arbitrary vertex of $V(H) \cap M$.
Note that by Claim~\ref{cl:comp:permitted-holes} we have $vx \notin E(G)$
and, consequently, by Observation~\ref{obs:comp:uniform},
$A \cap N(x) = \emptyset$.

Since $v \notin Z$ we have~$\binv(v) \cap Q_2 = \emptyset$, and in particular $r \notin Q_2$. Let $\rtop$ be the first vertex
of $Q_2$ encountered if we traverse $\Ptd$ from $r$ in the direction of $\qtop$,
and similarly define $\rbot$ (recall that $\qtop,\qbot \in Q_2$, so these
nodes are well-defined).
We claim the following:
\begin{claim}\label{cl:comp:X-cuts-edge}
$X$ contains $\adh(e_X) \cap A$
for some edge $e_X$ that lies on $\Ptd$ between $\rtop$ and $\rbot$.
\end{claim}
\begin{claimproof}
Assume the contrary.
Construct a set $B \subseteq A \setminus X$ as follows:
for every edge $e$ that lies on $\Ptd$ between $\rtop$ and $\rbot$,
insert into $B$ an arbitrary vertex of $\adh(e) \cap (A \setminus X)$.
Since every bag of $T$ is a clique, $G[B]$ is connected. 

Let $\Ptop$ be the subpath of $P$ between $v$ and the endpoint of $P$ in $\beta(\qtop)$
and $\Pbot$ be the subpath of $P$ between $v$ and the endpoint of $P$ in $\beta(\qbot)$.
Since $\rtop$ lies on the unique path in $T$ between $\binv(v)$
and $\qtop$, $\Ptop$ contains a vertex $\wtop \in \beta(\rtop)$;
note that $\wtop \neq v$ as $v \notin \beta(\rtop)$.
Similarly define $\wbot \in \beta(\rbot)$.
By the construction of $B$, $\wtop,\wbot \in N[B]$.

We infer that $G[V(H) \setminus \{v\} \cup B]$ contains a walk
between the neighbors of $x$ on the hole $H$ with internal vertices
in $V(G) \setminus N[x]$ (recall that $B \subseteq A$ and $A \cap N[x] = \emptyset$). By Observation~\ref{obs:find-hole}, $G-X$ contains a hole, which is easily seen to be permitted; a contradiction.
\end{claimproof}

By construction of~$R$ in the beginning of \S\ref{subsec:importantvertices}, there is an edge~$e$ on~$R$ that lies between $\rtop, \rbot \in Q_2$ on $\Ptd$. 
Let $X' = (X \setminus (\adh(e_X) \cap A)) \cup (\adh(e) \cap A)$,
where the edge $e_X$ is given by Claim~\ref{cl:comp:X-cuts-edge}.
By the choice of $e$ in the process of constructing $R$, we have
$|X'| \leq |X|$. Clearly, also $X \triangle X' \subseteq A$.
Since~$\adh(e \in R) \cap A \subseteq Z$, to prove that $X'$ contains strictly less vertices of $A \setminus Z$ than $X$
it suffices to show the following.
\begin{claim}
If $e_X$ is an edge whose existence is asserted by Claim~\ref{cl:comp:X-cuts-edge},
then $v \in \adh(e_X)$.
\end{claim}
\begin{claimproof}
By Claim~\ref{cl:comp:v-not-in-Q1} for $H$ and $v$, neither $\utop$ nor $\ubot$ is adjacent to $v$, in particular, they do not belong to $\beta(r)$.
By the definition of $Q_1$, neither of them belongs to any bag on $\Ptd$ between $\rtop$ and $\rbot$, exclusive.

Consequently, $\utop,\ubot \notin \adh(e_X)$ and thus
$P$ contains a vertex of $\adh(e_X) \cap A$. Since $H$ is a hole
in $G-(X \setminus \{v\})$, and $\adh(e_X) \cap A \subseteq X$,
this vertex is the vertex $v$.
\end{claimproof}

\widefigure{
\bigimage{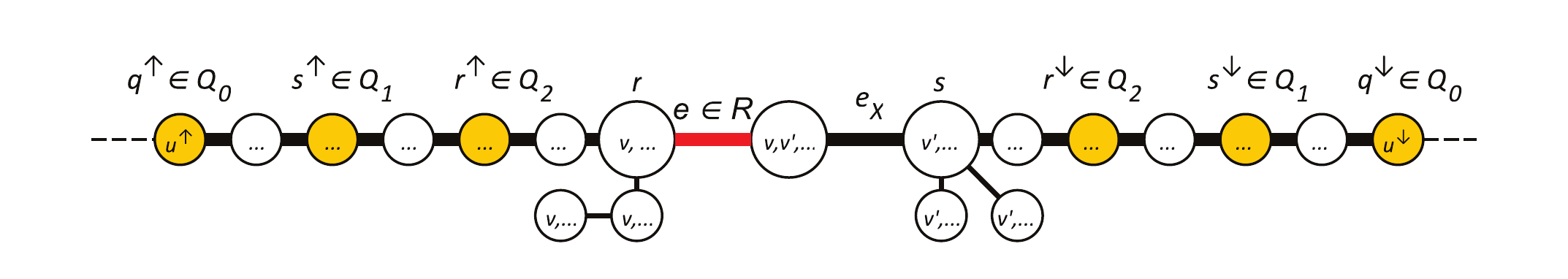}
}{
\caption{Illustration of the situation for Claim~\ref{cl:comp:P1-ends}. The set~$Q_2 \supseteq Q_1 \supseteq Q_0$ on the $\qtop \qbot$-path~$\Pi$ is drawn in yellow. A marked edge~$e \in R$ between~$\rtop$ and~$\rbot$ is drawn in red. An edge~$e_X$ for which~$\adh(e_X) \cap A \subseteq X$ is also shown.} \label{fig:IrrelevantVertexMarking3}
}

To reach a contradiction, it remains to show that $X'$ is a solution to $(G,k,M,E^h)$ as well.
Since $X \triangle X' \subseteq A$, $X'$ satisfies all the constraints
imposed by the set $E^h$. Thus, it remains to prove that $G-X'$ is chordal.
Assume the contrary, and let $H'$ be a hole in $G-X'$.

Since $G-X$ is chordal, $H'$ needs to contain a vertex $v' \in A \cap (\adh(e_X) \setminus \adh(e))$. Since $e$ and $e_X$ are edges of $\Ptd$ between two consecutive
vertices of $Q_2$ (namely, $\rtop$ and $\rbot$), 
it follows from the definition of $Q_2$ that $\adh(e) \cap Z_1 = \adh(e_X) \cap Z_1$.
We infer that $v' \notin Z_1$. Refer to Figure~\ref{fig:IrrelevantVertexMarking3} for an illustration.

We apply Claim~\ref{cl:comp:permitted-holes} to the hole $H'$ with vertex $v'$,
obtaining a path $P'$ with endpoints $\utop \in \beta(\qtop)$ and $\ubot \in \beta(\qbot)$, and a node $s \in V(\Ptd) \cap \binv(v')$.
Note that $s \notin Q_1$.
Let $\sstop$ be the first node
of $Q_1$ encountered if we traverse $\Ptd$ from $s$ in the direction of $\qtop$,
and similarly define $\ssbot$ (recall that $\qtop,\qbot \in Q_1$, so these
nodes are well-defined).
Observe that since $v' \in \adh(e_X)$, but $\sstop,\ssbot \notin \binv(v')$,
we have that $\rtop$, $\rbot$, $e$, and $e_X$ lie on the subpath of $\Ptd$ between $\sstop$ and $\ssbot$; possibly~$\rtop = \sstop$ or~$\rbot = \ssbot$.
We now show the following corollary of Claim~\ref{cl:comp:v-not-in-Q1}.
\begin{claim}\label{cl:comp:P1-ends}
The vertex $\utop$ is not present in any bag on $\Ptd$ between $\sstop$ (exclusive) and $\qbot$ (inclusive), i.e., $\binv(\utop) \cap V(\Ptd)$
is contained in the subpath of $\Ptd$ between $\sstop$ and $\qtop$.
Symmetrically,
the vertex $\ubot$ is not present in any bag on $\Ptd$ between $\ssbot$ (exclusive) and $\qtop$ (inclusive), i.e., $\binv(\ubot) \cap V(\Ptd)$
is contained in the subpath of $\Ptd$ between $\ssbot$ and $\qbot$.
\end{claim}
\begin{claimproof}
By Claim~\ref{cl:comp:v-not-in-Q1} for $H'$ and $v' \not \in Z_1$, neither $\utop$ nor $\ubot$ is adjacent to $v'$, in particular, they do not belong to $\beta(s)$.
By the definition of $Q_1$, neither of them belongs to any bag on $\Ptd$ between $\sstop$ and $\ssbot$, exclusive. 
The statement for $\utop$ follows from the fact that $\binv(\utop) \cap V(\Ptd)$ is connected and contains $\qtop$, and similarly for $\ubot$.
\end{claimproof}
Since all internal vertices of $P'$ lie in $A$, from Claim~\ref{cl:comp:P1-ends} it follows that $P'$ needs to contain a vertex
of $\adh(e) \cap A$, as $e$ lies between $\sstop$ and $\ssbot$ on $\Ptd$. Consequently, $P'$ is hit by $X'$, a contradiction.
This concludes the proof of the lemma.
\end{proof}

\subsection{Reducing \texorpdfstring{$A \setminus Z$}{A - Z}}

Our goal now is to reduce $A \setminus Z$. While by Lemma~\ref{lem:comp:Z} we know that in a yes-instance there is a solution
disjoint with $A \setminus Z$, we need to be careful as some permitted holes may contain vertices from $A \setminus Z$.
We first filter out a simple case.

\begin{observation}\label{obs:comp:v-outside-Ptd}
Any vertex of $A \setminus \bigcup_{q \in V(\Ptd)} \beta(q)$ is irrelevant.
\end{observation}
\begin{proof}
Let $v \in A \setminus \bigcup_{q \in V(\Ptd)} \beta(q)$. By Claim~\ref{cl:comp:permitted-holes},
no permitted hole passes through $v$. Consequently, any solution $X$ to $(G-v,k,M,E^h)$ is a solution
to $(G,k,M,E^h)$ as well: as $v \notin M$, $X$ satisfies all constraints imposed by $E^h$, while $G$ and $G-v$
contain the same permitted holes.
\end{proof}
This allows us to formulate the following reduction rule.
\begin{reduction}\label{red:comp:irrelevant}
If $A \setminus \bigcup_{q \in V(\Ptd)} \beta(q)$ is nonempty, remove any vertex from this set.
\end{reduction}

Other vertices of $A \setminus Z$ need to be handled differently.
For a vertex $v \in V(G)$, by \emph{bypassing the vertex $v$} we mean the following operation: we first turn $N_G(v)$ into a clique, and then delete $v$
from $G$. We first observe that bypassing a vertex preserves being chordal.
\begin{observation}\label{obs:bypass-chordal}
A graph constructed from a chordal graph by bypassing a vertex is chordal as well.
\end{observation}
\begin{proof}
Let $G_1$ be a chordal graph with clique tree $(T_1,\beta_1)$, let $v \in V(G_1)$, and let $G_2$ be constructed from $G_1$ by bypassing $v$.
Consider the following tree decomposition $(T_2,\beta_2)$: we take $T_2 = T_1$ and then, for every $p \in V(T_2)$
we put $\beta_2(p) = N(v)$ if $v \in \beta_1(p)$ and $\beta_2(p) = \beta_1(p)$ otherwise.
A direct check shows that $(T_2,\beta_2)$ is a tree decomposition of $G_2$ in which every bag induces a clique.
Thus, $G_2$ is chordal, as claimed.
\end{proof}

\begin{lemma}\label{lem:comp:A-Z}
Assume all previous reductions are inapplicable.
Let $v \in A \setminus Z$ and let $G'$ be constructed from $G$ by first turning $N_G(v)$ into a clique, and then removing the vertex $v$.
Then the instances $(G,k,M,E^h)$ and $(G',k,M,E^h)$ are equivalent.
\end{lemma}
\begin{proof}
In one direction, assume $(G,k,M,E^h)$ is a yes-instance.
Lemma~\ref{lem:comp:Z} asserts that there exists a solution $X$ that does not contain $v$.
Observation~\ref{obs:bypass-chordal}, applied to $G-X$ and vertex $v$ implies that $G'-X$ is chordal as well.
Consequently, $X$ is a solution to $(G',k,M,E^h)$ as well.

The other direction is significantly more involved.
Intuitively, we want to rely on Claim~\ref{cl:comp:v-not-in-Q1} to show that $v$ can only interact with permitted holes
``locally'', and thus any permitted hole in $G$ has its counterpart in $G'$: possibly shorter but still being a hole.
We start with showing that the new edges in $G'$ are actually only in $G[A]$.
\begin{claim}\label{cl:comp:cliquify-A}
If $u,w$ are two distinct nonadjacent neighbors of $v$ in $G$, then $u,w \in A \setminus Z_1$.
\end{claim}
\begin{claimproof}
Pick $u,w$ as in the statement. Observation~\ref{obs:comp:uniform} implies that $u,w \notin M$.
Since Reduction~\ref{red:comp:irrelevant} is not applicable to $v$, there exists $r \in \binv(v) \cap V(\Ptd)$.
Since $v \notin Z$, we have $r \notin Q_2$. Let $\rtop$ be the first vertex
of $Q_2$ encountered if we traverse $\Ptd$ from $r$ in the direction of $\qtop$,
and similarly define $\rbot$ (recall that $\qtop,\qbot \in Q_2$, so these
nodes are well-defined).

By contradiction, assume $u \in Z_1$.
Note that all nodes of $\binv(v) \cap V(\Ptd)$ lie on $\Ptd$ between $\rtop$ and $\rbot$, exclusive.
By the definition of $Q_2$, every the bag at every node of $\binv(v) \cap V(\Ptd)$ contains exactly the
same subset of $Z_1$; in particular, it contains $u$.
Since Reduction~\ref{red:comp:irrelevant} does not apply to $w$, $w$ is present in some bag on $\Ptd$;
since $vw \in E(G)$, $w$ is present in some bag of $\binv(v) \cap V(\Ptd)$.
We infer that $u$ and $w$ meet in some bag of $\binv(v) \cap V(\Ptd)$, a contradiction to the assumption
that they are nonadjacent.
\end{claimproof}

Assume now $(G',k,M,E^h)$ is a yes-instance, and let $X$ be a solution. We claim that $X$ is a solution to $(G,k,M,E^h)$ as well;
clearly it satisfies all the constraints imposed by $E^h$.
By contradiction, let $H$ be a hole in $G-X$; clearly, $H$ is permitted.
Since $H$ is not a hole in $G'-X$, it contains either the vertex $v$ or one of the edges of $E(G') \setminus E(G)$
is a chord of $H$ in $G'-X$. We infer that in both cases $H$ contains a vertex $w \in A \setminus Z_1$:
in the first case we take $w=v$, in the second case it follows from Claim~\ref{cl:comp:cliquify-A}.

We apply Claim~\ref{cl:comp:permitted-holes} to the hole $H$ and the vertex $w$, obtaining
a path $P$ with endpoints $\utop \in \beta(\qtop)$ and $\ubot \in \beta(\qbot)$. By the construction of $G'$, 
the subgraph $G'[V(P) \setminus \{v\}]$ is connected. By Claim~\ref{cl:comp:v-not-in-Q1}, the vertex $w$
is nonadjacent to neither $\utop$ nor $\ubot$, thus $\utop\ubot \notin E(G')$
and the shortest path $P'$ between $\utop$ and $\ubot$ in $G'[V(P) \setminus \{v\}]$ is of length at least two.
We infer that if we replace $P$ with $P'$ on the hole $H$, we obtain a hole in $G'-X$, a contradiction.
\end{proof}

Lemma~\ref{lem:comp:A-Z} justifies the following reduction rule.
\begin{reduction}\label{red:comp:cliquify}
If there exists a vertex $v \in A \setminus Z$, bypass then $v$.
\end{reduction}

To conclude, observe that if the reductions introduced in this section are inapplicable,
we have $|A|=\Oh(k^3|M|^9)$ for every connected component $A$ of $G-(M \cup S_Q) = \Gc - S_Q$. Since Proposition~\ref{prop:count:components:after:rules} shows that the number of such components is bounded by~$\Oh(k^{13} |M|^{20})$ after exhaustive applications of all the reduction rules and~$|S_Q| \in \Oh(k^6 |M|^{10})$, we obtain a final size bound of~$\Oh(k^{16} |M|^{29})$ vertices in an exhaustively reduced instance. All reduction rules can be applied in polynomial time. Each application either reduces the number of vertices, or adds an edge to the graph. It follows that an instance can be exhaustively reduced by~$\Oh(|V(G)| + |E(G)|)$ applications of a reduction rule, which means the entire kernelization runs in polynomial time. This concludes the proof of Theorem~\ref{thm:AChVD-kernel}.

\section{Approximation algorithm}\label{s:approximation}
This section is devoted to the proof of Theorem~\ref{thm:apx}. We start with an informal overview. Let $(G,k)$ be a \ChVD instance.
First, we can assume that $\log n \leq k \log k$, as otherwise
the algorithm of Cao and Marx~\cite{CaoM15} solves $(G,k)$ in polynomial time.

We observe that in a yes-instance there exists a balanced vertex separator
that consists of a clique and at most $k$ vertices.
To this end, it suffices to take
a minimum solution and a ``central'' maximal clique of the remaining chordal graph.
We can find an approximate variant of such a balanced vertex cut in the following fashion.
Since we can assume that $G$ is $C_4$-free (e.g., by greedily deleting
whole $C_4$'s in $G$), it has $\Oh(n^2)$ maximal cliques~\cite[Proposition 2]{Farber89a}. We iterate over
all choices of a maximal clique $K$ in $G$, and apply the approximation
algorithm for balanced vertex cut of~\cite{FeigeHL08} to $G-K$.
In this manner, we can always find a balanced vertex cut in $G$
consisting of a clique and $\Oh(k \sqrt{\log k})$ vertices.

We iteratively apply the above partitioning algorithm to a connected component
of $G$ that is not chordal, and remove the output cut.
A simple charging argument shows that in a yes-instance
such a procedure finishes after $\ell = \Oh(k \log n)$ steps.
In the end, we obtain a partition
$$V(G) = A_0 \uplus K_1 \uplus K_2 \uplus \ldots \uplus K_\ell \uplus X_0,$$
where $G[A_0]$ is chordal, every $K_i$ is a clique, and $|X_0| = \Oh(\ell k \sqrt{\log k})$.
We put $X_0$ into the solution we are constructing.
To obtain the desired approximation algorithm, it is sufficient to tackle
the special case where we are given a partition $V(G) = A \uplus B$,
where $G[A]$ is chordal and $B$ is a clique in $G$. We can initially apply an algorithm for this special case 
to the graph~$G[A_0 \uplus K_1]$, obtaining a solution~$X_1$. We apply it again to~$G[((A_0 \uplus K_1) - X_1) \uplus K_2]$ to obtain~$X_2$, and then we repeat.
After~$\ell$ iterations of finding an approximate solution and adding back the next clique, we have found an approximate solution to the original input.

To find an approximate solution in the mentioned special case, we again apply the balanced partitioning approach, but with a 
different toolbox. Let $L$ be a ``central'' maximal clique of $G[A]$, 
that is, a maximal clique such that every connected component of $G[A \setminus L]$
is of size at most $|A|/2$.
If we find a small (bounded polynomially in the value of an LP relaxation)
set of vertices
that hits all holes passing through $L$,
then we can delete such a set, discard $L$ (as no holes now pass through $L$)
and recurse on connected components of remaining graph.

Thus, we can now focus only on holes passing through $L$.
We obtain an orientation $\Gdown$ of $G$ as follows:
we root a clique tree of $G[A]$ in $L$, and orient every edge of $G$ ``downwards'',
that is, $uv \in E(G)$ becomes $(u,v) \in E(\Gdown)$
if $\top(u)$ is an ancestor of $\top(v)$,
breaking ties arbitrarily for $\top(u) = \top(v)$. 
We argue that the task of hitting all holes passing through $L$ can be reduced
to a \multicut instance in $\Gdown$, and we apply an approximation algorithm
of Gupta~\cite{Gupta03} to find a small solution.

The rest of this section is organized as follows.
First, in Section~\ref{s:apx:defs} we recall the definition
of the \multicut problem, and formally define LP relaxations
of \multicut and \ChVD. 
We also recall the approximation algorithm of Gupta~\cite{Gupta03}
that we will use.
In Section~\ref{s:apx:clique} we focus on approximating
\ChVD in a special case of graphs $G$ where $V(G)$ partitions
into $A \uplus B$ such that $G[A]$ is chordal and $B$ is a clique.
Finally, we wrap up the approximation algorithm in Section~\ref{s:apx:final}.

\subsection{Multicut and LP relaxations}\label{s:apx:defs}

An instance of the \multicut problem consists of a directed graph $G$
and a family $\terms \subseteq V(G) \times V(G)$ of \emph{terminal pairs}.
A set $X \subseteq V(G)$ is a feasible solution to a multicut instance
$(G,\terms)$ if, for every $(s,t) \in \terms$, there is no path from 
$s$ to $t$ in $G-X$.
The \multicut problem asks for a feasible solution of minimum cardinality.
That is, we focus on the unweighted node-deletion variant of \multicut,
and note that the set $X$ may contain terminals.

Given a \multicut instance $(G,\terms)$, a \emph{fractional solution} is
a function $\LPsol:V(G) \to [0,+\infty)$ such that for every $(s,t) \in \terms$
and every path $P$ from $s$ to $t$ in $G$, we have
$$\LPsol(P) := \sum_{v \in V(P)} \LPsol(v) \geq 1.$$
The equation above is a linear constraint on values of $\LPsol$
that defines the polytope of fractional solutions. Clearly, a feasible solution
to the \multicut instance $(G,\terms)$ induces a fractional solution that assigns
$1$ to the vertices of the solution. 
Furthermore, the separation oracle of 
the aforementioned polytope is straightforward to implement by a shortest-path algorithm such as Dijkstra's.
Consequently, using the ellipsoid method, in polynomial time
we can find a fractional solution of total value arbitrarily close to minimum possible.

Very similarly we define a fractional solution to \ChVD on an undirected graph $G$:
it is a function $\LPsol:V(G) \to [0,+\infty)$ such that for every hole $H$
in $G$ we have
$$\LPsol(H) := \sum_{v \in V(H)} \LPsol(v) \geq 1.$$
Again, the equation above defines a polytope containing all feasible
solutions to \ChVD on $G$, and there is a simple separation oracle
for this polytope.\footnote{Iterate over all possibilities of three consecutive vertices $v_1,v_2,v_3$ on the hole $H$, and compute the shortest path from $v_1$ to $v_3$
in $G-N[v_2]$ with respect to the distances defined by the current values of $\LPsol$.}

For a fractional solution $\LPsol$ (either to \multicut or \ChVD), we
use the notation $\LPsol(X)$ for a vertex set $X \subseteq V(G)$ for
$\sum_{v \in X} \LPsol(v)$. Furthermore, we use $|\LPsol|$ for $\LPsol(V(G))$
and $\LPsol(H) = \LPsol(V(H))$ for a subgraph $H$ of $G$ (e.g., a path).

\multicut is NP-hard even in undirected graphs with three terminal
pairs~\cite{DahlhausJPSY94}. 
The best known approximation ratio for undirected graphs is $\Oh(\log n)$~\cite{GargVY96}, but for general directed graphs only $\Oh(\sqrt{n})$-approximation
is known~\cite{Gupta03}.
We will need the following result from the work of Gupta~\cite{Gupta03}.
\begin{theorem}[\cite{Gupta03}]\label{thm:gupta}
Given a \multicut instance $(G,\terms)$ and a fractional solution $\LPsol$,
one can compute in polynomial time an integral solution of size $\Oh(|\LPsol|^2)$.
\end{theorem}
Although~\cite{Gupta03} tackles the edge-deletion variant of \multicut,
note that in directed graphs the edge- and node-deletion variants are 
equivalent via standard reductions.

Finally, we will frequently need the following observation
\begin{observation}[{\cite[Theorem 1]{GilbertRE84}}]\label{obs:central-clique}
For every graph $G$, its tree decomposition $(T,\beta)$, and a function
$f:V(G) \to \mathbb{R}_{\geq 0}$,
there exists a bag $t \in V(T)$ such that every connected component $C$
of $G-\beta(t)$ satisfies $f(V(C)) \leq f(V(G))/2$.
In particular, for every chordal graph $G$ 
and a function  $f:V(G) \to \mathbb{R}_{\geq 0}$
there exists a maximal clique $X$ in $G$
such that every connected component $C$ of $G-X$
satisfies $f(V(C)) \leq f(V(G))/2$.
\end{observation}

\subsection{Adding a clique to a chordal graph}\label{s:apx:clique}

In this section we develop an approximation algorithm for the following special
case.

\begin{lemma}\label{lem:apx:clique}
Assume we are given a graph $G$,
a fractional solution $\LPsol$ to \ChVD on $G$,
and a partition $V(G) = A \uplus B$, such 
that $G[A]$ is chordal and $G[B]$ is complete.
Then one can in polynomial time find an integral solution to \ChVD on $G$
of size $\Oh(|\LPsol|^2 \log |\LPsol|)$.
\end{lemma}

As an intermediate step towards Lemma~\ref{lem:apx:clique}, we show the following.
\begin{lemma}\label{lem:apx:clique2}
Let $G$, $\LPsol$, $A$, and $B$ be as in the statement of Lemma~\ref{lem:apx:clique}. Furthermore, let $L$ be a maximal clique of $G[A]$. 
Furthermore, assume that $\LPsol(v) = 0$ for
every $v \in B$ and $\LPsol(v) < 1/10$ for every $v \in A$.
Then one can in polynomial time find a set $X \subseteq V(G)$
of size at most $\Oh(|\LPsol|^2)$ such that
$G-X$ does not contain a hole passing through $L$.
\end{lemma}
\begin{proof}
We fix a clique tree $(T,\beta)$ of $G[A]$, and root it in a node $\root$
with $\beta(\root)=L$.
We define a partial order $\preceq_0$ on $A$ as $v \prec_0 u$ if $\top(v)$
is an ancestor of $\top(u)$, and pick $\preceq$ to be an arbitrary total
order extending $\preceq_0$. 
We obtain an orientation $\Gdown$ of $G[A]$
by orienting the edges of $G$ so that $\preceq$ becomes a topological
ordering of $\Gdown$ (i.e., we orient $uv \in E(G)$ as $(u,v)$ if $u \preceq v$).
In $\Gdown$, we define a distance measure $\dist_{\Gdown,\LPsol}$ as distances
with respect to the cost function $\LPsol$; that is, $\dist_{\Gdown,\LPsol}(u,v)$
is the minimum value of $\LPsol(P)$ over all $uv$-paths $P$ in $\Gdown$.
Finally, we define
$$\terms = \{(u,v) \in A \times A : \dist_{\Gdown,\LPsol}(u,v) \geq 1/10\}.$$

Clearly, $10\LPsol$ is a fractional solution to the \multicut instance $(\Gdown,\terms)$.
By applying the algorithm of Theorem~\ref{thm:gupta}, we obtain
a set $X \subseteq V(G)$ of size $\Oh(|\LPsol|^2)$ that is a solution
to \multicut on $(\Gdown,\terms)$. To finish the proof of the lemma,
it suffices to show that $X$ is also good for our purposes, that is,
$G-X$ does not contain any hole passing through $L$.

Assume the contrary, and let $H$ be any such hole in $G-X$.
Since $G[A]$ is chordal, $H$ contains at least one vertex of $B$.
Since both $L$ and $B$ are cliques in $G$, $H$ contains at most two vertices
of each of these cliques.
Consequently, there exists a unique decomposition of $H$ into two paths 
$P^1,P^2$, such that every path $P^i$ has an endpoint $x^i \in L$, and endpoint
$y^i \in B$, and all internal vertices in $A \setminus L$. 
We have $V(P^1) \cup V(P^2) = V(H)$.
Note that it is possible that $x^1=x^2$ or $y^1=y^2$.

Let $u,v$ be two consecutive internal vertices on $P^i$, such that
$u$ is closer to $x^i$ on $P^i$ than $v$. 
Observe that $\top(u)$ is an ancestor of $\top(v)$, as otherwise 
$v$ is adjacent in~$G$ to the predecessor of $u$ on $P^i$, a contradiction
to the assumption that $H$ is a hole. Consequently, $P^i-\{x^i,y^i\}$ 
is a directed path in $\Gdown$.

Let $u^i_j$ be the $j$-th internal vertex on $P^i$ starting from $x^i$,
that is, the consecutive vertices on $P^i$ are $x^i,u^i_1,u^i_2,\ldots,u^i_{\ell^i}, y^i$, where $\ell^i = |P^i|-1$.

%\begin{claim}
%Here is an alternative way to finish the proof from this point on; I believe it is simpler. I think the key difference to the existing approach is that you don't have to care whether the two ``replacement paths'' touch each other or not. It is enough that they have small weight and form a connection between the predecessor and successor of~$x^1$ that avoids~$N(x^1)$. \todo{Do you think this is a worthwhile simplification? It seems the requirement that LP assigns 0 to everything in the clique $B$ is not needed here.}
%\end{claim}
%\begin{claimproof}
Since~$X$ is a multicut for terminals~$\terms$ and~$P^i - \{x^i, y^i\}$ is a directed path in~$\Gdown - X$ for~$i \in \{1,2\}$, for any~$1 \leq a < b \leq \ell^i$ the pair~$(u^i_a, u^i_b)$ is not a terminal pair. By definition of~$\terms$, this implies that there is a path~$P$ in~$\Gdown$ from~$u^i_a$ to~$u^i_b$ with~$\LPsol(P) < 1/10$. We define two alternative paths~$\hat{P}^i$ to derive a contradiction. If~$P^i$ has at most two internal vertices, then~$\hat{P}^i := P^i$. Otherwise, define~$\hat{P}^i$ as~$(x^i, u^i_1, Q^i, y^i)$ where~$Q^i$ is a path in~$\Gdown$ from~$u^i_2$ to~$u^i_{\ell^i}$ with~$\LPsol(Q^i) < 1/10$. No vertex on~$Q^i$ is adjacent to~$x^1$ in~$G$: since~$u^i_2$ is not adjacent to~$x^1$ on the hole~$H$, we know~$u^i_2$ and~$x^1$ do not occur in a common bag. This implies that~$x^1$, which is in the root bag, does not belong to~$\top(u^i_2)$ and does not occur in the subtree rooted there; all vertices that can be reached from~$u^i_2$ in~$\Gdown$ have root bags that are descendants of~$\top(u^i_2)$. 

The concatenation of~$\hat{P}^1$ with~$\hat{P}^2$ forms a (possibly nonsimple) walk in~$G$ that visits~$x^1$ exactly once. Consider the walk~$W$ obtained by removing~$x^1$ from this cycle. It connects the successor of~$x^1$ on~$H$ to the predecessor of~$x^1$ on~$H$. All internal vertices of the walk are nonadjacent to~$x^1$, which follows from our claims on~$Q^i$ and the fact that~$H$ is a hole. Hence we can shortcut~$W$ to an induced path~$P$ connecting the predecessor and successor of~$x^i$ on~$H$, to obtain a hole~$H'$ in~$G$. Since~$\LPsol(v) < 1/10$ for each vertex~$v \in G$, and~$\LPsol(Q^i) < 1/10$, it follows that~$\LPsol(\hat{P^i}) < 3/10+1/10=4/10$. Hence~$\LPsol(H') < 8/10 < 1$, a contradiction to the assumption that $\LPsol$ is a fractional solution to \ChVD on $G$.

\end{proof}

\begin{proof}[Proof of Lemma~\ref{lem:apx:clique}]
First, we take $X_0 = \{v \in V(G) : \LPsol(v) \geq 1/20\}$, put $X_0$
into the solution we are constructing, and delete $X_0$ from the graph.
Clearly, $|X_0| \leq 20|\LPsol|$; from this point, we can assume that
$\LPsol(v) < 1/20$ for every $v \in V(G)$.

Second, we modify $\LPsol(v)$ as follows: we put $\LPsol_2(v) := 2\LPsol(v)$
for every $v \in A$ and $\LPsol_2(v) := 0$ for every $v \in B$.
Since $\LPsol(v) < 1/20$ for every $v \in V(G)$, the length of the shortest
hole in $G$ is at least $21$, while at most two vertices of such a hole can lie
in the clique $G[B]$. Consequently, $\LPsol_2$ is a fractional solution
to \ChVD on $G$ with the additional properties that $\LPsol_2(v) = 0$
for every $v \in B$ and $\LPsol_2(v) < 1/10$ for every $v \in A$.

We now perform the following iterative procedure.
Let $\mathcal{C}$ be the set of connected components of $G[A]$.
We pick a component $C \in \mathcal{C}$ with largest value $\LPsol_2(C)$.
If $\LPsol_2(C) < 1$, we infer that $G[C \cup B]$ is chordal; by the choice 
of $C$, this implies that the entire graph $G$ is chordal, and we are done.
Otherwise, 
let $L$ be a maximal clique of $G[C]$ such that $\LPsol_2(D) \leq \LPsol_2(C)/2$
for every connected component of $G[C]-L$ (Observation~\ref{obs:central-clique}).
We apply the algorithm of Lemma~\ref{lem:apx:clique2} to
$G[C \cup B]$ with partition $C \uplus B$, the clique $L$, and 
a fractional solution $\LPsol_2|_C$, obtaining a solution $X_C$
of size $\Oh(\LPsol_2(C)^2)$. We put $X_C$ into the solution we are constructing,
and delete it from the graph.
At this point, no hole in $G$ passes through $L$; we delete $L$ from the graph as well.

Consider a vertex $v \in C \setminus (L \cup X_C)$ after this step.
Let $D$ be a connected component of $G[A]$ containing $v$. Since we deleted
the entire clique $L$, we have $\LPsol_2(D) \leq \LPsol_2(C)/2$. 
Consequently, every vertex $v$ can be contained in a component $C$ chosen
by the algorithm at most $\lceil \log (1+|\LPsol_2|) \rceil$ times.
Whenever an algorithm considers a component $C$, we charge 
a cost of $\Oh(\LPsol_2(v) |\LPsol_2|)$ to every vertex $v \in C$.
Since the set $X_C$ is of size $\Oh(\LPsol_2(C)^2) \leq \Oh(\LPsol_2(C) |\LPsol_2|)$,the total size of all sets $X_C$ in the course of the algorithm is
bounded by the total charge. On the other hand, by the previous argument,
every vertex $v$ is charged a value of $\Oh(\LPsol_2(v) |\LPsol_2| \log |\LPsol_2|)$.
Consequently, the total size of all sets $X_C$ is $\Oh(|\LPsol_2|^2 \log |\LPsol_2|)$.
%\todo{True, but a little brief. How about relating these bounds to the recursion tree of the algorithm: its depth is logarithmic by the first claim, and for every level of the tree the sum of the solution sizes obtained there is $\Oh(|\LPsol_2|^2)$ by the second claim?}
\end{proof}

\subsection{The approximation algorithm}\label{s:apx:final}

Let $(G,k)$ be a \ChVD instance.
Our main goal in this section is to show that, after deleting a small
number of vertices from $G$, we can decompose $V(G)$ into a chordal graph
and a small number of cliques. Then, we can use Lemma~\ref{lem:apx:clique}
to add these cliques to the chordal graph one by one.

We start with a bit of preprocessing.
As discussed before, it is straightforward to implement a separation
oracle for the LP relaxation of \ChVD. Using the ellipsoid method,
we can compute a fractional solution $\LPsol$ to \ChVD on $G$ of cost
arbitrarily close to the optimum. For our purposes, a $2$-approximation $\LPsol$
suffices. If $|\LPsol| > 2k$, we conclude that $(G,k)$ is a no-instance.
Otherwise, we proceed with $\LPsol$ further.
We greedily take into the constructed solution and delete from $G$
any vertex $v$ with
$\LPsol(v) \geq 1/4$; note that there are at most $4|\LPsol| \leq 8k$
such vertices. 
Thus, from this point we can assume that $\LPsol(v) < 1/4$ for every $v \in V(G)$;
in particular, $G$ is $C_4$-free.

Our main partitioning tool is an $\Oh(\sqrt{\log \textrm{opt}})$-approximation algorithm for balanced vertex cut
by Feige, Hajiaghayi, and Lee~\cite{FeigeHL08}.
Given an $n$-vertex graph $G$, a \emph{balanced vertex cut} is a set $X \subseteq V(G)$ such that every connected component of $G-X$ has at most $2n/3$ vertices.
We use the approximation algorithm of~\cite{FeigeHL08} in
the following lemma.
\begin{lemma}\label{lem:apx:cut}
Given a \ChVD instance $(G,k)$, we can in polynomial time find
a set $Z \subseteq V(G)$ and a set $K \subseteq Z$ such that every connected
component of $G-Z$ has size at most $3n/4$, 
$G[K]$ is complete, and $|Z \setminus K| = \Oh(k \sqrt{\log k})$, or correctly
conclude that $(G,k)$ is a no-instance to \ChVD.
\end{lemma}
\begin{proof}
Since $G$ is $C_4$-free, $G$ has at most $n^2$ maximal cliques~\cite[Proposition 2]{Farber89a}.
Furthermore, we can enumerate all maximal cliques
in $G$ with polynomial delay~\cite{TsukiyamaIAS77}.
If there exists a maximal clique $K$ of size at least $n/4$, we can simply return
$Z=K$, as $|V(G) \setminus K| \leq 3n/4$.
Otherwise, for every maximal clique $K$ of $G$,
we invoke the approximation algorithm for balanced vertex cut of~\cite{FeigeHL08}
on the graph $G-K$.
If for some clique $K$ the algorithm returns a cut $Y$ of size $\Oh(k \sqrt{\log k})$, we output $Z := K \cup Y$ and $K$. Clearly, this output satisfies the desired properties.

It remains to argue that if $(G,k)$ is a yes-instance, then the algorithm finds
a desired cut.
Let $X$ be a solution to \ChVD on $G$ of size at most $k$.
By Observation~\ref{obs:central-clique},
the chordal graph $G-X$ contains a maximal clique $K_0$
such that every connected component of $G-(X \cup K_0)$ has at most $(n - |X|) / 2 \leq n/2$
vertices.
Consider the iteration of the algorithm for a maximal clique $K \supseteq K_0$ in~$G$.
%\todo{Maximality implies~$K = K_0$.} MP: I disagree. There are two different maximalities here: in G and in G-X. Bart: You're right. I added ``in G'' to the previous sentence to prevent this confusion for future readers.
Note that $X$ is a balanced cut of $G-K$ of size at most $k$:
every connected component of $G-(X \cup K)$ is of size at most
$$\frac{n}{2} = \frac{2}{3} \cdot \frac{3}{4}n \leq \frac{2}{3} |V(G) \setminus K|.$$
Consequently, for this choice of $K$, the algorithm of~\cite{FeigeHL08} finds
a balanced vertex cut of size $\Oh(k \sqrt{\log k})$. 
This finishes the proof of the lemma.
\end{proof}

By iteratively applying Lemma~\ref{lem:apx:cut}, we obtain the following.
\begin{lemma}\label{lem:apx:cut2}
Given a \ChVD instance $(G,k)$, we can in polynomial time either
correctly conclude that it is a no-instance, or find a partition
$V(G) = A_0 \uplus K_1 \uplus K_2 \uplus \ldots \uplus K_\ell \uplus X_0$
such that $G[A_0]$ is chordal, $G[K_i]$ is complete for every $1 \leq i \leq \ell$,
$\ell = \Oh(k \log n)$, and
$|X| = \Oh(\ell k \sqrt{\log k}) = \Oh(k^2 \sqrt{\log k} \log n)$.
\end{lemma}
\begin{proof}
Consider the following iterative procedure. We start with $H := G$, $\ell := 0$,
and $X_0 := \emptyset$.
As long as $H$ contains a connected component $C$ that is not chordal,
we apply the algorithm of Lemma~\ref{lem:apx:cut} to a \ChVD instance $(C,k)$. 
In case the algorithm returns that $(C,k)$ is a no-instance,
we can return the same answer for
$(G,k)$, as we will maintain the invariant that $H$ is an induced subgraph of $G$.
Otherwise, if a pair $(Z,K)$ is returned, we take $K$ as a next clique $K_i$,
insert $Z \setminus K$ into $X_0$, and delete $Z$ from $H$.
The algorithm terminates when $H$ is a chordal graph; we then put $A_0 := V(H)$.
Clearly, the algorithm outputs a partition
$V(G) = A_0 \uplus K_1 \uplus K_2 \uplus \ldots \uplus K_\ell \uplus X_0$ as desired,
and $|X| = \Oh(\ell k \sqrt{\log k})$. 
It remains to argue that if $(G,k)$ is a yes-instance, then the algorithm always
terminates after $\Oh(k \log n)$ steps, that is, $\ell = \Oh(k \log n)$.

Consider a solution $X$ to \ChVD on $(G,k)$. Whenever the algorithm
applies Lemma~\ref{lem:apx:cut} to a component $C$, we have $X \cap V(C) \neq \emptyset$, since $C$ is not chordal. We charge this step of the algorithm
to an arbitrarily chosen vertex $v \in X \cap V(C)$.
Note that after the algorithm finds $Z \subseteq V(C)$, we have either $v \in Z$ or $v$ lies in a connected component of $C-Z$ of size at most $2|V(C)|/3$. Consequently,
every vertex $v \in X$ is charged at most $\log_{3/2} n$ times, and in a yes-instance
the algorithm always terminates after $k \log_{3/2} n$ steps. (That is,
if it runs longer than expected, we terminate the algorithm and claim that $(G,k)$ is a no-instance.)
\end{proof}

We can now conclude the proof of Theorem~\ref{thm:apx}.
Let $(G,k)$ be a \ChVD instance.
First, if $\log n > k \log k$, the exact
FPT algorithm of Cao and Marx~\cite{CaoM15} with runtime~$2^{\Oh(k \log k)} \cdot n^{\Oh(1)}$ actually runs in polynomial time in~$n$, and we
can just solve the instance. 
%\todo{Same argument shows that if~$\log n > k^\epsilon$ for any positive~$\epsilon > 0$, then FPT alg runs in poly-time. This can improve solution size to~$k^{3 + \epsilon} polylog(k)$. Worth mentioning?} MP: I disagree. I think you get quasipolynomial time then. Bart: Yes, you're right.
Otherwise, we apply the algorithm of Lemma~\ref{lem:apx:cut2},
either concluding that $(G,k)$ is a no-instance, or finding the partition
$V(G) = A_0 \uplus K_1 \uplus K_2 \uplus \ldots \uplus K_\ell \uplus X_0$
with $\ell = \Oh(k \log n) = \Oh(k^2 \log k)$ and $|X_0| = \Oh(k^2 \sqrt{\log k} \log n) = \Oh(k^3 \log^{3/2} k)$.
Then, for $i=1,2,\ldots,\ell$, we construct sets $A_i$ and $X_i$ as follows.
We apply the algorithm of Lemma~\ref{lem:apx:clique}
to a graph $G[A_{i-1} \cup K_i]$ with the clique $K_i$ and the fractional solution
$\LPsol$ restricted to $A_{i-1} \cup K_i$. The algorithm returns
a solution $X_i$ to $\ChVD$ on $G[A_{i-1} \cup K_i]$ of size $\Oh(|\LPsol|^2 \log |\LPsol|) = \Oh(k^2 \log k)$; we put $A_i := (A_{i-1} \cup K_i) \setminus X_i$.

Note that every $G[A_i]$ is chordal, and 
$$V(G) \setminus A_\ell = X_0 \cup X_1 \cup \ldots \cup X_\ell.$$
Thus, we can return a solution $X := X_0 \cup X_1 \cup \ldots \cup X_\ell$
which is of size
$$|X| = |X_0| + \sum_{i=1}^\ell |X_i| = \Oh(k^3 \log^{3/2} k) + \Oh(k^2 \log k) \cdot \Oh(k^2 \log k) = \Oh(k^4 \log^2 k).$$
This concludes the proof of Theorem~\ref{thm:apx}.

\section{Kernelization algorithm}\label{s:kernelization}
The next theorem summarizes the data reduction procedure for the unannotated problem and shows that the annotations can be simulated by small gadgets.

\begin{theorem*}
\chvdkerneltheorem
\end{theorem*}
\begin{proof}
Given a \ChVD instance~$(G,k)$ and a modulator~$M_0$, apply Lemma~\ref{lem:annotate}. If it reports that~$(G,k)$ is a no-instance, output a constant-size no-instance of \ChVD. Otherwise, we obtain an equivalent instance~$(G',k',M,E^h)$ of \AChVD with~$k' \leq k$ and~$|M| \in \Oh(k \cdot |M_0|)$. Apply Theorem~\ref{thm:AChVD-kernel} to~$(G',k',M,E^h)$ to obtain an equivalent instance~$(G'',k',M,F^h)$ of \AChVD with~$\Oh(k^{16}|M|^{29}) \in \Oh(k^{45}|M_0|^{29})$ vertices. To turn this into an instance of the original \ChVD problem, we simulate the annotations. 

By definition of the annotated problem, we have~$xy \in E(G''[M])$ for each pair~$\{x,y\} \in F^h$ and therefore~$x,y \in M$. For each pair~$\{x,y\} \in F^h$, we simulate the annotation by adding two new vertices~$x',y'$ and the edges~$xx', x'y'$, and~$y'y$ to~$G''$, thereby creating a~$C_4$. As all holes passing through~$x'$ or~$y'$ also pass through~$x$ and~$y$, there is always a minimum chordal modulator in the resulting graph that does not contain any of the newly added vertices. Moreover, since we introduce a~$C_4$ containing~$x$ and~$y$ for each annotated pair~$\{x,y\} \in F^h$, such solutions hit at least one vertex of each annotated pair. Letting~$G^*$ denote the graph resulting from adding such $C_4$'s for all annotated pairs, we therefore establish that if~$G^*$ has a chordal modulator of size~$k'$, then~$G''$ has a solution of size~$k'$ that hits all annotated pairs. Conversely, any solution~$X$ to the annotated problem on~$G''$ is a chordal deletion set in~$G^*$, since in~$G^* - X$ the newly added vertices form a pendant path and are therefore not part of any hole. Consequently, the \ChVD instance~$(G^*, k')$ is equivalent to the annotated instance and therefore to the original input~$(G,k)$. Since the definition of \AChVD ensures that we only annotate pairs of vertices from~$M$, we add at most~$|M|^2$ vertices. Hence the asymptotic size bound for~$G^*$ is dominated by the size of~$G''$. The pair~$(G^*, k')$ is given as the output of the procedure.
\end{proof}

Combining the previous theorem with our approximation algorithm in Theorem~\ref{thm:apx}, we obtain a polynomial kernel for \ChVD. The size bound follows from plugging in~$\Oh(k^4 \log^2 k)$ for the modulator size.

\begin{corollary} \label{cor:kernel}
\ChVDlong has a kernel with~$\Oh(k^{161} \log^{58} k)$ vertices.
\end{corollary}

\section{Conclusion} \label{s:conclusion}
We presented a polynomial kernel for \ChVDlong based on new graph-theoretic insights and a $\poly(\opt)$-approximation algorithm. Our work raises several questions for future research. The most obvious questions concern the kernel size and the approximation factor; both have room for significant improvements. Since \ChVD is a vertex-deletion problem into an efficiently recognizable nontrivial and hereditary graph class, by the results of Dell and van Melkebeek~\cite{DellM14} there is no kernel for \ChVD that can be encoded in~$\Oh(k^{2-\epsilon})$ bits unless \containment. The bound of~$\Ohtilde(k^{161})$ vertices given by Corollary~\ref{cor:kernel} vertices is very far from this, and we expect the optimal kernel to have a size bound with degree at most ten. As Theorem~\ref{thm:chvd:kernel} shows, a better approximation ratio will directly translate into a better kernel size.

From a purely graph-theoretical perspective, one could investigate whether Erd\H{o}s-P\'{o}sa type packing/covering duality holds for holes in general graphs, rather than for holes in nearly-chordal graphs as treated in Lemma~\ref{lemma:algorithm:packing:vs:covering}. Is there a function~$f \colon \mathbb{N} \to \mathbb{N}$ such that for any~$k$, an arbitrary graph~$G$ either has at least~$k$ vertex-disjoint holes, or has a set of~$f(k)$ vertices that intersects all holes? A final open problem is to determine whether the \textsc{Interval Vertex Deletion} problem admits a polynomial kernel.

% Maybe: mention that, even though the kernel we get is somewhat large because we need a large approximation, using co-nondeterministic extensions of the kernelization lowerbound framework you can show that no kernel lower bound for \ChVD parameterized by~$k$ can be proven in that framework that is higher than the bound we obtain for \ChVD parameterized by the size of a given modulator.

\bibliographystyle{abbrvurl}
\bibliography{references}

\begin{thebibliography}{10}

\bibitem{AgarwalCMM05}
A.~Agarwal, M.~Charikar, K.~Makarychev, and Y.~Makarychev.
\newblock ${O}(\sqrt{\log n})$ approximation algorithms for min {U}n{C}ut, min
  2{CNF} deletion, and directed cut problems.
\newblock In {\em Proc. 37th STOC}, pages 573--581, 2005.
\newblock \href {http://dx.doi.org/10.1145/1060590.1060675}
  {\path{doi:10.1145/1060590.1060675}}.

\bibitem{BafnaBF99}
V.~Bafna, P.~Berman, and T.~Fujito.
\newblock A 2-approximation algorithm for the undirected feedback vertex set
  problem.
\newblock {\em SIAM Journal on Discrete Mathematics}, 12(3):289--297, 1999.
\newblock \href {http://dx.doi.org/10.1137/S0895480196305124}
  {\path{doi:10.1137/S0895480196305124}}.

\bibitem{Bar-NoyBFNS01}
A.~Bar{-}Noy, R.~Bar{-}Yehuda, A.~Freund, J.~Naor, and B.~Schieber.
\newblock A unified approach to approximating resource allocation and
  scheduling.
\newblock {\em J. {ACM}}, 48(5):1069--1090, 2001.
\newblock \href {http://dx.doi.org/10.1145/502102.502107}
  {\path{doi:10.1145/502102.502107}}.

\bibitem{BauerHS90}
D.~Bauer, S.~L. Hakimi, and E.~F. Schmeichel.
\newblock Recognizing tough graphs is {NP}-hard.
\newblock {\em Discrete Applied Mathematics}, 28(3):191--195, 1990.
\newblock \href {http://dx.doi.org/10.1016/0166-218X(90)90001-S}
  {\path{doi:10.1016/0166-218X(90)90001-S}}.

\bibitem{BeckerG96}
A.~Becker and D.~Geiger.
\newblock Optimization of {P}earl's method of conditioning and greedy-like
  approximation algorithms for the vertex feedback set problem.
\newblock {\em Artificial Intelligence}, 83(1):167--188, 1996.
\newblock \href {http://dx.doi.org/10.1016/0004-3702(95)00004-6}
  {\path{doi:10.1016/0004-3702(95)00004-6}}.

\bibitem{BliznetsFPP15}
I.~Bliznets, F.~V. Fomin, M.~Pilipczuk, and M.~Pilipczuk.
\newblock A subexponential parameterized algorithm for proper interval
  completion.
\newblock {\em {SIAM} J. Discrete Math.}, 29(4):1961--1987, 2015.
\newblock \href {http://dx.doi.org/10.1137/140988565}
  {\path{doi:10.1137/140988565}}.

\bibitem{BliznetsFPP16}
I.~Bliznets, F.~V. Fomin, M.~Pilipczuk, and M.~Pilipczuk.
\newblock Subexponential parameterized algorithm for interval completion.
\newblock In R.~Krauthgamer, editor, {\em Proceedings of the Twenty-Seventh
  Annual {ACM-SIAM} Symposium on Discrete Algorithms, {SODA} 2016, Arlington,
  VA, USA, January 10-12, 2016}, pages 1116--1131. {SIAM}, 2016.
\newblock \href {http://dx.doi.org/10.1137/1.9781611974331.ch78}
  {\path{doi:10.1137/1.9781611974331.ch78}}.

\bibitem{Bodlaender93}
H.~L. Bodlaender.
\newblock On linear time minor tests with depth-first search.
\newblock {\em J. Algorithms}, 14(1):1--23, 1993.
\newblock \href {http://dx.doi.org/10.1006/jagm.1993.1001}
  {\path{doi:10.1006/jagm.1993.1001}}.

\bibitem{MetaKernelization09}
H.~L. Bodlaender, F.~V. Fomin, D.~Lokshtanov, E.~Penninkx, S.~Saurabh, and
  D.~M. Thilikos.
\newblock ({M}eta) {K}ernelization.
\newblock In {\em 50th Annual {IEEE} Symposium on Foundations of Computer
  Science, {FOCS} 2009, October 25-27, 2009, Atlanta, Georgia, {USA}}, pages
  629--638. {IEEE} Computer Society, 2009.
\newblock \href {http://dx.doi.org/10.1109/FOCS.2009.46}
  {\path{doi:10.1109/FOCS.2009.46}}.

\bibitem{BodlaenderD10}
H.~L. Bodlaender and T.~C. van Dijk.
\newblock A cubic kernel for feedback vertex set and loop cutset.
\newblock {\em Theory Comput. Syst.}, 46(3):566--597, 2010.

\bibitem{BrandstadtLS99}
A.~Brandst\"{a}dt, V.~B. Le, and J.~P. Spinrad.
\newblock {\em Graph classes: a survey}.
\newblock Society for Industrial and Applied Mathematics, Philadelphia, PA,
  USA, 1999.

\bibitem{Buneman74}
P.~Buneman.
\newblock A characterisation of rigid circuit graphs.
\newblock {\em Discrete Mathematics}, 9(3):205 -- 212, 1974.
\newblock \href
  {http://dx.doi.org/http://dx.doi.org/10.1016/0012-365X(74)90002-8}
  {\path{doi:http://dx.doi.org/10.1016/0012-365X(74)90002-8}}.

\bibitem{BurrageEFLMR06}
K.~Burrage, V.~Estivill{-}Castro, M.~R. Fellows, M.~A. Langston, S.~Mac, and
  F.~A. Rosamond.
\newblock The undirected feedback vertex set problem has a poly(\emph{k})
  kernel.
\newblock In H.~L. Bodlaender and M.~A. Langston, editors, {\em Parameterized
  and Exact Computation, Second International Workshop, {IWPEC} 2006,
  Z{\"{u}}rich, Switzerland, September 13-15, 2006, Proceedings}, volume 4169
  of {\em Lecture Notes in Computer Science}, pages 192--202. Springer, 2006.
\newblock \href {http://dx.doi.org/10.1007/11847250_18}
  {\path{doi:10.1007/11847250_18}}.

\bibitem{BussG93}
J.~F. Buss and J.~Goldsmith.
\newblock Nondeterminism within {P}.
\newblock {\em SIAM J. Comput.}, 22(3):560--572, 1993.
\newblock \href {http://dx.doi.org/10.1137/0222038}
  {\path{doi:10.1137/0222038}}.

\bibitem{Cai03a}
L.~Cai.
\newblock Parameterized complexity of vertex colouring.
\newblock {\em Discrete Appl. Math.}, 127(3):415--429, 2003.
\newblock \href {http://dx.doi.org/10.1016/S0166-218X(02)00242-1}
  {\path{doi:10.1016/S0166-218X(02)00242-1}}.

\bibitem{Cao16}
Y.~Cao.
\newblock Linear recognition of almost interval graphs.
\newblock In R.~Krauthgamer, editor, {\em Proceedings of the Twenty-Seventh
  Annual {ACM-SIAM} Symposium on Discrete Algorithms, {SODA} 2016, Arlington,
  VA, USA, January 10-12, 2016}, pages 1096--1115. {SIAM}, 2016.
\newblock \href {http://dx.doi.org/10.1137/1.9781611974331.ch77}
  {\path{doi:10.1137/1.9781611974331.ch77}}.

\bibitem{CaoM14b}
Y.~Cao and D.~Marx.
\newblock Interval deletion is fixed-parameter tractable.
\newblock In C.~Chekuri, editor, {\em Proceedings of the Twenty-Fifth Annual
  {ACM-SIAM} Symposium on Discrete Algorithms, {SODA} 2014, Portland, Oregon,
  USA, January 5-7, 2014}, pages 122--141. {SIAM}, 2014.
\newblock \href {http://dx.doi.org/10.1137/1.9781611973402.9}
  {\path{doi:10.1137/1.9781611973402.9}}.

\bibitem{CaoM15}
Y.~Cao and D.~Marx.
\newblock Chordal editing is fixed-parameter tractable.
\newblock {\em Algorithmica}, pages 1--20, 2015.
\newblock \href {http://dx.doi.org/10.1007/s00453-015-0014-x}
  {\path{doi:10.1007/s00453-015-0014-x}}.

\bibitem{worker2013-opl}
M.~Cygan, L.~Kowalik, and M.~Pilipczuk.
\newblock Open problems from workshop on kernels, 2013.
\newblock URL: \url{http://worker2013.mimuw.edu.pl/slides/worker-opl.pdf}.

\bibitem{DahlhausJPSY94}
E.~Dahlhaus, D.~S. Johnson, C.~H. Papadimitriou, P.~D. Seymour, and
  M.~Yannakakis.
\newblock The complexity of multiterminal cuts.
\newblock {\em {SIAM} J. Comput.}, 23(4):864--894, 1994.
\newblock \href {http://dx.doi.org/10.1137/S0097539792225297}
  {\path{doi:10.1137/S0097539792225297}}.

\bibitem{DellM14}
H.~Dell and D.~van Melkebeek.
\newblock Satisfiability allows no nontrivial sparsification unless the
  polynomial-time hierarchy collapses.
\newblock {\em J. {ACM}}, 61(4):23:1--23:27, 2014.
\newblock \href {http://dx.doi.org/10.1145/2629620}
  {\path{doi:10.1145/2629620}}.

\bibitem{DrangeFPV15}
P.~G. Drange, F.~V. Fomin, M.~Pilipczuk, and Y.~Villanger.
\newblock Exploring the subexponential complexity of completion problems.
\newblock {\em {TOCT}}, 7(4):14, 2015.
\newblock \href {http://dx.doi.org/10.1145/2799640}
  {\path{doi:10.1145/2799640}}.

\bibitem{ErdosP65}
P.~Erd\H{o}s and L.~P\'{o}sa.
\newblock On independent circuits contained in a graph.
\newblock {\em Canad. J. Math.}, 17:347--352, 1965.
\newblock \href {http://dx.doi.org/10.4153/CJM-1965-035-8}
  {\path{doi:10.4153/CJM-1965-035-8}}.

\bibitem{Farber89a}
M.~Farber.
\newblock On diameters and radii of bridged graphs.
\newblock {\em Discrete Mathematics}, 73(3):249--260, 1989.
\newblock \href {http://dx.doi.org/10.1016/0012-365X(89)90268-9}
  {\path{doi:10.1016/0012-365X(89)90268-9}}.

\bibitem{FeigeHL08}
U.~Feige, M.~Hajiaghayi, and J.~R. Lee.
\newblock Improved approximation algorithms for minimum weight vertex
  separators.
\newblock {\em {SIAM} J. Comput.}, 38(2):629--657, 2008.
\newblock \href {http://dx.doi.org/10.1137/05064299X}
  {\path{doi:10.1137/05064299X}}.

\bibitem{FlumG06}
J.~Flum and M.~Grohe.
\newblock {\em Parameterized Complexity Theory}.
\newblock Springer-Verlag, 2006.

\bibitem{FominJP12}
F.~V. Fomin, B.~M.~P. Jansen, and M.~Pilipczuk.
\newblock Preprocessing subgraph and minor problems: When does a small vertex
  cover help?
\newblock In {\em Proc. 7th IPEC}, pages 97--108, 2012.
\newblock \href {http://dx.doi.org/10.1007/978-3-642-33293-7_11}
  {\path{doi:10.1007/978-3-642-33293-7_11}}.

\bibitem{FominLMS12}
F.~V. Fomin, D.~Lokshtanov, N.~Misra, and S.~Saurabh.
\newblock Planar $\mathcal{F}$-{D}eletion: Approximation, kernelization and
  optimal {FPT} algorithms.
\newblock In {\em Proc. 53rd FOCS}, pages 470--479, 2012.
\newblock \href {http://dx.doi.org/10.1109/FOCS.2012.62}
  {\path{doi:10.1109/FOCS.2012.62}}.

\bibitem{FominSV13}
F.~V. Fomin, S.~Saurabh, and Y.~Villanger.
\newblock A polynomial kernel for proper interval vertex deletion.
\newblock {\em {SIAM} J. Discrete Math.}, 27(4):1964--1976, 2013.
\newblock \href {http://dx.doi.org/10.1137/12089051X}
  {\path{doi:10.1137/12089051X}}.

\bibitem{GajarskyHOORRVS13}
J.~Gajarsk{\'{y}}, P.~Hlinen{\'{y}}, J.~Obdrz{\'{a}}lek, S.~Ordyniak, F.~Reidl,
  P.~Rossmanith, F.~S. Villaamil, and S.~Sikdar.
\newblock Kernelization using structural parameters on sparse graph classes.
\newblock In H.~L. Bodlaender and G.~F. Italiano, editors, {\em Algorithms -
  {ESA} 2013 - 21st Annual European Symposium, Sophia Antipolis, France,
  September 2-4, 2013. Proceedings}, volume 8125 of {\em Lecture Notes in
  Computer Science}, pages 529--540. Springer, 2013.
\newblock \href {http://dx.doi.org/10.1007/978-3-642-40450-4_45}
  {\path{doi:10.1007/978-3-642-40450-4_45}}.

\bibitem{GargVY94}
N.~Garg, V.~V. Vazirani, and M.~Yannakakis.
\newblock Multiway cuts in directed and node weighted graphs.
\newblock In {\em Proc. 21st ICALP}, pages 487--498, 1994.

\bibitem{GargVY96}
N.~Garg, V.~V. Vazirani, and M.~Yannakakis.
\newblock Approximate max-flow min-(multi)cut theorems and their applications.
\newblock {\em SIAM Journal on Computing}, 25(2):235--251, 1996.
\newblock \href {http://dx.doi.org/10.1137/S0097539793243016}
  {\path{doi:10.1137/S0097539793243016}}.

\bibitem{Gavril72}
F.~Gavril.
\newblock Algorithms for minimum coloring, maximum clique, minimum covering by
  cliques, and maximum independent set of a chordal graph.
\newblock {\em SIAM Journal on Computing}, 1(2):180--187, 1972.
\newblock \href {http://dx.doi.org/10.1137/0201013}
  {\path{doi:10.1137/0201013}}.

\bibitem{GiannopoulouJLS15}
A.~C. Giannopoulou, B.~M.~P. Jansen, D.~Lokshtanov, and S.~Saurabh.
\newblock Uniform kernelization complexity of hitting forbidden minors.
\newblock In M.~M. Halld{\'{o}}rsson, K.~Iwama, N.~Kobayashi, and B.~Speckmann,
  editors, {\em Automata, Languages, and Programming - 42nd International
  Colloquium, {ICALP} 2015, Kyoto, Japan, July 6-10, 2015, Proceedings, Part
  {I}}, volume 9134 of {\em Lecture Notes in Computer Science}, pages 629--641.
  Springer, 2015.
\newblock \href {http://dx.doi.org/10.1007/978-3-662-47672-7_51}
  {\path{doi:10.1007/978-3-662-47672-7_51}}.

\bibitem{GilbertRE84}
J.~R. Gilbert, D.~J. Rose, and A.~Edenbrandt.
\newblock A separator theorem for chordal graphs.
\newblock {\em SIAM Journal on Algebraic and Discrete Methods}, 5(3):306--313,
  1984.

\bibitem{Gupta03}
A.~Gupta.
\newblock Improved results for directed multicut.
\newblock In {\em Proceedings of the Fourteenth Annual {ACM-SIAM} Symposium on
  Discrete Algorithms, January 12-14, 2003, Baltimore, Maryland, {USA.}}, pages
  454--455. {ACM/SIAM}, 2003.
\newblock URL: \url{http://dl.acm.org/citation.cfm?id=644108.644181}.

\bibitem{HeggernesHJKV13}
P.~Heggernes, P.~van~'t Hof, B.~M.~P. Jansen, S.~Kratsch, and Y.~Villanger.
\newblock Parameterized complexity of vertex deletion into perfect graph
  classes.
\newblock {\em Theor. Comput. Sci.}, 511:172--180, 2013.
\newblock \href {http://dx.doi.org/10.1016/j.tcs.2012.03.013}
  {\path{doi:10.1016/j.tcs.2012.03.013}}.

\bibitem{KammerT09}
F.~Kammer and T.~Tholey.
\newblock The k-disjoint paths problem on chordal graphs.
\newblock In {\em Proc. 35th WG}, volume 5911 of {\em Lecture Notes in Computer
  Science}, pages 190--201, 2009.
\newblock \href {http://dx.doi.org/10.1007/978-3-642-11409-0_17}
  {\path{doi:10.1007/978-3-642-11409-0_17}}.

\bibitem{KimLPRRSS13}
E.~J. Kim, A.~Langer, C.~Paul, F.~Reidl, P.~Rossmanith, I.~Sau, and S.~Sikdar.
\newblock Linear kernels and single-exponential algorithms via protrusion
  decompositions.
\newblock In {\em Proc. 40th ICALP}, pages 613--624, 2013.

\bibitem{SurveyKratsch14}
S.~Kratsch.
\newblock Recent developments in kernelization: {A} survey.
\newblock {\em Bulletin of the {EATCS}}, 113, 2014.
\newblock URL: \url{http://eatcs.org/beatcs/index.php/beatcs/article/view/285}.

\bibitem{KratschW12}
S.~Kratsch and M.~Wahlstr{\"o}m.
\newblock Representative sets and irrelevant vertices: New tools for
  kernelization.
\newblock In {\em Proc. 53rd FOCS}, pages 450--459, 2012.
\newblock \href {http://dx.doi.org/10.1109/FOCS.2012.46}
  {\path{doi:10.1109/FOCS.2012.46}}.

\bibitem{KratschW14}
S.~Kratsch and M.~Wahlstr{\"{o}}m.
\newblock Compression via matroids: {A} randomized polynomial kernel for odd
  cycle transversal.
\newblock {\em {ACM} Transactions on Algorithms}, 10(4):20:1--20:15, 2014.
\newblock \href {http://dx.doi.org/10.1145/2635810}
  {\path{doi:10.1145/2635810}}.

\bibitem{LewisY80}
J.~M. Lewis and M.~Yannakakis.
\newblock The node-deletion problem for hereditary properties is {NP}-complete.
\newblock {\em J. Comput. Syst. Sci.}, 20(2):219--230, 1980.
\newblock \href {http://dx.doi.org/10.1016/0022-0000(80)90060-4}
  {\path{doi:10.1016/0022-0000(80)90060-4}}.

\bibitem{SurveyLokshtanov11}
D.~Lokshtanov.
\newblock Kernelization: An overview.
\newblock In O.~Owe, M.~Steffen, and J.~A. Telle, editors, {\em Fundamentals of
  Computation Theory - 18th International Symposium, {FCT} 2011, Oslo, Norway,
  August 22-25, 2011. Proceedings}, volume 6914 of {\em Lecture Notes in
  Computer Science}, pages 39--40. Springer, 2011.
\newblock \href {http://dx.doi.org/10.1007/978-3-642-22953-4_3}
  {\path{doi:10.1007/978-3-642-22953-4_3}}.

\bibitem{SurveyLokshtanovMS12}
D.~Lokshtanov, N.~Misra, and S.~Saurabh.
\newblock Kernelization - preprocessing with a guarantee.
\newblock In H.~L. Bodlaender, R.~Downey, F.~V. Fomin, and D.~Marx, editors,
  {\em The Multivariate Algorithmic Revolution and Beyond - Essays Dedicated to
  Michael R. Fellows on the Occasion of His 60th Birthday}, volume 7370 of {\em
  Lecture Notes in Computer Science}, pages 129--161. Springer, 2012.
\newblock \href {http://dx.doi.org/10.1007/978-3-642-30891-8_10}
  {\path{doi:10.1007/978-3-642-30891-8_10}}.

\bibitem{Marx06d}
D.~Marx.
\newblock Chordal deletion is fixed-parameter tractable.
\newblock In F.~V. Fomin, editor, {\em Graph-Theoretic Concepts in Computer
  Science, 32nd International Workshop, {WG} 2006, Bergen, Norway, June 22-24,
  2006, Revised Papers}, volume 4271 of {\em Lecture Notes in Computer
  Science}, pages 37--48. Springer, 2006.
\newblock \href {http://dx.doi.org/10.1007/11917496_4}
  {\path{doi:10.1007/11917496_4}}.

\bibitem{Marx10}
D.~Marx.
\newblock Chordal deletion is fixed-parameter tractable.
\newblock {\em Algorithmica}, 57(4):747--768, 2010.
\newblock \href {http://dx.doi.org/10.1007/s00453-008-9233-8}
  {\path{doi:10.1007/s00453-008-9233-8}}.

\bibitem{NemhauserT75}
G.~Nemhauser and L.~Trotter.
\newblock {Vertex packings: structural properties and algorithms.}
\newblock {\em Math. Program.}, 8:232--248, 1975.
\newblock \href {http://dx.doi.org/10.1007/BF01580444}
  {\path{doi:10.1007/BF01580444}}.

\bibitem{Rose72}
D.~J. Rose.
\newblock A graph-theoretic study of the numerical solution of sparse positive
  definite systems of linear equations.
\newblock {\em Graph theory and computing}, pages 183--217, 1972.

\bibitem{Thomasse10}
S.~Thomass{\'e}.
\newblock A $4k^2$ kernel for feedback vertex set.
\newblock {\em ACM Trans. Algorithms}, 6(2), 2010.
\newblock \href {http://dx.doi.org/10.1145/1721837.1721848}
  {\path{doi:10.1145/1721837.1721848}}.

\bibitem{TsukiyamaIAS77}
S.~Tsukiyama, M.~Ide, H.~Ariyoshi, and I.~Shirakawa.
\newblock A new algorithm for generating all the maximal independent sets.
\newblock {\em {SIAM} J. Comput.}, 6(3):505--517, 1977.
\newblock \href {http://dx.doi.org/10.1137/0206036}
  {\path{doi:10.1137/0206036}}.

\bibitem{BevernMN12}
R.~van Bevern, H.~Moser, and R.~Niedermeier.
\newblock Approximation and tidying - {A} problem kernel for s-plex cluster
  vertex deletion.
\newblock {\em Algorithmica}, 62(3-4):930--950, 2012.
\newblock \href {http://dx.doi.org/10.1007/s00453-011-9492-7}
  {\path{doi:10.1007/s00453-011-9492-7}}.

\bibitem{ZhangSFCWKB94}
P.~Zhang, E.~A. Schon, S.~G. Fischer, E.~Cayanis, J.~Weiss, S.~Kistler, and
  P.~E. Bourne.
\newblock An algorithm based on graph theory for the assembly of contigs in
  physical mapping of {DNA}.
\newblock {\em Computer Applications in the Biosciences}, 10(3):309--317, 1994.
\newblock \href {http://dx.doi.org/10.1093/bioinformatics/10.3.309}
  {\path{doi:10.1093/bioinformatics/10.3.309}}.

\end{thebibliography}

\end{document}